\documentclass[10pt]{article}
\usepackage[letterpaper, left=1in, top=1in, right=1in, bottom=1in]{geometry}

\usepackage{amsfonts}
\usepackage{amssymb}
\usepackage{url}
\usepackage{booktabs}
\usepackage{graphicx}
\graphicspath{{Figures/}}
\usepackage{amsthm,bm}
\usepackage{pbox}
\usepackage{mathrsfs}
\usepackage{enumitem}
\usepackage{subfigure}
\usepackage{color, colortbl}

\usepackage{chemarrow}

\usepackage{wrapfig}

\definecolor{Gray}{gray}{0.95}
\definecolor{LightCyan}{rgb}{0.88,1,1}
\newcolumntype{g}{>{\columncolor{Gray}}c}

\usepackage[hidelinks]{hyperref}
\hypersetup{
	colorlinks, 
	linkcolor=black, 
	anchorcolor=blue, 
	citecolor=blue
}

\newtheorem{theorem}{Theorem}
\newtheorem{corollary}{Corollary}
\newtheorem{lemma}{Lemma}
\newtheorem{remark}{Remark}
\newtheorem{claim}{Claim}

\newtheorem{definition}{Definition}

\newtheorem{proposition}{Proposition}

\usepackage[many]{tcolorbox}

\tcolorboxenvironment{theorem}{
  breakable,
  colback=gray!20,
  boxrule=0pt,
  boxsep=1pt,
  left=1em,right=1em,top=1em,bottom=1em,
  arc=10pt,
  before skip=\topsep,
  after skip=\topsep,
}

\usepackage{mathtools, cuted}

\usepackage[lined,ruled]{algorithm2e}
\usepackage{algorithmic}

\usepackage{bbm}
\usepackage{dsfont}

\usepackage{accents}
\DeclareRobustCommand{\ubar}[1]{\underaccent{\bar}{#1}}

\begin{document}

\title{Threshold Policies with Tight Guarantees for\\ Online Selection with Convex Costs}

\author{Xiaoqi Tan\thanks{University of Alberta \& Amii. Email: {\tt xiaoqi.tan@ualberta.ca}}
\and Siyuan Yu\thanks{University of Alberta. Email: {\tt syu3@ualberta.ca}} 
\and Raouf Boutaba\thanks{University of Waterloo. Email: {\tt rboutaba@uwaterloo.ca}} 
\and Alberto Leon-Garcia\thanks{University of Toronto. Email: {\tt alberto.leongarcia@utoronto.ca}} 
}

%\date{}

\begin{titlepage}
\maketitle

\begin{abstract}
This paper provides threshold policies with tight guarantees for online selection with convex cost (OSCC). In OSCC, a seller wants to sell some asset to a sequence of buyers with the goal of maximizing her profit. The seller can produce additional units of the asset, but at non-decreasing marginal costs. At each time, a buyer arrives and offers a price. The seller must make an immediate and irrevocable decision in terms of whether to accept the offer and produce/sell one unit of the asset to this buyer. The goal is to develop an online algorithm that selects a subset of buyers to maximize the seller's profit, namely, the total selling revenue minus the total production cost. Our main result is the development of a class of simple threshold policies that are logistically simple and easy to implement, but have provable optimality guarantees among all deterministic algorithms. We also derive a lower bound on competitive ratios of randomized algorithms and prove that the competitive ratio of our threshold policy asymptotically converges to this lower bound when the total production output is sufficiently large.  Our results generalize and unify various online search, pricing, and auction problems, and provide a new perspective on the impact of non-decreasing marginal costs on real-world online resource allocation problems.
\end{abstract}

\end{titlepage}

\section{Introduction}

We consider the following online selection problem:  a seller is selling some asset to a sequence of buyers who arrive one at a time, aiming to maximize the profit. The seller can produce $ k \geq 1 $ units of the asset in total, but at non-decreasing marginal costs, namely, it is increasingly more costly to produce additional units of the asset. At each time $ t = 1, 2, \cdots $, buyer $ t $ arrives and offers a price $ p_t $, and the seller needs to make an immediate and irrevocable decision  in terms of whether to accept the offer and produce/sell one unit of the asset to buyer $ t $. The goal is to develop online algorithms that select a subset of \textit{at most} $ k $ buyers to maximize the seller's profit, namely, the total selling revenue  minus the total production cost. In this paper, we refer to this problem as \textit{online selection with convex costs} (OSCC) since a non-decreasing marginal cost per unit implies that the total production cost is convex w.r.t. the total production output.

The idea of incorporating such convex costs in online selection is primarily motivated by real-world online resource allocation problems  invovling ``diseconomy-of-scale" costs (e.g.,  \cite{network_design_cost_2016, wierman_speed_scaling_2009}), where i) only a subset of incoming requests can be satisfied (due to capacity constraints), and ii) the seller needs to balance the \textit{value-cost tradeoff} by selecting a subset of ``worthy" requests, one at a time, so that  these online decisions turn out to be good choices in hindsight.  For example, in cloud computing \cite{XZhang_2015, Zhang2017}, when incoming jobs have different priorities, how to decide which subset of jobs to admit to maximize allocation efficiency while at the same time to balance the power and cooling costs of computing servers? In electric vehicle (EV) charging \cite{OMD_EV_2011, OKP_EV_BoSun_2020}, how to select a subset of charging requests to balance the aggregate utility of EV owners and the cost of electricity supply? Similar online selection problems have also been studied in the context of admission control in queueing systems \cite{online_admission_control_2009, admission_queue_setup_costs_2020} and various online search, pricing, and auction problems such as \textmd{online time series search} \cite{time_series_search_2001}, \textmd{$ k $-max search} \cite{k_search_2009}, \textmd{one-way trading} \cite{one_way_trading_2019}, \textmd{online conversion} \cite{online_conversion_2021}, and \textmd{online auctions}  \cite{Blum_2011, Huang_2019}. However, despite decades of research, the impact of non-decreasing marginal costs on online selection is not yet fully understood. In fact, it is not clear how to design optimal online selection algorithms with general non-decreasing marginal costs, even for special cases (e.g., linear marginal cost).

We study OSCC under the  worst-case competitive analysis framework, where the sequence of buyers to arrive is unknown and does not necessarily follow any pattern. This means, if we denote the arrival instance by $ \mathcal{I} = \{p_1, p_2, \cdots, p_T\} $, then both the prices  $ \{p_t\}_{\forall t} $ and the total number of buyers $ T $ remain unknown a priori. The performance of an online algorithm is quantified by its \textit{competitive ratio}. Given an arrival instance $ \mathcal{I} $, let us denote by  $ \textsf{OPT}(\mathcal{I}) $ the optimal profit achieved in the offline setting when the information of the arrival instance $ \mathcal{I} $ is known beforehand. Let  $ \textsf{ALG}(\mathcal{I}) $ denote the profit achieved by an online algorithm \textsf{ALG}. Then, the competitive ratio of \textsf{ALG} is defined as follows
%%\vspace{-0.1cm}
\begin{equation}\label{equation_alpha}
	\alpha \triangleq \max_{\mathcal{I}}\ \frac{ \textsf{OPT}(\mathcal{I})}{\mathbb{E}[\textsf{ALG}(\mathcal{I})]},
	%%\vspace{-0.1cm}
\end{equation}
where $ \alpha\geq 1 $ and the closer to 1 the better, and the expectation is taken w.r.t. the possible randomization of the online algorithm.  

\subsection{Overview of Main Results}
To solve OSCC, one of the main challenges we face is that \textit{we must guard against accepting prices that look good early on but cause the rapid growth of marginal costs as more units are produced}. In this paper, we focus on a subclass of algorithms: threshold policies. A threshold policy makes decisions of accepting or rejecting prices based on whether they exceed a predesigned threshold. Meanwhile, we  allow the competitive ratio to be setup-dependent, i.e., the value of $ \alpha $ can be related to the production cost function $ f $, the capacity $ k $, and the \textit{fluctuation ratio} $ \rho $ defined as follows:
\begin{equation}\label{rho_def}
\rho \triangleq p_{\max}/p_{\min}.
\end{equation}
In Eq. \eqref{rho_def}, $ p_{\min} $ (rep. $ p_{\max} $) denotes the minimum (rep. maximum) price offered by the buyers. By definition, $ \rho\in [1,+\infty) $ always holds and a larger $ \rho $ indicates a higher level of price volatility, hence the name ``fluctuation ratio."  

The major contribution of this paper is the design of threshold policies with tight guarantees for various setups of $ f $, $ k $, and $ \rho $. In particular, we derive tight, setup-dependent lower bounds of competitive ratios and prove that these lower bounds cannot be overcome by randomized algorithms when $ k $ is finite. In contrast, when $ k $ is large, a simple threshold policy can achieve a competitive ratio matching these lower bounds. Specifically, our main results are threefold:
\begin{itemize}
	\item \textbf{Optimal deterministic algorithms}. In Theorem \ref{theorem_optimality}, we show the \textit{existence} and \textit{uniqueness} of a threshold-based online selection algorithm which achieves the optimal competitive ratio, denoted by $ \textsf{CR}_f^*(\rho, k) $, of all deterministic algorithms. 

	\item \textbf{Lower bound for randomized algorithms}. In Theorem \ref{theorem_hardness_results_general}, we derive a lower bound, denoted by $ \textsf{CR}_f^{\textsf{lb}}(\rho, k) $, and show that no randomized algorithm can achieve a competitive ratio better than $ \textsf{CR}_f^{\textsf{lb}}(\rho, k) $ when $ k $ is finite. Here the superscript ``\textsf{lb}" means ``\textsf{lower bound}." 
	
	\item \textbf{Asymptotic convergence}. In Theorem \ref{theorem_asymptotic}, we prove that both $ \textsf{CR}_f^*(\rho, k) $  and  $ \textsf{CR}_f^{\textsf{lb}}(\rho,k) $ converge to the same asymptotic lower bound $ \underline{\textsf{CR}}_f(\rho) $ when $ k \rightarrow +\infty$. Thus, our proposed threshold policy is not only optimal among all deterministic algorithms, but also asymptotically optimal among randomized algorithms.   
\end{itemize}

As illustrated in Fig. \ref{Fig_CR_bound}, for any given $ k\in \{1,2, \cdots\} $,  we have $\textsf{CR}_f^*(1,k) =  \textsf{CR}_f^{\textsf{lb}}(1,k) =  \underline{\textsf{CR}}_f(1) = 1 $, indicating the reduction of OSCC to a trivial deterministic selection problem if all buyers offer the same price (i.e., $ \rho = 1 $). In general, $\textsf{CR}_f^*(\rho,k) \geq  \textsf{CR}_f^{\textsf{lb}}(\rho,k) \geq  \underline{\textsf{CR}}_f(\rho) $ holds and they are all increasing functions of $ \rho\in [1,+\infty) $, meaning that a higher price volatility leads to a less competitive performance. Our results show that $\textsf{CR}_f^*(\rho,k) $, $ \textsf{CR}_f^{\textsf{lb}}(\rho,k) $, and $ \underline{\textsf{CR}}_f(\rho) $ do not have closed-form expressions in general -- their values are implicitly given by solving some system of nonlinear equations. Given the generality of $ f $, we argue that this is inevitable. In fact, even in the simplest case when the production cost function $ f $ is linear (e.g.,  $ f(y) = a y $ with $ a\geq 0 $)\footnote{A linear production cost function $ f $ implies constant marginal costs, i.e., producing each unit of the asset costs the same.},  the optimal competitive ratio $\textsf{CR}_{f(y) = a y}^*(\rho,k) $ and the lower bound $ \textsf{CR}_{f(y) = a y}^{\textsf{lb}}(\rho,k) $ still have no close-form expressions, although in this case the asymptotic lower bound $ \underline{\textsf{CR}}_{f(y) = a y}(\rho) $ can be written as
\begin{equation}\label{CR_linear}
\underline{\textsf{CR}}_{f(y) = a y}(\rho) = 1 + \ln \left(\rho(a)\right),
\end{equation}
where $ \rho(a) = \frac{p_{\max} - a}{p_{\min}-a} $ is a shifted value of the fluctuation ratio $ \rho $ (note that here we abuse the notation and use the same symbol to denote  $ \rho $  and $ \rho(a) $). Compared with online selection without costs, our results show that considering non-decreasing marginal costs (or convex production costs) makes the competitive analysis of online selection much harder.

\begin{wrapfigure}{r}{0.35\textwidth}
	\begin{center}
		\includegraphics[height=4.7cm]{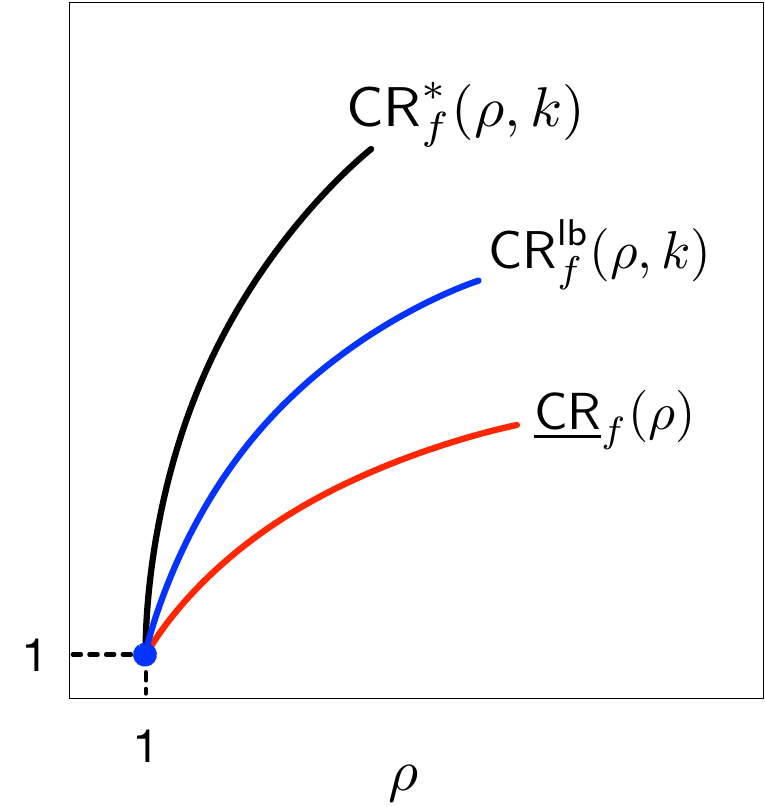}
	\end{center}
	\caption{Illustration of $\textsf{CR}_f^*(\rho,k) $, $ \textsf{CR}_f^{\textsf{lb}}(\rho,k) $, and $ \underline{\textsf{CR}}_f(\rho) $.}
	\label{Fig_CR_bound}
\end{wrapfigure}

In addition, we also obtain some interesting yet counter-intuitive results regarding the impact of non-decreasing marginal costs on online selection: \textit{a faster-rising marginal cost provides a better and more robust performance guarantee for online selection in the worst-case}. To be more specific, as one would expect, the existence of non-decreasing marginal costs has a negative impact on the seller's profit in the offline setting. In the online setting, however, our results show that a non-decreasing marginal cost actually leads to a better and more robust performance in terms of having a smaller competitive ratio (compared to online selection without costs). Moreover, the faster the marginal costs grow, the better the competitive ratio will be. For example, previous studies (e.g., \cite{OKP_Zhou_2008, Tan_ORA_2020, Zhang2017}) have shown that no online algorithm can achieve a competitive ratio better than $ 1 + \ln \rho  $ for OSCC with $ f = 0 $. Our results show that a better competitive ratio can be obtained if $ f $ is convex. Moreover, strong convexity can further improve the competitive ratio with a better worst-case performance guarantee. These findings provide a new perspective on the impact of non-decreasing marginal costs on real-world online resource allocation problems, especially if robustness or worst-case guarantee is a major concern. Meanwhile, we also expect that these results could shed light on a very fruitful research agenda that examines more general settings of OSCC such as a broader class of cost models, multi-unit demand, and multiple types of resources, etc.

\subsection{Overview of Our Techniques} 
To analyze the best-possible competitive ratio for OSCC, we adopt the following techniques to tackle the challenges en route. 
\begin{itemize}[leftmargin=*]
	\item \textit{System of sufficient equations (SoSEs)}. We  develop a group of $ \alpha $-parameterized inequalities based on the online primal-dual framework and show that there exists an $ \alpha $-competitive  threshold policy if these inequalities are satisfied. The problem then reduces to  appropriately designing the threshold  so as to minimize the competitive ratio $ \alpha $ while preserving the feasibility of these $ \alpha $-parameterized inequalities. This requires solving an optimization problem with non-convex and mixed-integer constraints.  To address this challenge, we  derive a system of equations, dubbed SoSEs, and prove that i) there exists a unique set of positive real numbers that satisfy the SoSEs, and ii) the threshold policy designed based on these positive real numbers is $ \textsf{CR}_f^*(\rho, k) $-competitive, which is the optimal competitive ratio of all deterministic threshold policies.
	
	\item \textit{System of necessary equations (SoNEs)}. A key challenge en route is the proof of why $ \textsf{CR}_f^*(\rho,k) $ is optimal not just for the subclass of threshold policies, but for all deterministic algorithms. Our idea is to construct a series of ``hard" arrival instances and show that if there exists an $ \alpha $-competitive deterministic  algorithm $ \textsf{ALG} $, then $ \textsf{ALG} $ must achieve at least $ \frac{1}{\alpha} $ fraction of optimum, denoted by $ \textsf{OPT} $, under these hard arrival instances. This helps us construct another system of equations, dubbed SoNEs, by repeatedly applying $\textsf{ALG} = \frac{1}{\alpha}\textsf{OPT}$. Here, the term ``\textit{necessity}" means that whenever there exists an $ \alpha $-competitive deterministic algorithm, then it is necessary to have a unique set of positive real numbers to satisfy the SoNEs. Based on this principle, we prove that for any $ \alpha =  \textsf{CR}_f^*(\rho,k) - \epsilon $ with $ \epsilon > 0 $, the SoNEs has no solution, leading to the conclusion that $ \textsf{CR}_f^*(\rho,k) $ is the smallest-possible competitive ratio of all deterministic algorithms. Our proof towards this conclusion is given in Appendix \ref{sec_proof_of_theorem_optimality}.
	
	\item \textit{Convergence to ordinary differential equations (ODEs)}. As to the asymptotic convergence, we leverage the fact that the importance or weight of each unit of the asset becomes negligible when $ k $ is large. Based on this property, we prove that the SoSEs and SoNEs asymptotically converge to the same ODE when $ k \rightarrow +\infty $. The asymptotic lower bound $ \underline{\textsf{CR}}_f(\rho) $ can then be derived by analyzing the existence and uniqueness of solutions to the ODE.
\end{itemize}

We emphasize that \textit{the SoSEs and SoNEs are developed separately without any prior knowledge of each other, but coincide in the end}. This allows for the conclusion that our proposed threshold policy is optimal among all deterministic algorithms. Meanwhile, the construction of the SoNEs also plays a key role in deriving the lower bound on competitive ratios of randomized algorithms. 

\subsection{Related Work}
For setups where producing additional units is free  (i.e., $ f = 0 $), the competitive ratio of OSCC has been studied extensively under many streams of literature, which we review in the following:  
	
\textit{$ k $-max search}: If the total number of buyers $ T $ is known and exactly $ k $ buyers are to be selected, then the presented online selection problem is the $ k $-max search problem of Lorenz et al. \cite{k_search_2009}. In $ k $-max search, the goal is to search for $ k $ highest prices in a sequence of fixed length $ T $. When the prices are chosen from the real interval $ [p_{\min}, p_{\max}] $, Lorenz et al. \cite{k_search_2009} developed optimal deterministic and randomized algorithms for $ k $-max search. In particular, when $ k $ is large, the competitive ratio of an optimal $ k $-max search algorithm behaves asymptotically like $ \ln (p_{\max}/p_{\min}) $. While for finite $ k \geq 1 $, the optimal competitive ratios derived in \cite{k_search_2009} are given as some implicit functions of $ k, p_{\min} $, and $ p_{\max} $, which in this paper is referred to as setup-dependent competitive ratios\footnote{The $ k $-max search problem is also closely related to the $ k $-\textit{secretary} problem \cite{k_secretary_2005} if buyers are presented in \textit{random order} \cite{random_order_Gupta_2021} and the goal is to select $ k $ of them such that the expected sum of the selected prices is maximized. Since in this work we primarily focus on competitive analysis of OSCC in the worst-case setting, we refer to \cite{random_order_Gupta_2021} for a recent survey of random order model and the $ k $-secretary problem (and its many variants).}.

\textit{One-way trading}: Closely related to $ k $-max search is the one-way trading problem, which was first pioneered by \cite{time_series_search_2001} and then followed by a growing body of literature with many variants, e.g., \cite{one_way_trading_2019, OKP_one_way_trading_Cao_2020, OKP_EV_BoSun_2020}. In one-way trading, an investor wants to trade a total amount of 1 dollar into yen with sequential arrivals of exchange rates. At time $ t = 1, 2, \cdots, T $, rate $ p_t $ is revealed and if $ x_t \in [0, 1] $ dollars are traded, then the investor gains $ p_t x_t $ amount of yen. The goal is to maximize the amount of yen acquired in the end, namely, $ \sum_{t=1}^T p_t x_t $,  subject to the budget limit $ \sum_{t=1}^T x_t \leq 1 $. EI-Yaniv et al. \cite{time_series_search_2001} proved that one-way trading is essentially equivalent to 1-max search in the sense that any deterministic/randomized one-way trading algorithm can be viewed as a randomized 1-max search algorithm. In particular, if the prices are chosen from the real interval $ [p_{\min}, p_{\max}] $, a simple threshold-based 1-max search algorithm can achieve the optimal competitive ratio of $ \sqrt{p_{\max}/p_{\min}} $ (among all deterministic algorithms). EI-Yaniv et al. \cite{time_series_search_2001} also developed a randomized 1-max search algorithm that is $ O(\log(p_{\max}/p_{\min})) $-competitive when $ p_{\max}/p_{\min} \rightarrow +\infty $.

\textit{Online auctions}: If incentive issue (i.e., buyers may take strategic behaviors) is taken into account, then the presented online selection problem has also been extensively studied in the context of online auctions. Specifically, an auctioneer is selling $ k $ copies of identical items to a sequence of buyers. At each time $ t = 1, 2, \cdots $, buyer $ t $ offers his/her bid $ p_t $ for one copy of the item, which may or may not be the buyer's true valuation $ v_t $. The goal is to develop an online mechanism such that a desired objective (e.g., welfare or profit maximization) can be achieved and at the same time, to guarantee that it is in buyers' best interest to bid with  $ p_t = v_t $. Mechanisms with such properties are considered to be truthful or incentive compatible, which is a golden standard to ensure economic efficiency \cite{AGT}. When each buyer's valuation is within $ [v_{\min}, v_{\max}] $, Bartal et al. \cite{CA_multi_unit_2003} proposed an $ O(\log (v_{\max}/v_{\min})) $-competitive online mechanism when there are $ \Omega(\log (v_{\max}/v_{\min})) $ copies of each item. Huang et al. \cite{Huang_2019} later showed that no online algorithms can be $ o(\log (v_{\max}/v_{\min}) ) $-competitive, indicating that the logarithmic lower bound cannot be overcome even by randomized algorithms. 
	
\textit{Online knapsack}. Online selection is also related to a special type of online knapsack problem where each item $ t $ has a \textit{unit weight} with value $ p_t $ (i.e., the weight of each item equals 1). The standard online knapack problem has received a surge in attention since the seminal work by Marchetti-Spaccamela and Vercellis \cite{stochastic_knapsack_1995}. In the standard setting  when both the weight and value are arbitrary, no online algorithm can be $ O(1) $-competitive, even for algorithms with randomization \cite{stochastic_knapsack_1995}. Thus, most of the existing literature focuses on restricted variants of the classic setting \cite{OKP_Zhou_2008, Zhang2017, OKP_EV_BoSun_2020} or algorithms beyond worst-case analysis, e.g.,  average-case analysis \cite{Average_knapsack_1995}, online optimization with predictions \cite{okp_predictions_2021}, etc. Among these works,  \cite{OKP_Zhou_2008} is particularly related. To be more specific, when the weight of each item is small compared to the capacity of the knapsack\footnote{This is equivalent to our presented online selection problem when $ k \rightarrow +\infty$.}, and the value-to-weight ratio of each item is bounded within the range of $ [L, U] $, the authors of \cite{OKP_Zhou_2008} developed a threshold policy which is $ \big(1+\ln(U/L) \big)$-competitive. It was further proved by Yao's minimax principle that no online algorithm can achieve a competitive ratio better than $ 1+\ln(U/L) $ \cite{OKP_Zhou_2008}. 

Our main contribution is to extend all of the preceding results to setups where producing additional units of the asset is not free (i.e., $ f \neq 0 $). As mentioned earlier, such a setup arises naturally from online resource allocation problems involving   production or operating costs  \cite{Tan_ORA_2020, Zhang2017, online_allocation_regularizer_2021, Blum_2011, Huang_2019}. In particular, Tan et al. \cite{Tan_ORA_2020} considered a general online resource allocation problem with convex production costs, and proposed an optimal threshold policy under the assumption that the demand of each buyer is \textit{infinitesimal}. A similar problem  was also studied in a more complex setting of online combinatorial auctions in which buyers have combinatorial valuations over bundles of resources \cite{Huang_2019, Blum_2011}. However, their models do not consider the capacity limit (i.e., the seller can produce infinitely-many units of the asset at increasing marginal costs). More importantly, the key difference between this work and \cite{Tan_ORA_2020, Huang_2019} is that the unit demand of each buyer in OSCC introduces 0-1 integral objectives and constraints, leading to the failure of the techniques proposed by \cite{Tan_ORA_2020, Huang_2019} for our problem. It is also worth noting that Huang et al. \cite{Huang_2019} later extended their main results regarding the infinitesimal case to the 0-1 integral and limited $ k$-supply cases, but their competitive ratios are not optimal. While in this paper, we derive the optimal competitive ratios for deterministic algorithms under various setups of OSCC, and show that these competitive ratios are asymptotically optimal for randomized algorithms.  

The techniques adopted in this paper are related to the online primal-dual framework, a tool that has proved powerful in approaching various online optimization problems, including online matching \cite{b_matching, willma_OR_2020}, Adwords problems \cite{Adword2009, Adwords2007}, online covering/packing problems \cite{online_routing_FOCS, Buchbinder2009}. It has also been generalized to handle online optimization with nonlinear and combinatorial objectives \cite{concave_return, covering_packing, Buchbinder2015, Blum_2011, Huang_2019}, or more recently, with machine learned advice \cite{OPD_learning_2020}. For more details of the online primal-dual framework, please refer to \cite{OPD2009}.

\section{The Model}
\label{section_preliminaries}
Throughout the paper,  we use $ [k] =  \{1,2,\cdots, k\} $ to denote the set of positive integers not exceeding $ k $, and use $ [0,k] $ to denote the continuous interval of real numbers between 0 and $ k $. 

\subsection{OSCC: Problem Statement and Assumptions}
\label{section_OSCC_statement}
A seller is selling some asset to a sequence of $ T\geq 1 $ buyers who arrive one at a time. The seller can produce $ k \geq 1 $ units of the asset in total, but at non-decreasing marginal costs. Specifically, the production cost is characterized by the following function $ f $:
\begin{align}\label{equation_f}
	f(y)=
	\begin{cases}
		%0 & \text{if } x = 0,\\
		\text{convex and increasing} & \text{if } y \in [0, k],\\
		+\infty & \text{otherwise}.
	\end{cases}
\end{align}
Here, we denote by $ f(i)  $ the \textit{total  cost} of producing the first $ i $ units. In other words, if $ c_i $ denotes the \textit{marginal cost} of producing the $ i $-th unit, then $ f(i) = c_1 + c_2 + \cdots + c_i, \forall i\in [k]$. Conversely, we can also write the marginal cost $ c_i = f(i) - f(i-1) $, and thus a convex and increasing $ f $ guarantees that the marginal costs are non-decreasing, i.e., $ c_{i+1} \geq c_i $ holds for all $ i\in [k-1] $.  At time $ t = 1, 2, \cdots, T $, buyer $ t $ arrives and offers $ p_t \in [p_{\min}, p_{\max}] $ for one unit of the asset, where $ p_{\min} $ and $ p_{\max} $ are respectively the minimum and maximum price offered by the buyers.\footnote{Ma et al. \cite{willma_OR_2020} studied an online assortment and pricing problem in which customers are allowed to bid with a discrete set of prices (e.g., $ \{0.98, 1.98, 2.98\} $ \$/unit). Alternatively,  buyers may offer arbitrary prices that do not necessarily follow any lower and upper limits \cite{Blum_2011}, or the prices can be predicted with machine learned advice \cite{secretary_ML_advice_NuerIPS_2020, online_selection_constrained_adversary_ICML_2021}. Our problem setting is a middle ground between these models and is commonly used in various online selection problems (e.g., \cite{Huang_2019, CA_multi_unit_2003, k_search_2009, time_series_search_2001}). } We denote by $ x_t = 1 $ if buyer $ t $ is selected (i.e., $ p_t $ is accepted) and $ x_t = 0 $ otherwise.  The goal is to select a subset of buyers to maximize the seller's profit $ \sum_t p_t x_t - f(\sum_t x_t) $.

Note that we do not assume the knowledge of $ T $, so it is possible that we have fewer buyers to select than the total units the seller can produce (i.e., $ T < k $). We do not consider infinite marginal costs, so $ \{c_i\}_{\forall i} $ is assumed to be upper bounded by $ c_{\max} $, namely, $ c_1\leq c_2\leq \cdots \leq c_k\leq c_{\max} $. Meanwhile, we also assume zero start-up cost, namely, $ f(0) = 0 $, and thus $ c_1 = f(1) - f(0) = f(1) $. 

We can divide the elements described above into two groups:
\begin{itemize}[leftmargin=*]
\item The \textit{Setup} $ \mathcal{S} = \{f, p_{\min}, p_{\max},k\} $, which consists of parameters that formally define  OSCC\footnote{We add two remarks here regarding the setup of OSCC. First, we can eliminate $ k $ in $ \mathcal{S} $ since it is included in the production cost function $ f $ according to Eq. \eqref{equation_f}. Second, to guarantee that a setup $ \mathcal{S} $ is interesting, we implicitly assume $ p_{\max}\geq p_{\min} > c_1 $. Specifically, we assume that $ p_{\min} > c_1 $ since any price $ p_t \leq  c_1 $ should be immediately rejected. Meanwhile, it is not interesting to consider the setup with $ p_{\max} < c_1  $ since in this scenario we have $ p_t \leq p_{\max} < c_1, \forall t $, and thus the naive strategy, which is also optimal, is to simply reject all incoming buyers.}.   
\item The \textit{Arrival Instance} $  \mathcal{I} = \{p_1, p_2, \cdots, p_T\} $, which consists of the information revealed over time. 
\end{itemize}

Given a setup $ \mathcal{S} $, an online algorithm  must decide, based on information revealed over time in $ \mathcal{I} $, whether a price should be accepted or rejected. The goal is to develop online algorithms that achieve a competitive ratio that is as close to 1 as possible.

\subsection{Definitions and Notations}
Here we introduce some definitions and notations that are key to the upcoming analysis. Based on the sequence of marginal costs $ \{c_i\}_{\forall i} $, we define $ \Gamma(p): [p_{\min},p_{\max}] \rightarrow [k] $ by
\begin{align}\label{eq_def_Gamma}
	\Gamma(p) \triangleq \sum_{i=1}^{k} i\cdot \mathds{1}_{\big\{p\in [c_i,c_{i+1})\big\}}, \quad  p\in  [p_{\min},p_{\max}].
\end{align}
Here, $ \mathds{1}_{\{\text{`A'}\}} $ is an indicator function which equals 1 if `A' is true and 0 otherwise. Note that in Eq. \eqref{eq_def_Gamma} we define $ c_{k+1} = +\infty $ as the marginal cost of producing the ``virtual"  $ (k+1) $-th unit. 

One can interpret $ \Gamma(p) $ as the maximum number of units to be produced if all the buyers offer the same price $ p $. In this case, the seller has no incentive to produce more than $ \Gamma(p) $ units since the production cost of the $ (\Gamma(p)+1) $-th unit would be higher than the offered price $ p $. Based on the definition of $ \Gamma(p) $, we also define $ \ubar{k} $ and $ \bar{k} $ by:
\begin{equation}\label{eq_M}
	\ubar{k}  \triangleq \Gamma(p_{\min}), \quad \bar{k} \triangleq \Gamma(p_{\max}),
\end{equation}
which  means the maximum number of units to be produced if all the prices are $ p_{\min} $ and $ p_{\max} $, respectively. Note that $ \ubar{k} $ and $ \bar{k} $ are determined as long as the setup $ \mathcal{S} $ is given, and $ 1\leq \ubar{k} \leq  \bar{k} \leq  k $ always holds.  For example, suppose $ k = 20$ and consider $ c_3 \leq  p_{\min} <  c_4 $ and $ c_{10} \leq  p_{\max} <  c_{11} $, then $ \ubar{k} = 3 $ and $ \bar{k} = 10 $. In this case, $ \bar{k} $ is strictly less than $ k $,  meaning that it is not profitable to sell more than $ \bar{k} $ units although the seller can (physically) produce $ k > \bar{k} $ units in total. On the other hand, if the setup $ \mathcal{S} $ is such that $  p_{\max} \geq p_{\min} > c_{20} $, then $ \ubar{k} =  \bar{k} = k = 20 $, meaning that it is profitable for the seller to produce $ k $ units even if all the buyers offer the minimum price $ p_{\min} $. 

\begin{definition}[Min-Profit]\label{def_g}
	We define $ g(i) $ as a function of $ i \in \{0,1,\cdots,\ubar{k}\} $  as follows:
	\begin{align}\label{eq_g}
		g(i)  \triangleq p_{\min}i - f(i), \quad i\in \{0,1,\cdots,\ubar{k}\}.
	\end{align}
\end{definition}

By definition, one can interpret $ g(i) $ as the minimum profit of producing $ i $ units of the asset. Thus, we refer to $ g $ as the \textit{min-profit function} hereinafter.  Based on the definition of $ \ubar{k} $ in Eq. \eqref{eq_M}, $ g(i) $ is strictly-increasing over  $ i\in \{0,1,\cdots,\ubar{k}\} $, i.e., $ g(i) > g(i-1) $ holds for $ i\in [\ubar{k}]$.

\begin{definition}[Min-Production]\label{theorem_inverse_g} 
	We define $ g^{\textsf{inv}}(\varsigma) $ as a function of  $ \varsigma \in [0, g(\ubar{k})]   $ as follows:
	\begin{align}\label{eq_inverse_g}
		g^{\textsf{inv}}(\varsigma) \triangleq \sum_{i=0}^{\ubar{k}-1} (i+1) \cdot \mathds{1}_{\big\{\varsigma \in (g(i), g(i+1)]\big\}},
	\end{align}
	which maps a real value $ \varsigma \in [0, g(\ubar{k}))] $ to an integer in $ [\ubar{k}] $.
\end{definition}

We refer to  $ g^{\textsf{inv}}(\varsigma) $ as the \textit{min-production function} because it captures the minimum number of units the seller should produce to achieve a profit no less than $ \varsigma $. Mathematically,  the definition of $ g^{\textsf{inv}} $ is an extension of the standard inverse function of $ g(i) $ to the discrete domain by rounding the value of $ g^{\textsf{inv}}(\varsigma) $ up to the nearest integer so that $ g\left(g^{\textsf{inv}}(\varsigma)\right) \geq  \varsigma $. For example, given an integer  $ i\in [\ubar{k}] $, if $ \varsigma = g(i) $, then $ g^{\textsf{inv}}(\varsigma) = i $; if $ \varsigma = g(i) + \epsilon $ for some small positive value $ \epsilon > 0 $, then $ g^{\textsf{inv}}(\varsigma) = i+1 $ as long as $ g(i) + \epsilon \leq g(i+1) $. 

\begin{definition}[Conjugate]\label{def_conjugate}
    Given the production cost function $ f $,  we define its conjugate $ f^{*} $ by
    %\vspace{-0.1cm}
	\begin{equation}\label{eq_f_star}%\vspace{-0.1cm}
		f^{*}(p) \triangleq  \max_{i \in \{0,1,\cdots, k\}}\ p i - f(i), \quad   p\in [0,+\infty).
	\end{equation}
\end{definition}

The definition of $ f^{*} $ is a generalized version of the standard Fenchel conjugate \cite{convex_nonlinear_book_2006}. Here, $ i $ takes integer values only while the Fenchel conjugate is defined on continuous functions. Lemma \ref{theorem_conjugate} below shows that $ f^* $ can be written as an analytical function of $ p $ and $ \Gamma(p) $ for any $ p\in [p_{\min},p_{\max}] $.

\begin{lemma}\label{theorem_conjugate}
	The conjugate function $ f^*(p) $ is continuous, piecewise linear, and strictly increasing over $ p\in \mathbb{R}^+ $. Meanwhile, for $ p\in [p_{\min},p_{\max}] $, $ f^{*}(p) $ can be written as 
	\begin{align}\label{f_star}
		f^{*}(p)  = p \Gamma(p) - f\left(\Gamma(p)\right).
	\end{align}
\end{lemma}

The conjugate $ f^{*}(p) $ captures the maximum profit achievable under a sequence of buyers whose offered prices are less than or equal to $ p $. For example, by Eq. \eqref{eq_M} and Eq. \eqref{f_star}, $ f^*(p_{\min}) = p_{\min}\Gamma(p_{\min}) - f(\Gamma(p_{\min})) = p_{\min}\ubar{k} - f(\ubar{k})$, meaning that if the arrival instance is such that all the buyers offer $ p_{\min} $, then the maximum profit achievable is $ p_{\min}\ubar{k} - f(\ubar{k}) $ (selling $ \ubar{k} $ units at price $ p_{\min} $).

\section{Threshold Policies and Competitive Analysis}
\label{section_threshold_policies}
In this section, we first introduce a threshold-based online selection algorithm and explain how it works. After that, we  discuss the principles of how to design a competitive threshold.

\subsection{Threshold-based Online Selection (TOS)}

We propose a threshold policy in Algorithm \ref{TP} below, dubbed $ \textsf{TOS}_{\boldsymbol{\lambda}} $. Given a predesigned threshold $ \boldsymbol{\lambda} = (\lambda_0, \lambda_1,\cdots, \lambda_{\bar{k}}) $, $ \textsf{TOS}_{\boldsymbol{\lambda}} $ makes decisions of rejecting or accepting buyers based on whether their offered prices exceed the corresponding threshold or not. Note that under $ \textsf{TOS}_{\boldsymbol{\lambda}} $, at most $ \bar{k} $ units of the asset will be produced. This follows our definition of $ \bar{k} $ in Eq. \eqref{eq_M} as it is either not profitable to produce more than $ \bar{k} $ units (when $ \bar{k} < k $), or we already reach the capacity limit (when $ \bar{k} = k $).

\begin{algorithm}
	\caption{Threshold-based Online Selection ($ \textsf{TOS}_{\boldsymbol{\lambda}} $)}\label{online_mechanism_U}	
	\begin{algorithmic}[1]
		\STATE \textbf{Inputs:} $ \boldsymbol{\lambda} = (\lambda_0, \lambda_1,\cdots, \lambda_{\bar{k}})$. 
		\STATE \textbf{Initialization}: $ i = 0 $.
		
		\WHILE{a new buyer $ t $ arrives}
		
		\IF {$ p_t - \lambda_i < 0 $ or $ i > \bar{k} $}\label{if_utility_negative_U} %$ z_{n-1}^t + y_{n*}^t> 1 $ for some $ t\in\mathcal{T}_n $
		\STATE $ p_t $ is rejected (i.e., set $ x_t =0 $) \label{rejection_1_U}
		
		\ELSE 
		\STATE $ p_t $ is accepted (i.e., set $ x_t = 1 $)\label{accepted_U}
		
		\STATE $ i = i + 1$.
		
		\ENDIF \label{end_if_U}
		\ENDWHILE
	\end{algorithmic}
	\label{TP}
\end{algorithm}

As a threshold policy, $ \textsf{TOS}_{\boldsymbol{\lambda}} $ is simple  and implementation-friendly\footnote{More importantly, threshold policies have the merit of being incentive compatible, and thus avoid incentive issues such as manipulation of the market when buyers are strategic \cite{AGT}. For example, $ \textsf{TOS}_{\boldsymbol{\lambda}} $ can be implemented as a \textit{posted price mechanism} \cite{Tan_ORA_2020, Blum_2011, Huang_2019} by publishing the threshold $ \lambda_i $ to each incoming buyer $ t $, and let buyer $ t $ solve line 4-line 9 of $ \textsf{TOS}_{\boldsymbol{\lambda}} $.}. 
Our objective is to design a threshold $ \boldsymbol{\lambda} $ so that $ \textsf{TOS}_{\boldsymbol{\lambda}} $ achieves a competitive performance.  Intuitively, the threshold should never be lower than $ p_{\min} $ as it suffices to make any buyer selected with  $ p_{\min} $.  Similarly, there is no need to design a threshold higher than $ p_{\max} $ since  $ p_{\max}  + \epsilon $ is enough to reject all buyers for any $ \epsilon \rightarrow 0^+ $. Thus, the threshold $ \boldsymbol{\lambda} $ should be a sequence of $ \bar{k} + 1 $ positive real numbers within the range $ [p_{\min},p_{\max}] $, and the design objective is to guarantee that  $ \textsf{TOS}_{\boldsymbol{\lambda}} $ achieves a bounded competitive ratio that is as close to 1 as possible.

\subsection{Admission Thresholds: Definition and A Three-Case Overview}%\vspace{-0.1cm}
Our first proposition below shows that, $ \textsf{TOS}_{\boldsymbol{\lambda}} $ achieves a bounded competitive ratio only if a certain number of buyers are selected at the beginning regardless of their offered prices.

\begin{proposition}\label{theorem_k_alpha}
	For any $ \alpha \geq 1 $, if $\normalfont \textsf{TOS}_{\boldsymbol{\lambda}} $ is $ \alpha $-competitive, then $ \lambda_0 =  \lambda_1 = \cdots = \lambda_{\tau} =  p_{\min} $ holds for some $ \tau \in \mathcal{T}^{(\alpha)} $, where $ \mathcal{T}^{(\alpha)} $ is defined as
	\begin{align}\label{eq_K_alpha}
		\mathcal{T}^{(\alpha)} \triangleq \left\{ g^{\textsf{inv}}\Big(\frac{1}{\alpha} f^{*}(p_{\min})\Big) - 1,\cdots, \ubar{k}-1 \right\}.
	\end{align}
	In Eq. \eqref{eq_K_alpha}, $ g^{\textsf{inv}} $ is the min-production function defined in Eq. \eqref{eq_inverse_g}.
\end{proposition}
\begin{proof}
	The proof is based on the fact that if $\textsf{TOS}_{\boldsymbol{\lambda}} $ is $ \alpha $-competitive, then $\textsf{TOS}_{\boldsymbol{\lambda}} $ must achieve at least $ \frac{1}{\alpha} $ portion of optimum under all possible arrival instances. Let us assume the arrival instance is given by $ \mathcal{I} = \{p_{\min},p_{\min},\cdots,p_{\min}\} $, where there exist $ k $ buyers with price $ p_{\min} $. In the offline setting, the optimal decision is to produce $ \ubar{k} $ units and the optimal profit is given by
	$\textsf{OPT}(\mathcal{I}) = p_{\min}\ubar{k} - f(\ubar{k}) = f^*(p_{\min})$. Here, we use the definition of $ \ubar{k} $ in Eq. \eqref{eq_M} and the property of $ f^* $ in Lemma \ref{theorem_conjugate}. In the online setting, $ \textsf{TOS}_{\boldsymbol{\lambda}} $ is $ \alpha $-competitive indicates that there must exist an integer $ \tau \in \{0,1,\cdots,\ubar{k}-1\} $ such that 
	%\vspace{-0.2cm}
	\begin{equation*}%\vspace{-0.1cm}
		\textsf{TOS}_{\boldsymbol{\lambda}}(\mathcal{I}) 
		= p_{\min}(\tau+1) - f(\tau+1) = g(\tau+1) \geq \frac{1}{\alpha}\textsf{OPT}(\mathcal{I}) = \frac{1}{\alpha} f^*(p_{\min}),
	\end{equation*}
	which thus leads to $ \tau \geq  g^{\textsf{inv}}(\frac{1}{\alpha} f^{*}(p_{\min})) - 1$.  On the other hand,  $ \tau+1 \leq \ubar{k} $ always holds as no more than $ \ubar{k} $ buyers with price $ p_{\min} $ should be selected. Combining the lower and upper bounds of $ \tau $, we complete the proof.
\end{proof}

Proposition \ref{theorem_k_alpha} drastically simplifies the design and analysis of threshold policies because it restricts to thresholds which are horizontal at the beginning (i.e., \textit{horizontal segment}). Intuitively, after selecting some buyers at the initial stages, the seller should gradually increase the threshold when more and more buyers are selected (i.e., \textit{increasing segment}) because the marginal cost is non-decreasing. To formally define such type of thresholds, we define \textit{admission thresholds} below. 

\begin{definition}[Admission Threshold]\label{def_admission_threshold}
	An admission threshold $ \boldsymbol{\lambda} = (\lambda_0, \lambda_1,\cdots, \lambda_{\tau}, \cdots, \lambda_{\bar{k}}) $ is a sequence of $ \bar{k}+1 $ monotonically non-decreasing positive numbers such that  
	\begin{align}\label{eq_admission_threshold}
		\lambda_0 = \lambda_1 =  \cdots = \lambda_\tau  = p_{\min}  <  \lambda_{\tau+1}\leq \cdots\leq \lambda_{\bar{k}-1} \leq \lambda_{\bar{k}} \leq p_{\max},
	\end{align}
	where $ \tau $ is an integer within $ \{0,1,\cdots, \ubar{k}-1\} $.
\end{definition}

\begin{figure}
	\centering
	\subfigure[\textsc{High-Value}: $ c_k <  p_{\min} \leq  p_{\max} $]{\includegraphics[width=4.7cm]{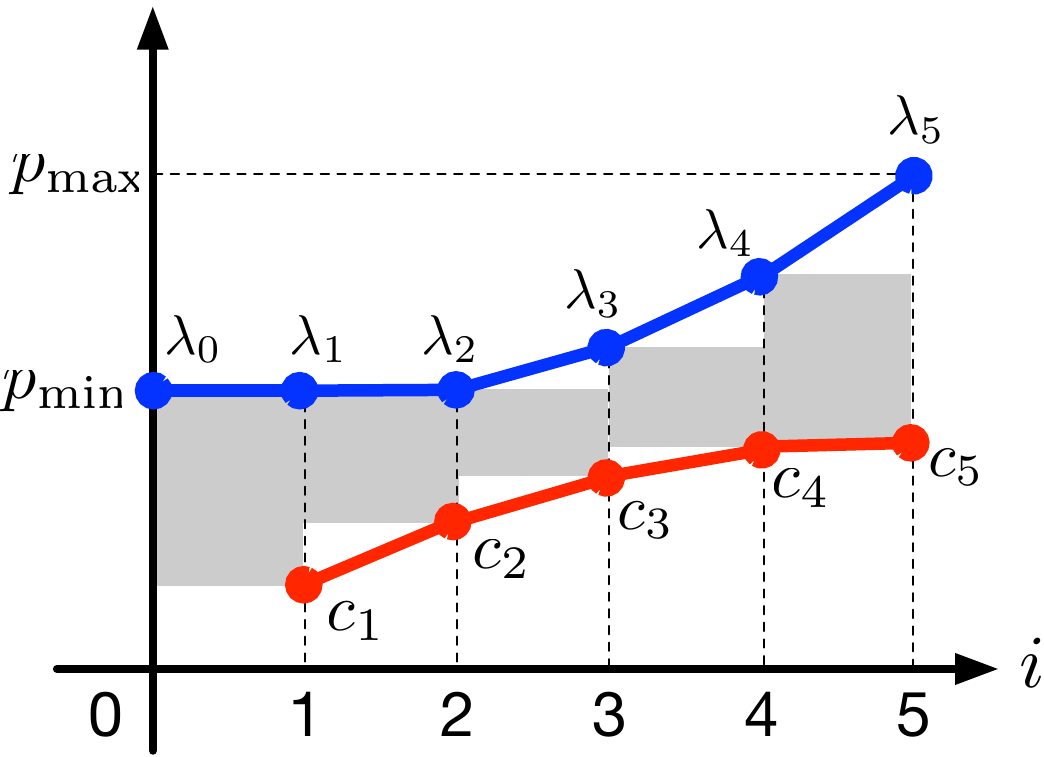}}
	\qquad 
	\subfigure[\textsc{Low-Value}: $ p_{\min} \leq p_{\max} \leq c_k $]{\includegraphics[width=4.7cm]{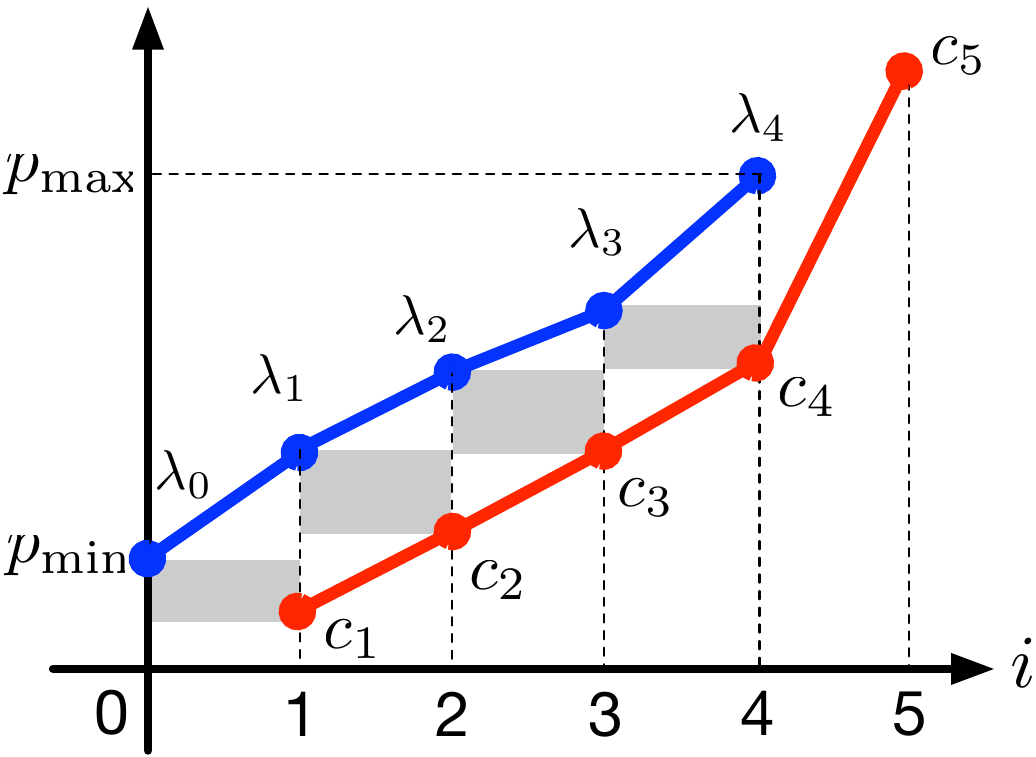}}
	\qquad 
	\subfigure[\textsc{Mix-Value}: $ p_{\min} < c_k < p_{\max} $]{\includegraphics[width=4.7cm]{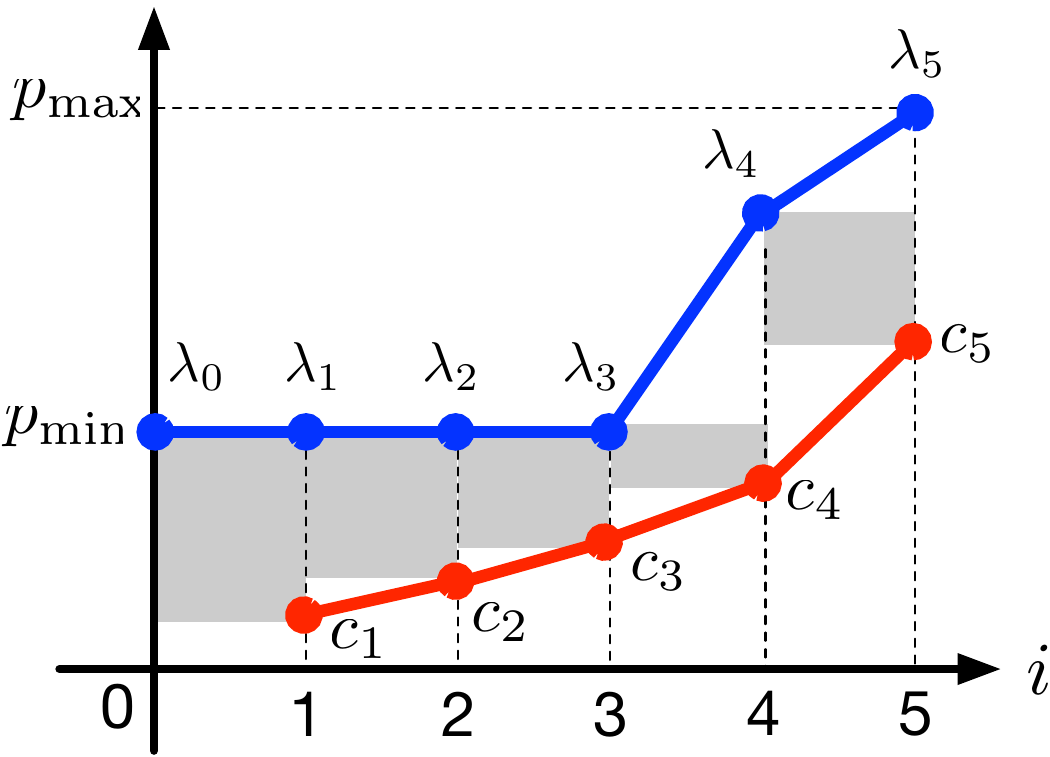}}
	
	%\vspace{-0.3cm}
	\caption{Illustration of admission thresholds with $ k = 5 $ in three cases. In subfigure (a), the turning point $ \tau = 2 $ (i.e., $ \lambda_0 = \lambda_1 = \lambda_2 = p_{\min} $) and $ \ubar{k} = \bar{k} = k = 5$; In subfigure (b), the turning point $ \tau = 0 $ (i.e., $ \lambda_0 = p_{\min} $), $ \ubar{k} = 1$, $ \bar{k} =4$, and $ k = 5$; In subfigure (c), the turning point $ \tau = 3 $ (i.e., $ \lambda_0 = \lambda_1 = \lambda_2 = \lambda_3 =  p_{\min} $), $\ubar{k} = 4 $, and $ \bar{k} = k = 5$. Recall that the sequence of marginal costs are non-decreasing, i.e.,  $ c_{i+1} \geq  c_i $ holds for all $ i\in [k-1] $.
	}
	\label{fig_three_cases}
	%\vspace{-0.4cm}
\end{figure}

For a given admission threshold, the integer $ \tau $ determines the turning point after which the threshold starts to increase. Thus, we refer to $ \tau $ as the \textbf{turning point} hereinafter.  Intuitively, an admission threshold $ \boldsymbol{\lambda} $
can be uniquely determined once the turning point $ \tau $ and the sequence of $ \bar{k} - \tau $ positive real numbers $ \{\lambda_{\tau+1}, \lambda_{\tau+2}, \cdots, \lambda_{\bar{k}}\} $ are given (i.e., the increasing segment). 

(\textbf{Admission Thresholds in Three Cases}) As illustrated in Fig. \ref{fig_three_cases}, depending on the setup $ \mathcal{S} $, the admission threshold can be designed in the following three cases:
\begin{enumerate}[leftmargin = *]
	\item The \textsc{High-Value} case (Fig. \ref{fig_three_cases}(a)): $ c_k < p_{\min}\leq p_{\max} $. In this case, $ \ubar{k} = \bar{k} = k $. The \textsc{High-Value} case may occur if the costs of resource  supply is much lower than the value/revenue it generates\footnote{This could be a monopoly where $ p_{\min} $ and $ p_{\max} $ are set to be much higher than the production costs.}. 
	
	\item The \textsc{Low-Value} case (Fig. \ref{fig_three_cases}(b)): $ p_{\min}\leq p_{\max} \leq c_k $. In this case, $ \ubar{k} \leq  \bar{k} \leq k $. The \textsc{Low-Value} case can model application scenarios where the costs of resource supply is comparable to or even higher than the value/revenue it generates\footnote{This could be a service platform that is poorly managed, e.g., a highly congested communication network \cite{network_RA, network_resource_mechanism_jsac_2006} when  delay/latency is modeled as a convex cost function of link utilization (i.e., the production costs in our OSCC formulation).}.
	
	\item The \textsc{Mix-Value} case (Fig. \ref{fig_three_cases}(c)): $ p_{\min} < c_k  < p_{\max} $. In this case, $ \ubar{k} <  \bar{k} = k $. As suggested by the name, the \textsc{Mix-Value} case sits in between the \textsc{High-Value} and \textsc{Low-Value} cases.
\end{enumerate}

The above three cases provide a complete characterization of admission thresholds under different setups of $ \mathcal{S} = \{f, p_{\min}, p_{\max}, k\} $. It is worth emphasizing that the purpose of discussing these three cases is not just conceptual --- Section \ref{section_lower_bound} and Section \ref{section_asymptotic_properties} will show that the lower bound analysis and asymptotic properties   of the admission thresholds in these three cases are different.

\subsection{Design Principles of Competitive Admission Thresholds}
\label{section_design_principles}

Proposition \ref{theorem_CR_threshold} below shows that $ \textsf{TOS}_{\boldsymbol{\lambda}} $ attains a constant competitive ratio as long as $ \boldsymbol{\lambda} $ is an admission threshold. 

\begin{proposition}\label{theorem_CR_threshold}
	Given a setup $ \mathcal{S} $, for any admission threshold $ \boldsymbol{\lambda} = \{\lambda_0,\lambda_1,\cdots, \lambda_\tau, \cdots, \lambda_{\bar{k}}\} $ with $ \tau \in \{0,1,\cdots,\ubar{k}-1\} $, $\normalfont \textsf{TOS}_{\boldsymbol{\lambda}}$ achieves a bounded competitive ratio given by  $ \alpha $:
	\begin{align}\label{eq_cr_admission_threshold}
		\alpha = \max_{j\in [\bar{k}-\tau-1]} \left\{\frac{f^{*}(\lambda_{\tau+j})}{\sum_{i=0}^{\tau+j-1}(\lambda_i - c_{i+1})}, \ \frac{f^{*}(p_{\max})}{\sum_{i=0}^{\bar{k}-1}(\lambda_i - c_{i+1})}\right\}.
	\end{align}
\end{proposition}

\begin{proof} 
	Eq. \eqref{eq_cr_admission_threshold} can be proved by evaluating $ \textsf{OPT}(\mathcal{I})/\textsf{TOS}_{\boldsymbol{\lambda}}(\mathcal{I}) $ over all possible worst-case scenarios. In what follows  we use the term ``\textit{total product}" to denote the total produced/sold units under $ \textsf{TOS}_{\boldsymbol{\lambda}} $ in the end. We now prove that the competitive ratio of an admission threshold is determined by the worst-case of the following four scenarios.
	
	\textit{Scenario-I}: the total product is less than or equal to $ \tau $. In this scenario, we have $ \textsf{OPT}(\mathcal{I})/\textsf{TOS}_{\boldsymbol{\lambda}}(\mathcal{I}) = 1$ since both online and offline decisions select all the buyers available.
	
	\textit{Scenario-II}: the total product is $ \tau+1 $. Let us consider the following arrival instance 
	%\vspace{-0.1cm}
	\begin{equation*}%\vspace{-0.1cm}
		\mathcal{I}^{(\tau+1)} = \Big\{\underbrace{p_{\min},\ \cdots,\ p_{\min}}_{\tau+1},\ \underbrace{\lambda_{\tau+1}-\epsilon,\ \cdots,\ \lambda_{\tau+1}-\epsilon}_{k}\Big\},
	\end{equation*} 
	where the first $ \tau+1 $ identical buyers all offer the same price $ p_{\min} $ and the remaining $ k $ identical buyers offer $ \lambda_{\tau+1}-\epsilon $. Here, $ \epsilon $ is a small positive value which is arbitrarily close to 0, i.e., $ \epsilon \rightarrow 0^+ $.  Under $ \mathcal{I}^{(\tau+1)} $, $ \textsf{TOS}_{\boldsymbol{\lambda}} $ selects the first $ \tau+1 $ buyers and rejects all the remaining buyers as their prices are less than $ \lambda_{\tau+1} $. Thus,  the profit achieved by $ \textsf{TOS}_{\boldsymbol{\lambda}} $ is 
	$ \textsf{TOS}_{\boldsymbol{\lambda}}(\mathcal{I}^{(\tau+1)}) = p_{\min}(\tau+1) - \sum_{i=1}^{\tau+1}c_i = \sum_{i=0}^\tau (\lambda_i - c_{i+1}) $, where we use the fact that $ \lambda_0 = \lambda_1 =\cdots = \lambda_\tau = p_{\min} $.
	
	In the offline setting, the optimal profit in hindsight is to reject the first $ \tau + 1 $ buyers but accept buyers with price $ \lambda_{\tau+1} -\epsilon $ only, that is,  $
	\textsf{OPT}(\mathcal{I}^{(\tau+1)}) = \max_{i\in[k]}\ (\lambda_{\tau+1}-\epsilon) i - f(i) = f^{*}(\lambda_{\tau+1}-\epsilon) $, where we use the definition of $ f^* $ in Eq. \eqref{eq_f_star}.
	Thus, the ratio $ \textsf{OPT}(\mathcal{I}^{(k+1)})/\textsf{TOS}_{\boldsymbol{\lambda}}(\mathcal{I}^{(k+1)})$ is given by
	\begin{equation*}
		\frac{\textsf{OPT}(\mathcal{I}^{(\tau+1)})}{\textsf{TOS}_{\boldsymbol{\lambda}}(\mathcal{I}^{(\tau+1)}) }= \frac{f^{*}(\lambda_{\tau+1} -\epsilon)}{\sum_{i=0}^{\tau} (\lambda_i -c_{i+1})} \ \overset{\epsilon \rightarrow 0^+}{\rarrowfill{0.8cm}} \  \frac{f^{*}(\lambda_{\tau+1})}{\sum_{i=0}^{\tau} (\lambda_i -c_{i+1})},
	\end{equation*}
	where the right-hand-side convergence is due to  $ \epsilon \rightarrow 0^+ $. 
	We argue that the above ratio is indeed the worst-case scenario when the total product is $ \tau+1 $. By worst-case, we mean that the minimum profit achievable by $ \textsf{TOS}_{\boldsymbol{\lambda}} $ is $ \textsf{TOS}_{\boldsymbol{\lambda}}(\mathcal{I}^{(\tau+1)})  $, and the maximum profit achievable in the offline setting is $ \textsf{OPT}(\mathcal{I}^{(\tau+1)}) $ with $ \epsilon \rightarrow 0^+ $. Note that for any $ \epsilon < 0 $, $ \textsf{TOS}_{\boldsymbol{\lambda}} $ will select more than $ \tau+1 $ buyers under the arrival instance $ \mathcal{I}^{(\tau+1)} $, which thus violates our definition of \textit{Scenario-II}. Therefore, if the total product is $ \tau+1 $, it is impossible to have any $ (\tau+2)$-th buyer whose price is larger than $ \lambda_{\tau+1} $. For this reason, the maximum profit achievable in the offline setting is indeed $ f^*(\lambda_{\tau+1}-\epsilon) $ with $ \epsilon \rightarrow 0^+ $.
	
	\textit{Scenario-III}: the total product is larger than $ \tau + 1 $ but less than $ \bar{k} $. Consider that the total product is $ \tau+j $ for  $ j = 2,3,\cdots,\bar{k}-\tau-1 $. Similar to \textit{Scenario-II}, we consider the following arrival instance $ \mathcal{I}^{(\tau+j)} $:
	%\vspace{-0.2cm}
	\begin{equation*}%\vspace{-0.1cm}
		\mathcal{I}^{(\tau+j)} = \Big\{\underbrace{p_{\min},\ \cdots,\ p_{\min}}_{\tau+1},\ \lambda_{\tau+1},\ \lambda_{\tau+2},\ \cdots,\ \lambda_{\tau+j-1}, \ \underbrace{\lambda_{\tau+j}-\epsilon,\ \cdots,\ \lambda_{\tau+j}-\epsilon}_{k}\Big\}, 
	\end{equation*}
	where the first $ \tau+1 $ identical buyers offer price $ p_{\min} $, the $ (\tau+j) $-th buyer offers price $ \lambda_{\tau+j-1} $, and the remaining $ k $ identical buyers all offer price $ \lambda_{\tau+j}-\epsilon $. Here, $ \epsilon $ is defined in the same way as \textit{Scenario-II}, namely, $ \epsilon \rightarrow 0^+ $.  We can prove similarly that the ratio $ \textsf{OPT}(\mathcal{I}^{(\tau+j)})/\textsf{TOS}_{\boldsymbol{\lambda}}(\mathcal{I}^{(\tau+j)})$  is captured by $ \frac{f^{*}(\lambda_{\tau+j})}{\sum_{i=0}^{\tau+j-1}(\lambda_i - c_{i+1})} $ for $ j = 2,3,\cdots,\bar{k}-\tau-1 $.
	
	\textit{Scenario-IV}: the total product is $ \bar{k} $. Let us consider the following arrival instance $ \mathcal{I}^{(\bar{k})} $:
	%\vspace{-0.1cm}
	\begin{equation*}%\vspace{-0.2cm}
		\mathcal{I}^{(\bar{k})} = \Big\{\underbrace{p_{\min},\ \cdots,\ p_{\min}}_{\tau+1},\ \lambda_{\tau+1},\ \lambda_{\tau+2},\ \cdots,\ \lambda_{\bar{k}-1}, \ \underbrace{p_{\max},\ \cdots,\ p_{\max}}_{k} \Big\},
	\end{equation*}
	Under $ \mathcal{I}^{(\bar{k})} $, similar to our previous analysis, the ratio between the maximum offline profit and the minimum online profit is:
	\begin{equation*}
		\frac{\textsf{OPT}(\mathcal{I}^{(\bar{k})})}{\textsf{TOS}_{\boldsymbol{\lambda}}(\mathcal{I}^{(\bar{k})})} = \frac{\max_{i\in[k]}\ p_{\max} i - f(i)}{ p_{\min}(\tau+1) + \sum_{i=\tau+1}^{\bar{k}-1} \lambda_i - \sum_{i=0}^{\bar{k}-1}c_{i+1}} = \frac{f^*(p_{\max})}{\sum_{i=0}^{\bar{k}-1} (\lambda_i - c_{i+1})}.
	\end{equation*}  
	Here in \textit{Scenario-IV} we do not  construct $ \mathcal{I}^{(\bar{k}-1)} $ based on $ p_{\max} - \epsilon $ with $\epsilon \rightarrow 0^+$ since  $ \bar{k} $ is the maximum number of buyers that should be accepted. 
	
	Combining the discussion of the above four scenarios, Proposition \ref{theorem_CR_threshold} follows. 
\end{proof}
%\vspace{-0.1cm}

Proposition \ref{theorem_CR_threshold} provides a way of analyzing the competitive ratio of $ \textsf{TOS}_{\boldsymbol{\lambda}} $ with an arbitrary admission threshold $ \boldsymbol{\lambda} $. A direct application of Proposition \ref{theorem_CR_threshold} is given in 
Corollary \ref{theorem_OPD_inequality} below, which summarizes a group of $ \alpha $-parameterized inequalities that once satisfied, $ \textsf{TOS}_{\boldsymbol{\lambda}} $ is $ \alpha $-competitive.

\begin{corollary}[\textsc{Sufficient Inequalities}]
	\label{theorem_OPD_inequality}
	{\normalfont $\textsf{TOS}_{\boldsymbol{\lambda}}$} is $ \alpha $-competitive  if  $\boldsymbol{\lambda} = \{\lambda_0,\lambda_1,\cdots, \lambda_\tau, \cdots, \lambda_{\bar{k}}\}$ is an admission threshold with turning point $ \tau \in \mathcal{T}^{(\alpha)} $ and $ \lambda_{\bar{k}} = p_{\max} $, and  satisfies:
	%\vspace{-0.1cm}
	\begin{equation}\label{eq_OPD_sufficiency}%\vspace{-0.2cm}
		%& p_{\min}(k+1) - f(k+1) \geq \frac{1}{\alpha} f^{*}(p_{k+1})\\
		\sum_{j=0}^i \big(\lambda_j - c_{j+1}\big) \geq \frac{1}{\alpha} f^{*}(\lambda_{i+1}), \quad \forall i\in \{\tau,\tau+1,\cdots,\bar{k}-1\}.
	\end{equation}
\end{corollary}

We skip the proof of Corollary \ref{theorem_OPD_inequality} as it largely follows Proposition \ref{theorem_CR_threshold}\footnote{Corollary \ref{theorem_OPD_inequality} can also be proved by the online primal-dual approach. For more details, please refer to Appendix \ref{appendix_OPD}.}. The fact that the turning point $ \tau $ must be an integer within $ \mathcal{K}^{(\alpha)} $ is based on  Proposition \ref{theorem_k_alpha}. The following remark provides a  geometric and economic interpretation of the sufficient inequalities in Eq. \eqref{eq_OPD_sufficiency}.

%\vspace{-0.1cm}
\begin{remark}%[\textsc{Geometric and Economic Interpretation}]
	\label{remark_geometric_economic}
	Geometrically, the value of $ \lambda_i - c_{i+1} $ in Eq. \eqref{eq_OPD_sufficiency} equals the areas of the $ (i+1) $-th grey rectangle illustrated in Fig. \ref{fig_three_cases}. In this regard, Corollary \ref{theorem_OPD_inequality} argues that, if for any $ i\in \{\tau, \tau+1,\cdots, \bar{k}-1\} $, the total areas of the first $ i+1 $ grey rectangles (from left to right) is no less than $ \frac{1}{\alpha} $ fraction of $ f^{*}(\lambda_{i+1}) $, then $ \normalfont \textsf{TOS}_{\boldsymbol{\lambda}} $ is $ \alpha $-competitive. From an economic perspective, if we consider $ \boldsymbol{\lambda} $ as a sequence of ``reserved prices", then $ \lambda_i - c_{i+1} $ denotes the ``reserved profit" of producing the $ i $-th unit. An $ \alpha $-competitive threshold policy requires that an additional unit be produced only if a minimum of profit is reserved at each step.
\end{remark}
%\vspace{-0.1cm}

Corollary \ref{theorem_OPD_inequality} is key to the upcoming analysis. Now the problem reduces to designing an admission threshold $ \boldsymbol{\lambda} $ so that the sufficient inequalities in Eq. \eqref{eq_OPD_sufficiency} are satisfied with an $ \alpha $ that is as small as possible. 
Mathematically, we need to solve the following optimization problem
%\vspace{-0.1cm}
\begin{equation}\label{eq_optimization}%\vspace{-0.1cm}
	\normalfont
	\underset{\alpha, \tau, \lambda_{\tau+1},\cdots, \lambda_{\bar{k}}}{\textsf{minimize}}\ \  \alpha \quad   \textsf{subject to}\   \tau\in \mathcal{T}^{(\alpha)}, \text{ Eq. }\eqref{eq_admission_threshold}, \text{ and} \text{ Eq. } \eqref{eq_OPD_sufficiency}.
\end{equation}
It is challenging, however, to directly solve Problem \eqref{eq_optimization} as it involves a group of non-convex and mixed-integer constraints. As one of the major technical contributions of this paper, the next section is dedicated to identifying the structural property of Problem \eqref{eq_optimization} and characterizing the existence and uniqueness of the optimal admission threshold.

\section{Main Results: Optimal Threshold, Lower Bound, and Asymptotic Properties} 
\label{section_main_results}

This section presents the main results of this paper. We first prove in Section \ref{section_optimal_threshold} that there exists a unique optimal admission threshold, denoted by $ \boldsymbol{\lambda}^* $, so that $ \textsf{TOS}_{\boldsymbol{\lambda}^*} $ achieves the best-possible competitive ratio of all deterministic algorithms. In Section \ref{section_lower_bound}, we give a lower bound on competitive ratios that no randomized algorithm can outperform. Moreover, we prove in Section \ref{section_asymptotic_properties} that the competitive ratio of  $ \textsf{TOS}_{\boldsymbol{\lambda}^*} $ asymptotically converges to this lower bound when $ k\rightarrow +\infty $.  We illustrate our results with a case study in Section \ref{section_constant_packing_costs}, and end this section with a discussion of the impact of convex costs on online selection in Section \ref{section_upper_bound}.

%\vspace{-0.2cm}
\subsection{Optimal Threshold: Existence and Uniqueness}
\label{section_optimal_threshold}
To solve Problem \eqref{eq_optimization}, our first step is to define the following system of equations:
\begin{align}\label{eq_system_of_equations_X}
	\Big(\textsf{SoE}(\bm{\chi}^{(\tau)})\Big):\ 
	%\chi_0 = 
	\frac{f^{*}\big(\chi_1^{(\tau)}\big)}{g(\tau+1)} = \frac{f^{*}\big(\chi_2^{(\tau)}\big)-f^{*}\big(\chi_1^{(\tau)}\big)}{\chi_1^{(\tau)}-c_{\tau+2}} = \cdots =  \frac{f^{*}(p_{\max})-f^{*}\big(\chi_{\bar{k}-\tau-1}^{(\tau)}\big)}{\chi_{\bar{k}-\tau-1}^{(\tau)} - c_{\bar{k}}}.
\end{align}
Eq. \eqref{eq_system_of_equations_X} consists of $ \bar{k} - \tau -1 $ variables, denoted by $ \bm{\chi}^{(\tau)} = \{\chi_1^{(\tau)}, \chi_2^{(\tau)}, \cdots, \chi_{\bar{k}-\tau-1}^{(\tau)}\} $. In what follows we refer to the above system of equations by $ \textsf{SoE}(\bm{\chi}^{(\tau)}) $.

For a given setup $ \mathcal{S}$, $ \textsf{SoE}(\bm{\chi}^{(\tau)}) $ is well-defined as long as $ \tau $ is given. The following Proposition \ref{theorem_uniqueness} argues that $ \textsf{SoE}(\bm{\chi}^{(\tau)}) $ has a unique solution for any $ \tau\in \{0,1,\cdots,\ubar{k}-1\} $. 

\begin{proposition}\label{theorem_uniqueness}
	For any $ \tau\in \{0,1,\cdots,\ubar{k}-1\} $, there exists a unique set of $ \bar{k}-\tau-1$ positive real numbers, denoted by $ \bm{\chi}^{(\tau)} = \{\chi_1^{(\tau)}, \chi_2^{(\tau)}, \cdots, \chi_{\bar{k}-\tau-1}^{(\tau)}\} $, that satisfy the system of equations $\normalfont \textsf{SoE}(\bm{\chi}^{(\tau)}) $ characterized by Eq. \eqref{eq_system_of_equations_X}. Moreover, the solution has the following properties:
	\begin{itemize}
		\item Monotonic:  $ \chi_1^{(\tau)} \leq  \chi_2^{(\tau)} \leq \cdots \leq \chi_{\bar{k}-\tau-1}^{(\tau)} \leq p_{\max} $.
		\item Lower Bounded: $ \chi_i^{(\tau)} > c_{i+\tau+1}  $ for all $ i = [\bar{k}-\tau-1] $.
	\end{itemize}
\end{proposition}

The proof of Proposition \ref{theorem_uniqueness}, as well as the rationality of why Eq. \eqref{eq_system_of_equations_X} is related to the optimal solution to Problem \eqref{eq_optimization}, is  given in Appendix \ref{proof_of_solution_SoE}. Based on Proposition \ref{theorem_uniqueness}, we give our optimal threshold design in Theorem \ref{theorem_optimality} below.

\begin{theorem}[\textsc{Optimal Threshold}]\label{theorem_optimality}
	Given a setup $ \mathcal{S} = \{f, p_{\min}, p_{\max}, k\} $,  $\normalfont \textsf{TOS}_{\boldsymbol{\lambda}^*} $ achieves the optimal  competitive ratio of all deterministic online algorithms, denoted by $\textsf{CR}_f^*(\rho, k) $, if and only if  $\boldsymbol{\lambda}^* = \big\{\lambda_0^*, \lambda_1^*,\cdots, \lambda_\tau^*,\cdots, \lambda_{\bar{k}}^*\big\} $ is an admission threshold such that
	\begin{itemize}
		\item The lower and upper limits: $ \lambda_0^* = \lambda_1^* = \cdots = \lambda_{\tau}^* =  p_{\min} $ and $ \lambda_{\bar{k}}^* =  p_{\max} $. 
		\item The turning point $ \tau $ of the admission threshold  is given by
		\begin{align}\label{eq_k_minimum}
			\tau = g^{\textsf{inv}}\Big(\frac{f^{*}(p_{\min})}{\normalfont \textsf{CR}_f^*(\rho, k)}\Big)  - 1.
		\end{align}
		
		\item The optimal competitive ratio $ \normalfont \textsf{CR}_f^*(\rho, k) $  and $ \{\lambda_{\tau+1}^*, \lambda_{\tau+2}^*, \cdots, \lambda_{\bar{k}-1}^*, \lambda_{\bar{k}}^*\} $ satisfy:
		\begin{align}\label{eq_system_of_equations}
			\normalfont
			\textbf{(SoSE):}\quad 
			\textsf{CR}_f^*(\rho, k) =
			\frac{f^{*}\big(\lambda_{\tau+1}^*\big)}{g(\tau+1)}
			=  \frac{f^{*}\big(\lambda_{\tau+2}^*\big)-f^{*}\big(\lambda_{\tau+1}^*\big)}{\lambda_{\tau+1}^*-c_{\tau+2}} = \cdots =  \frac{f^{*}(\lambda_{\bar{k}}^*)-f^{*}\big(\lambda_{\bar{k}-1}^*\big)}{\lambda_{\bar{k}-1}^* - c_{\bar{k}}}.
		\end{align}
	\end{itemize}
\end{theorem}

We prove Theorem \ref{theorem_optimality} in Appendix \ref{sec_proof_of_theorem_optimality}. The design of the turning point $ \tau $ in Eq. \eqref{eq_k_minimum} and the system of equations characterized by Eq. \eqref{eq_system_of_equations} identify the sufficient conditions that once satisfied, $ \textsf{TOS}_{\boldsymbol{\lambda}^*} $ achieves the optimal competitive ratio of all deterministic online algorithms. For this reason, we refer to Eq. \eqref{eq_system_of_equations} as the \textbf{system of sufficient equations} (\textbf{SoSE}). We give the following remarks to explain the intuitions of Theorem \ref{theorem_optimality}. 
\begin{itemize}%[leftmargin=*]
	\item \textbf{Intuition of Eq. \eqref{eq_k_minimum}}.
	The definition of admission thresholds requires $ \lambda_{\tau+1}\geq p_{\min} $ for all $ \tau \in [\ubar{k}] $, meaning that  $ f^*(\lambda_{\tau+1}^*) \geq f^*( p_{\min})$ holds since $ f^* $ is strictly increasing. For any $ \alpha\in [1,+\infty) $, Proposition \ref{theorem_k_alpha} argues that the turning point $ \tau $ must be an integer in $ \mathcal{T}^{(\alpha)} $,  whose lower bound is $ g^{\textsf{inv}}(\frac{1}{\alpha} f^{*}(p_{\min})) - 1 $. Our result shows that for any $ \alpha\in [1,+\infty) $, the turning point $ \tau $ should be as small as possible (to be proved in Section \ref{sec_proof_of_theorem_optimality}). Since this holds for general $ \alpha \in [1,+\infty)$, it holds for $\alpha =  \textsf{CR}_f^*(\rho, k) $ as well.  Thus, Eq. \eqref{eq_k_minimum} follows.
	
	\item \textbf{Intuition of Eq. \eqref{eq_system_of_equations}}.  If the sequence of positive real numbers $ \{\lambda_{\tau+1}^*, \lambda_{\tau+2}^*,\cdots, \lambda_{\bar{k}-1}^*, \lambda_{\bar{k}}^*\} $ satisfy the \textbf{SoSE} in Eq. \eqref{eq_system_of_equations}, it can be easily shown that the equalities hold in the sufficient inequalities given by Eq. \eqref{eq_OPD_sufficiency}, which thus indicates that $ \textsf{TOS}_{\boldsymbol{\lambda}^*} $ is $ \textsf{CR}_f^*(\rho, k) $-competitive. 
	
	\item \textbf{Uniqueness of $ \boldsymbol{\lambda}^* $}. Proposition \ref{theorem_uniqueness} shows that $ \textsf{SoE}(\bm{\chi}^{(\tau)}) $ has a unique solution for any $ \tau\in \{0,1,\cdots,\ubar{k}-1\} $. When the turning point  $ \tau  $ is given by Eq. \eqref{eq_k_minimum}, it can be shown that $ \tau $ is at least $ 0 $ (when $ \textsf{CR}_f^*(\rho, k) = +\infty $), and at most $ \ubar{k} - 1 $ (when $ \textsf{CR}_f^*(\rho, k) = 1 $). Thus, the \textbf{SoSE} in Eq. \eqref{eq_system_of_equations} has a unique solution, meaning that $ \boldsymbol{\lambda}^* $ uniquely exists.
	
	\item \textbf{Computation of $ \boldsymbol{\lambda}^* $}. There is no closed-form solution for Eq. \eqref{eq_system_of_equations}  in general. To obtain the optimal admission threshold  $ \boldsymbol{\lambda}^* $, one needs to solve the \textbf{SoSE} in Eq. \eqref{eq_system_of_equations} by  searching over the one-dimensional space of $ \textsf{CR}_f^*(\rho, k) \in [1,+\infty) $. Note that this can be performed in an offline fashion by various off-the-shelf numerical algorithms such as bisection searching, and thus there is no complexity issue here.
\end{itemize}

%\vspace{-0.2cm}
\subsection{Lower Bound for Randomized Algorithms}
\label{section_lower_bound}

In Theorem \ref{theorem_optimality}, $ \textsf{CR}_f^*(\rho,k) $ is written as a function of $ \rho  $ and $ k $. Intuitively,  $ \textsf{CR}_f^*(\rho,k) $ is monotonically increasing in $ \rho\in [1,+\infty) $ and monotonically decreasing in $ k \geq 1 $. Theorem \ref{theorem_hardness_results_general} below shows that $ \textsf{CR}_f^*(\rho,k) $ is lower bounded by  $ \textsf{CR}_f^{\textsf{lb}}(\rho,k) $ for all $ \rho\in [1,+\infty) $ and $ k\geq 1 $, where $ \textsf{CR}_f^{\textsf{lb}}(\rho,k) $ is the lower bound of competitive ratios that even randomized algorithms cannot overcome.

\begin{theorem}[Lower Bound] 
	\label{theorem_hardness_results_general}
	For any given setup $ \mathcal{S} = \{f, p_{\min}, p_{\max}, k\} $, no online algorithms (possibly randomized) is $ \normalfont \big( \textsf{CR}_f^{\textsf{lb}}(\rho,k) -\epsilon\big)$-competitive for any $ \epsilon > 0 $, where $ \textsf{CR}_f^{\textsf{lb}}(\rho,k) $ is given by
	\begin{equation}\label{eq_CR_lb_general}
		\normalfont
		\textsf{CR}_f^{\textsf{lb}}(\rho,k) = \frac{p_{\min} \ubar{k} - f(\ubar{k})  }{p_{\min}\gamma^{(1)} - f(\gamma^{(1)})}, 
	\end{equation}
	where $ \gamma^{(1)} $ is a real value within $ (0, \ubar{k}] $. Specifically, let us define the right-hand-side of Eq. \eqref{eq_CR_lb_general} as $F(\gamma^{(1)})$ to indicate that it is an explicit function of $ \gamma^{(1)}\in (0, \ubar{k}] $. Together with $ \big\{ \gamma^{(\ell)} \big\}_{\ell = \{2, \cdots, \bar{k}-\ubar{k} + 2\} } $, they form a unique set of increasing positive real numbers (i.e., $ 0 < \gamma^{(1)} < \gamma^{(2)} < \cdots < \gamma^{(\bar{k}-\ubar{k} + 1)} < \gamma^{(\bar{k}-\ubar{k} + 2)} =  \bar{k}  $) that satisfy
	\begin{equation} \label{eq_beta_gamma} 
		\frac{ q^{(\ell + 1)}(\ubar{k} + \ell -1)}{\exp\Big(\frac{F(\gamma^{(1)})}{\ubar{k} + \ell -1}\gamma^{(\ell + 1)}\Big)}  - 
		\frac{q^{(\ell)}(\ubar{k} + \ell -1)}{ \exp\Big(\frac{F(\gamma^{(1)})}{\ubar{k} + \ell -1} \gamma^{(\ell)} \Big) } 
		= \int_{\gamma^{(\ell)}}^{\gamma^{(\ell+1)}} \frac{ F(\gamma^{(1)}) f'(y)}{\exp\Big(\frac{F(\gamma^{(1)})}{\ubar{k} + \ell -1}  y \Big)}  dy, \quad  \forall \ell = [\bar{k} - \ubar{k} + 1],
	\end{equation}
	where $ q^{(\ell)} = c_{\ubar{k} + \ell -1} $ for $ \ell = \{2,3,\cdots, \bar{k}-\ubar{k} + 1\}  $,  $ q^{(1)}  = p_{\min}$,  and $ q^{(\bar{k}-\ubar{k} + 2)} = p_{\max} $.
\end{theorem}
\begin{proof}
	Theorem \ref{theorem_hardness_results_general}  is proved in Appendix \ref{proof_theorem_hardness_results}.  The idea is to construct a series of ``hard" arrival instances and then prove that if there exists any randomized algorithm with a bounded competitive ratio,  Eq. \eqref{eq_beta_gamma} must have a feasible solution in variables  $ \big\{\gamma^{(\ell)} \big\}_{\forall \ell} $  in order to comply with the definition of $ \alpha $ in Eq. \eqref{equation_alpha}. We then prove that any competitive ratio smaller than $ \textsf{CR}_f^{\textsf{lb}}(\rho,k) $ will inevitably lead to the non-existence of such $ \big\{ \gamma^{(\ell)} \big\}_{\forall \ell}  $, leading to the conclusion that the competitive ratio is tightly lower bounded by $ \textsf{CR}_f^{\textsf{lb}}(\rho,k) $. 
\end{proof}

In some cases, the system of equations in Eq. \eqref{eq_beta_gamma} can be greatly simplified. For example, in the \textsc{High-Value} case when $ p_{\max}\geq p_{\min} > c_k $, we have $ \ubar{k} = \bar{k} = k $, and thus Eq. \eqref{eq_beta_gamma} reduces to having  $ \gamma^{(1)} $ as the only variable. A formal statement of this result is given in Corollary \ref{theorem_hardness_results} below.

\begin{corollary}[\textsc{Lower Bound: The High-Value Case}]
	\label{theorem_hardness_results}
	Given a setup $ \mathcal{S} $ with $ p_{\max}\geq p_{\min} > c_k $,  no online algorithm (possibly randomized) is $ \normalfont \big( F(\gamma^{(1)}) -\epsilon\big)$-competitive for any $ \epsilon > 0 $, where $ F(\gamma^{(1)}) $ is defined as the right-hand-side of Eq. \eqref{eq_CR_lb_general} and $ \gamma^{(1)} \in (0,\ubar{k}] $ is the unique root to the following equation:
	\begin{equation}\label{eq_gamma}
		\frac{p_{\max}}{\exp(F(\gamma^{(1)}))} - \frac{p_{\min}}{\exp\left(\gamma^{(1)} F(\gamma^{(1)})/k\right)} =  \int_{\gamma^{(1)}}^{k} \frac{F(\gamma^{(1)}) f'(y)}{k\exp\left(yF(\gamma^{(1)})/k\right)} dy.
	\end{equation}
\end{corollary}

\begin{proof}
	As illustrated in Fig. \ref{fig_three_cases}(a), the \textsc{High-Value} case implies $ \ubar{k} = \bar{k} = k $. Thus, Corollary \ref{theorem_hardness_results} follows Theorem \ref{theorem_hardness_results_general} by substituting $ q^{(1)} = p_{\min} $, $ q^{(2)} = p_{\max} $, and $ \gamma^{(2)} = k $ into Eq. \eqref{eq_beta_gamma}.
\end{proof}

Similar to Eq. \eqref{eq_system_of_equations}, we remark that Eq. \eqref{eq_beta_gamma} also has no closed-form solution in general --- even in the \textsc{High-Value} case with linear marginal costs (or equivalently, quadratic production costs). While finding the exact solution to Eq. \eqref{eq_beta_gamma} requires solving a system of $ \bar{k} - \ubar{k} + 1 $  equations with $ \bar{k} - \ubar{k} + 1 $ variables (i.e., $ \big\{ \gamma^{(\ell)} \big\}_{\ell = [\bar{k}-\ubar{k} + 1] } $), a numerical solution can easily be obtained via a one-dimensional bisection search over $ \gamma^{(1)}\in (0,\ubar{k}]$. 
For discussions of how to solve Eq. \eqref{eq_beta_gamma} numerically, please refer to Remark \ref{remark_computation_gamma_1} in Appendix \ref{proof_theorem_hardness_results}.

\subsection{Asymptotic Convergence of $ \textsf{CR}_f^*(\rho, k) $ and $ \textsf{CR}_f^{\textsf{lb}}(\rho,k) $}
\label{section_asymptotic_properties}

We argue that the lower bound $ \textsf{CR}_f^{\textsf{lb}}(\rho,k)  $ characterized by Theorem \ref{theorem_hardness_results_general} is tight, especially for setups with a large $ k $. This result is formally proved by Theorem \ref{theorem_asymptotic} below, which shows that $ \textsf{CR}_f^*(\rho, k) $  and $ \textsf{CR}_f^{\textsf{lb}}(\rho,k) $ are asymptotically equivalent to each other  when $ \bar{k} \rightarrow +\infty$ (recall that $ \bar{k} \rightarrow +\infty$ implies $ k\rightarrow +\infty $, but not vice versa).

\begin{theorem}[\textsc{Asymptotic Lower Bound}]\label{theorem_asymptotic}
	For any given setup $ \mathcal{S} = \{f, p_{\min}, p_{\max}, k\} $,
	$\normalfont \textsf{TOS}_{\boldsymbol{\lambda}^{\star}} $ is asymptotically optimal among all online algorithms (including those with randomization) when $ \bar{k} $ is sufficiently large, namely, 
	%\vspace{-0.1cm} 
	\begin{equation}\normalfont %\vspace{-0.1cm}
		\lim\limits_{\bar{k} \rightarrow \infty} \textsf{CR}_f^*(\rho, k) = \lim\limits_{ \bar{k} \rightarrow \infty} \textsf{CR}_f^{\textsf{lb}}(\rho,k) = \underline{\textsf{CR}}_f(\rho),
	\end{equation}
	where $ \normalfont \underline{\textsf{CR}}_f(\rho) $ is the asymptotic lower bound that depends on $ f $ and $ \rho $ only.
\end{theorem}

The proof of  Theorem \ref{theorem_asymptotic} is given in Appendix \ref{proof_of_theorem_asymptotic}. Here, we briefly explain the intuitions. When $ \bar{k} $ is sufficiently large, the importance or weight of each unit of the asset becomes arbitrarily low. Using this property, we can show that the \textbf{SoSE} in Eq. \eqref{eq_system_of_equations} asymptotically converges to an ordinary differential equation (ODE). For any given cost function $ f $ in the form of Eq. \eqref{equation_f} and fluctuation ratio $ \rho \in [1, +\infty) $, the asymptotic lower bound  $ \underline{\textsf{CR}}_f(\rho) $ can be numerically calculated by solving the resulting ODE with some boundary conditions. Due to space constraints, we defer the details to Appendix \ref{proof_of_theorem_asymptotic} as well. 

Combining the results of Theorem \ref{theorem_optimality} and Theorem \ref{theorem_asymptotic}, we claim that  the threshold policy $ \textsf{TOS}_{\boldsymbol{\lambda}^*} $ is not only an optimal deterministic algorithm, but also asymptotically optimal among all online algorithms.

\subsection{Case Study: Linear Production Costs ($ f(y) = a y $ for $ y\in [0,k] $)}
\label{section_constant_packing_costs}

In this subsection, we use linear production costs as an example to illustrate Theorem \ref{theorem_optimality}, Theorem \ref{theorem_hardness_results_general}, and Theorem \ref{theorem_asymptotic}. In  particular, we assume $ f(y) = a y $ for $ y \in [0,k] $ and $ a $ is a non-negative constant (i.e., $ a \geq 0 $). In this case,  the sequence of marginal costs are all constant, i.e., $c_1 = c_2 =\cdots = c_k =  a $ (including  $ a=0 $ as a special case). For this reason, the \textsc{High-Value} case is the only active one as $ p_{\max}\geq p_{\min} > c_1 = c_k = a $ always holds (i.e., $ \ubar{k} = \bar{k} = k $).

\paragraph{(\textbf{Illustration of Theorem \ref{theorem_optimality}})}
When $ f(y) =  a y $,  we have $ g(i) = (p_{\min}-a)i$ and  $ f^*(p) = (p-a)k $. Substituting $ g $ and $ f^* $ into the \textbf{SoSE} in Eq. \eqref{eq_system_of_equations} leads to
\begin{align}\label{eq_system_equations_constant}
	\textsf{CR}_f^*(\rho, k) = \frac{k \lambda_{\tau+1}^*}{(p_{\min}-a)(\tau+1)}
	=  \frac{k(\lambda_{\tau+2}^*-\lambda_{\tau+1}^*)}{\lambda_{\tau+1}^* -a} = \cdots =  \frac{k(\lambda_k^* - \lambda_{k-1}^*)}{\lambda_{k-1}^* - a}.
\end{align}
Solving Eq. \eqref{eq_system_equations_constant} leads to the following analytical solution:
\begin{align}\label{eq_p_m}
	\lambda_i^* = \textsf{CR}_f^*(\rho, k)\cdot\Big(1 + \frac{\textsf{CR}_f^*(\rho, k)}{k}\Big)^{i-\tau-1}\cdot \frac{\tau+1}{k}\cdot (p_{\min} - a) + a, \quad  i = \{\tau+1,\cdots, k\}.
\end{align}
Based on Eq. \eqref{eq_k_minimum} and the definition of $ g^{\textsf{inv}} $ in Eq. \eqref{eq_inverse_g}, we have 
\begin{align*}
	\tau  = g^{\textsf{inv}}\Big(\frac{f^*(p_{\min})}{\textsf{CR}_f^*(\rho, k)} \Big)  = \big\lceil \frac{k}{\textsf{CR}_f^*(\rho, k)}\big\rceil - 1.
\end{align*}
Substituting $ \tau $ into Eq. \eqref{eq_p_m}, and further using the condition of $ \lambda_k^* = p_{\max} $ (the first bullet in Theorem \ref{theorem_optimality}), we reach to the following equation of $ \textsf{CR}_f^*(\rho, k) $:
\begin{align}\label{eq_CR_f_0}
	\normalfont
	\Big(1 + \frac{\textsf{CR}_f^*(\rho, k)}{k} \Big)^{k-\lceil \frac{k}{\textsf{CR}_f^*(\rho, k)}\rceil}\cdot 	\frac{\textsf{CR}_f^*(\rho, k)}{k}\cdot \big\lceil \frac{k}{\textsf{CR}_f^*(\rho, k)}\big\rceil 
	%= \frac{p_{\max} - \sigma}{p_{\min} - \sigma}  
	= \rho(a),
\end{align}
where $ \rho(a) $ is the shifted value of the fluctuation ratio $ 
\rho $ defined  as follows:
\begin{align}
	\label{eq_rho}
	%\rho \triangleq  \frac{p_{\max}}{p_{\min}}, \qquad 
	\rho(a) \triangleq  \frac{p_{\max} - a}{p_{\min} - a}, \quad a \in [0,p_{\min}).
\end{align}
Recall that $ \rho(a) $ first appears in Eq. \eqref{CR_linear}.  Proposition \ref{theorem_uniqueness} implies that Eq. \eqref{eq_CR_f_0} has a unique  root in variable $ \textsf{CR}_f^*(\rho,k) \in  [1,+\infty) $, based on which the optimal threshold $ \boldsymbol{\lambda}^* $ readily follows Eq. \eqref{eq_p_m}. 

\paragraph{(\textbf{Illustration of Theorem \ref{theorem_hardness_results_general}})}
When $ f(y) = a y $, we have $ f'(y) = a $ and $ F(\gamma^{(1)}) = \frac{k}{\gamma^{(1)}} $.  In this case, Eq. \eqref{eq_beta_gamma} reduces to Eq. \eqref{eq_gamma} and can be written as:
\begin{align*}
	\frac{p_{\max}}{\exp(k/\gamma^{(1)})}  - \frac{p_{\min}}{\exp(1)} =  \frac{a}{\exp(1)} - \frac{a}{\exp(k/\gamma^{(1)})},
\end{align*}
which thus leads to the following solution
\begin{align}\label{eq_CR_lb_constant}
	\normalfont
	\gamma^{(1)} = \frac{k}{1 + \ln \big(\rho(a)\big)}, \quad \textsf{CR}_f^{\textsf{lb}}(\rho,k) = F(\gamma^{(1)})  = 1 + \ln \big(\rho(a)\big).
\end{align}
Based on Theorem \ref{theorem_hardness_results_general}, no online algorithm  can be  $ (1 + \ln (\rho(a)) - \epsilon) $-competitive for any $ \epsilon > 0 $. 

\paragraph{(\textbf{Illustration of Theorem \ref{theorem_asymptotic}})} 
Eq. \eqref{eq_CR_lb_constant} shows that $ \textsf{CR}_f^{\textsf{lb}}(\rho,k) $ is independent of $ k $, meaning that $ \textsf{CR}_f^{\textsf{lb}}(\rho,k) = \underline{\textsf{CR}}_f(\rho) = 1 + \ln(\rho(\sigma)) $. We emphasize that this is a special scenario as the sequence of marginal costs are constant. For general setups with increasing marginal costs, $ \textsf{CR}_f^{\textsf{lb}}(\rho,k) $ indeed depends on $ k $, and converges to $ \textsf{CR}_f^{\textsf{lb}}(\rho) $ when $ k \rightarrow +\infty$. %The details are given by Corollary  \ref{theorem_constant_packing_costs} below.

\begin{corollary}[\textsc{Linear Production Costs}]\label{theorem_constant_packing_costs}
	Given $ \mathcal{S} $ with $ f(y) = a y $ for $ a \geq 0 $, the optimal threshold $ \boldsymbol{\lambda}^* $ can be constructed by Eq. \eqref{eq_p_m}, where the optimal competitive ratio $ \normalfont \textsf{CR}_f^*(\rho, k) $ is the unique root to Eq. \eqref{eq_CR_f_0}. Meanwhile, when $ k \rightarrow +\infty$, we have
	\begin{align*}\normalfont
		\lim\limits_{k\rightarrow \infty}\  \textsf{CR}_f^*(\rho, k) = \lim\limits_{k\rightarrow \infty}\ \textsf{CR}_f^{\textsf{lb}}(\rho,k) 
		= \underline{\textsf{CR}}_f(\rho) =  1 + \ln \big(\rho(a)\big).
	\end{align*}
\end{corollary}
\begin{proof}
	Corollary  \ref{theorem_constant_packing_costs} directly follows our above analysis except the convergence of $ \textsf{CR}_f^*(\rho, k) $, whose proof is elementary. For any $ \alpha \in [1,+\infty) $, the following inequality holds when $ k $ is large:
	\begin{equation*}
		\big(1 + \frac{\alpha}{k} \big)^{k-\lceil \frac{k}{\alpha}\rceil}\cdot 	\frac{\alpha}{k}\cdot \big\lceil \frac{k}{\alpha}\big\rceil  \geq \big(1 + \frac{\alpha}{k} \big)^{k-\frac{k}{\alpha}} \geq e^{\alpha - 1} \big(1- \frac{\alpha^2}{k}\big), %\text{ when } \frac{\alpha^2}{M} < 1 \ (i.e., M \rightarrow \infty)
	\end{equation*}
	where $ e $ is the Euler's number. Substituting $ \alpha  = \textsf{CR}_f^*(\rho,k) $ into the above inequality leads to 
	\begin{equation*}
		1 + \ln \big(\rho(a)\big) \leq \textsf{CR}_f^*(\rho, k)  \leq 1 + \ln \big(\rho(a)\big) - \ln \Big(1- \frac{(\textsf{CR}_f^*(\rho, k) )^2}{k}\Big). %\leq 1+ \ln \rho - \ln\Big(1- \frac{\big(1+\ln \rho\big)^2}{M}\Big) = 1+ \ln \rho +\ln\Big(\frac{M}{M - \big(1+\ln \rho\big)^2}\Big)
	\end{equation*}
	Thus, when $ k \rightarrow +\infty $, the optimal competitive ratio $ \textsf{CR}_f^*(\rho,k) $ converges to $ 1 + \ln \big(\rho(a)\big)  $. 
\end{proof}

It has been proven by various existing studies (e.g., \cite{OKP_Zhou_2008, Tan_ORA_2020, Zhang2017}) that the logarithmic competitive ratio $ 1 + \ln(\rho(\sigma))  $ is the best possible for all online algorithms in the setting with linear $ f $ (i.e., constant marginal costs). 
Thus, Corollary \ref{theorem_constant_packing_costs} validates Theorem \ref{theorem_optimality}, Corollary \ref{theorem_hardness_results}, and Theorem \ref{theorem_asymptotic} in this special case. \textit{For general convex $ f $, we argue that Theorem \ref{theorem_optimality}, Theorem \ref{theorem_hardness_results_general}, and Theorem \ref{theorem_asymptotic} provide the first set of results that characterize the \textit{optimal competitive ratio}  $\normalfont \textsf{CR}_f^*(\rho, k) $ and the \textit{lower bound} $\normalfont \textsf{CR}_f^{\textsf{lb}}(\rho, k) $  with a proven convergence to the \textit{asymptotic lower bound} $\normalfont \underline{\textsf{CR}}_f(\rho) $ when $ \bar{k} $ is large}.

\subsection{Impact of Convex Production Costs on Online Selection}
\label{section_upper_bound}

For online selection with linear production costs (or constant marginal costs), Section \ref{section_constant_packing_costs} has demonstrated that  $ \textsf{TOS}_{\boldsymbol{\lambda}^*} $ is asymptotically $ (1 + \ln(\rho(a))) $-competitive when $ k \rightarrow +\infty $. In this subsection, we analyze the impact of non-constant marginal costs on the performance of $ \textsf{TOS}_{\boldsymbol{\lambda}^*} $, and show that (strong)  convexity of $ f $ helps improve the competitive ratio of online selection. 

\begin{proposition}%[Upper Bound for Convex Cost Functions]
	\label{theorem_convex_upper_bound}
	Given any setup $ \mathcal{S} $ with $  p_{\min} > c_{\max} $, $ \textsf{CR}_f^*(\rho, k) $ and $ \underline{\textsf{CR}}_f(\rho) $ satisfy:
	%%\vspace{-0.2cm}
	\begin{equation}\label{eq_upper_bound_convex}
		\normalfont
		\Big(1 + \frac{\textsf{CR}_f^*(\rho, k)}{k} \Big)^{k-\lceil \frac{k}{\textsf{CR}_f^*(\rho, k)} \rceil} \leq \rho\left(c_k\right), \quad \underline{\textsf{CR}}_f(\rho) \leq 1 + \ln\big(\rho(c_{\max})\big),
	\end{equation}
	where $ c_{\max} $ is the upper bound of the marginal costs, i.e., $ c_i \leq c_{\max}, \forall i\in [k]$.
\end{proposition}

Proposition \ref{theorem_convex_upper_bound} is proved in  Appendix \ref{proof_of_theorem_convex_upper_bound}. Recall that Corollary \ref{theorem_constant_packing_costs} shows that $ \underline{\textsf{CR}}_f(\rho) = 1 + \ln(\rho(a)) $ when $ f(y) = a
y $ and $ c_k = c_{\max} = a $. Thus, we can unify Corollary \ref{theorem_constant_packing_costs} and Proposition \ref{theorem_convex_upper_bound}, and give Corollary \ref{theorem_tightly_upper_bounded} below.

\begin{corollary}\label{theorem_tightly_upper_bounded}
	Given any setup $ \mathcal{S} $ with $   p_{\min} > c_{\max} $, the optimal competitive ratio $ \textsf{CR}_f^*(\rho, k) $ and its asymptotic lower bound $ \underline{\textsf{CR}}_f(\rho) $ are tightly upper bounded by their counterparts when $ f $ is linear (including $ f = 0 $ as a special case). 
\end{corollary}

Proposition \ref{theorem_convex_upper_bound} and Corollary \ref{theorem_tightly_upper_bounded} give an interesting yet counter-intuitive result: Compared to classic online  selection problems (without costs), it appears to be more challenging to solve OSCC when $ f $ is nonlinear (for both online and offline settings). Nevertheless, both  the optimal competitive ratio $ \textsf{CR}_f^*(\rho, k) $ and its asymptotic lower bound $ \underline{\textsf{CR}}_f(\rho) $ are better than their respective counterparts when $ f $ is zero or linear. 

Despite being counter-intuitive at first glance, it can be shown  mathematically that the optimal turning point $ \tau $, given by Eq. \eqref{eq_k_minimum},  is monotonic w.r.t. the convexity of $ f $: \textit{the more convex $ f $ is, the smaller the optimal turning point $ \tau $ will be, leading to a shortened horizontal segment in $ \boldsymbol{\lambda}^* $ (i.e., the part with $ \lambda_0^* = \lambda_1^* = \cdots = \lambda_\tau^* =  p_{\min} $, as illustrated in Fig. \ref{fig_three_cases}) and consequently, a better competitive ratio as fewer low-price buyers will be selected}. This explanation is formally stated by Theorem \ref{theorem_strongly_convex} below, which gives a tighter upper bound based on the strong convexity of $ f $.
\begin{theorem}\label{theorem_strongly_convex}
	Given any setup $ \mathcal{S} $ with $ p_{\min} > c_{\max} $, if $ f(y) $ is $ \mu $-strongly convex in $ y\in [0,k] $, i.e.,  $
	f(\delta y_1+ (1-\delta) y_2) \le \delta f(y_1) + (1-\delta)f(y_2) - \frac{\mu}{2}\delta(1-\delta)\Vert y_1-y_2\rVert^2,  \forall  \delta \in [0,1]$, 
	then the optimal competitive ratio $ \textsf{CR}_f^*(\rho, k) $ and its asymptotic lower bound $\underline{\textsf{CR}}_f(\rho) $ satisfy
	\begin{subequations}\label{eq_upper_bound_strong_convex}
		\begin{align}
			& \normalfont  \Big(1 + \frac{\textsf{CR}_f^*(\rho, k)}{k} \Big)^{k- \big\lceil\frac{k}{\textsf{CR}_f^*(\rho, k)}\Big(\xi\cdot \textsf{CR}_f^*(\rho, k) - \sqrt{\left(\xi\cdot \textsf{CR}_f^*(\rho, k) - 1 \right)^2 + \textsf{CR}_f^*(\rho, k) -1} \Big) \big\rceil} \leq  \rho(c_k), \label{eq_upper_bound_strong_convex_1}\\
			& \normalfont \underline{\textsf{CR}}_f(\rho)
			\leq 
			\frac{1}{2} + \frac{\zeta-1}{2\zeta - 1} \ln\big(\rho(c_{\max})\big) + \sqrt{\frac{1}{4}+\frac{\zeta-1}{2\zeta-1} \ln\big(\rho(c_{\max})\big) + \frac{\zeta^2}{(2\zeta-1)^2}\ln^2\big(\rho(c_{\max})\big)}, \label{eq_upper_bound_strong_convex_2}
		\end{align}
	\end{subequations}
	where $ \xi $ and $ \zeta $ are two constants given by:
	$$ 
	\xi = \frac{p_{\min} - f'(0)}{\mu k}\in \big[1- \frac{1}{2k},+\infty\big), \quad 
	\zeta =  \lim_{k\rightarrow +\infty} \frac{p_{\min} - f'(0)}{\mu k} \in \big [1,+\infty\big).
	$$
	Meanwhile, when $ \xi $ and $ \zeta $ approach infinity, Eq. \eqref{eq_upper_bound_strong_convex} reduces to Eq. \eqref{eq_upper_bound_convex}.
\end{theorem}

Theorem \ref{theorem_strongly_convex} is our general result  regarding the upper bounds of $ \textsf{CR}_f^*(\rho, k) $ and $ \underline{\textsf{CR}}_f(\rho) $, and the proof is given in Appendix \ref{proof_of_strongly_convex}. Eq. \eqref{eq_upper_bound_strong_convex} implies that  strong convexity pushes $ \textsf{CR}_f^*(\rho,k)$ and $ \underline{\textsf{CR}}_f(\rho) $ further down. For example, when $ \zeta = 1 $, Eq. \eqref{eq_upper_bound_strong_convex_2} is given by
%%\vspace{-0.1cm}
\begin{equation*}%%\vspace{-0.1cm}
	\normalfont \underline{\textsf{CR}}_f(\rho) \leq \frac{1}{2}\Big(1 +\sqrt{1+4\ln^2\big(\rho(c_{\max})\big)}\Big),
\end{equation*}
where the right-hand-side is strictly smaller than $ 1 + \ln(\rho(c_{\max})) $ for $ \rho(c_{\max}) \in (1,+\infty)$. When $ \mu = 0 $ (i.e., $ f $ is convex but not strongly convex such as $ f(y) = a y $), we have $ \xi = \zeta = +\infty $. In this case, Eq. \eqref{eq_upper_bound_strong_convex} reduces to Eq. \eqref{eq_upper_bound_convex}. Thus, Theorem \ref{theorem_strongly_convex} generalizes  Corollary \ref{theorem_constant_packing_costs} and Proposition \ref{theorem_convex_upper_bound}, and provides a unified  characterization of the upper bounds of
$\normalfont \textsf{CR}_f^*(\rho, k) $ and $\normalfont \underline{\textsf{CR}}_f(\rho) $.

\section{Empirical Performance Evaluation}

In this section we perform numerical simulations to evaluate the empirical performance of  $\textsf{TOS}_{\boldsymbol{\lambda}^*} $ under various families of setups and different types of arrival instances.

\subsection{Simulation Setup}

\begin{wrapfigure}{r}{0.35\textwidth}
	%\vspace{-0.1cm}
	\begin{center}
		\includegraphics[height=4.7cm]{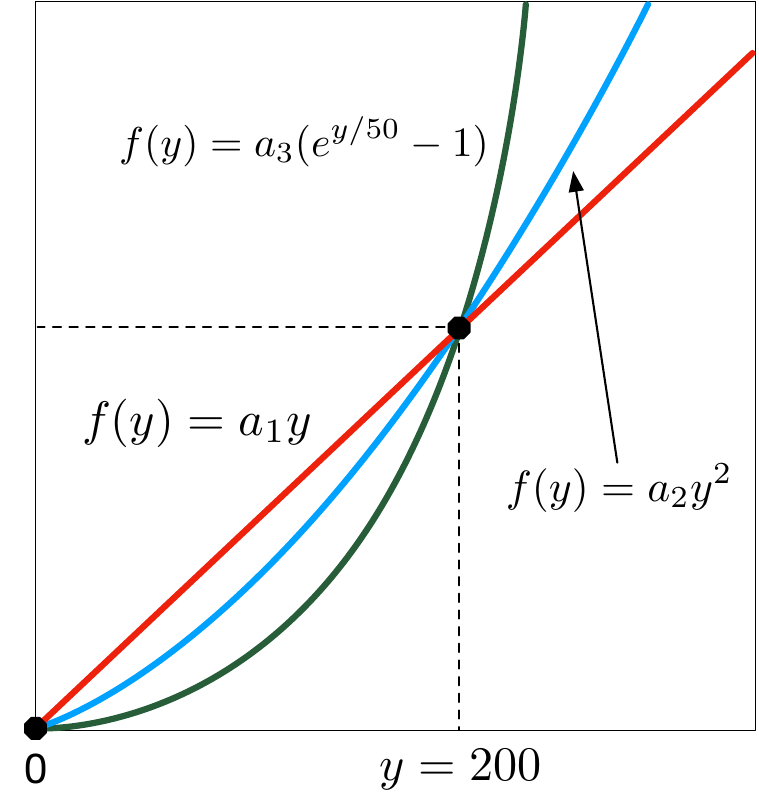}
	\end{center}
	%\vspace{-0.2cm}
	\caption{Illustration of the three cost functions used in the simulation.}
	\label{cost_functions}
	%\vspace{-0.2cm}
\end{wrapfigure}

Recall that a given setup $ \mathcal{S} = \{f, p_{\min}, p_{\max}, k\} $ fully defines the problem of OSCC. Throughout our simulation, we set $ p_{\min} = 50 $ and vary $ p_{\max} $ from 50 to 6000 so that the fluctuation ratio $ \rho $ stays within $ [1, 120] $. In most cases,  $ T $ is set to be 300, 400, 500, or 1000, and the details will be specified for each set of numerical experiments later on.

(\textbf{Cost Functions})  We focus on three common production cost functions including linear ($ f(y) = a_1y $), quadratic ($ f(y) = a_2y^2 $), and exponential ($ f(y) = a_3(e^{y/50}-1) $).  We set the coefficients $ a_1 = 40 $, $ a_2 = 0.2 $, and $ a_3 = 145.5 $ such that these three cost functions roughly meet at $ y = 200 $ (see Fig. \ref{cost_functions}), making them comparable to each other in quantities but contrasting in growth rate. One  major focus of our simulation is to  evaluate the impact of different production costs on the competitive ratio of $\textsf{TOS}_{\boldsymbol{\lambda}^*} $. In particular, we are interested in demonstrating how the increasing rate of production costs (related to the marginal costs) impact the worst-case and empirical performance of $\textsf{TOS}_{\boldsymbol{\lambda}^*} $ under different arrival instances.

(\textbf{Arrival Instances})  To simulate different arrival patterns, we construct the following three types of arrival instances (unless otherwise specified, the length of each arrival instance is $ T = 500 $):
\begin{itemize}%[leftmargin=*]
	\item Type 1: \texttt{low2high}. For arrival instances that are of type \textsf{low2high}, the prices of the first half are between $ p_{\min} $ and $ (p_{\min} + p_{\max})/2 $, and those of the second half are between   $ (p_{\min} + p_{\max})/2 $ and $ p_{\max} $. We use \texttt{low2high} to simulate the scenario when  being aggressive at the beginning (i.e., producing/selling too many units in earlier stages) may be penalized as it is likely to have sufficient number of buyers with higher prices later on. 
	
	\item Type 2: \texttt{random}. In this type of arrival instances, all the prices are randomly sampled from the real interval $ [p_{\min}, p_{\max}] $. 
	
	\item Type 3: \texttt{high2low}. This type is configured in contrast to \texttt{low2high}. We use \texttt{high2low} to simulate the scenario when being too reserved at the beginning (i.e., rejecting too many offers in earlier stages) may be penalized.
\end{itemize} 

(\textbf{Empirical Performance Metrics}) 
To evaluate the empirical performance of $ \textsf{TOS}_{\boldsymbol{\lambda}^*} $ under different input instances, we generate $ N = 1000 $ samples of arrival instances $ \mathcal{I}_n $ with $ n\in [N] $, and  calculate the empirical ratio and average empirical ratio of $ \textsf{TOS}_{\boldsymbol{\lambda}^*} $ as follows: 
%\vspace{-0.1cm}	
\begin{equation*}%\vspace{-0.1cm}
	\textsf{Empirical Ratio (ER)} \triangleq  \frac{ \textsf{OPT}(\mathcal{I}_n)}{ \textsf{TOS}_{\boldsymbol{\lambda}^*}(\mathcal{I}_n)},\   \textsf{Average Empirical Ratio (AER)} \triangleq \frac{1}{N} \sum_{n=1}^{N} \frac{ \textsf{OPT}(\mathcal{I}_n)}{ \textsf{TOS}_{\boldsymbol{\lambda}^*}(\mathcal{I}_n)},
\end{equation*}
where $\textsf{TOS}_{\boldsymbol{\lambda}^*}(\mathcal{I}_n)$ denotes the profit achieved by our threshold policy $ \textsf{TOS}_{\boldsymbol{\lambda}^*} $ in the online setting, and $\textsf{OPT}(\mathcal{I}_n)$ is the optimal profit achievable in hindsight. For each given $ \mathcal{I}_n $, we compute $\textsf{OPT}(\mathcal{I}_n)$ by solving an integer optimization problem using Gurobi.

\begin{figure} 
	\centering
	\subfigure[$ \alpha^* = \textsf{CR}_f^*(\rho, k) $ \textit{vs.} $ \rho $]{\includegraphics[height=4cm]{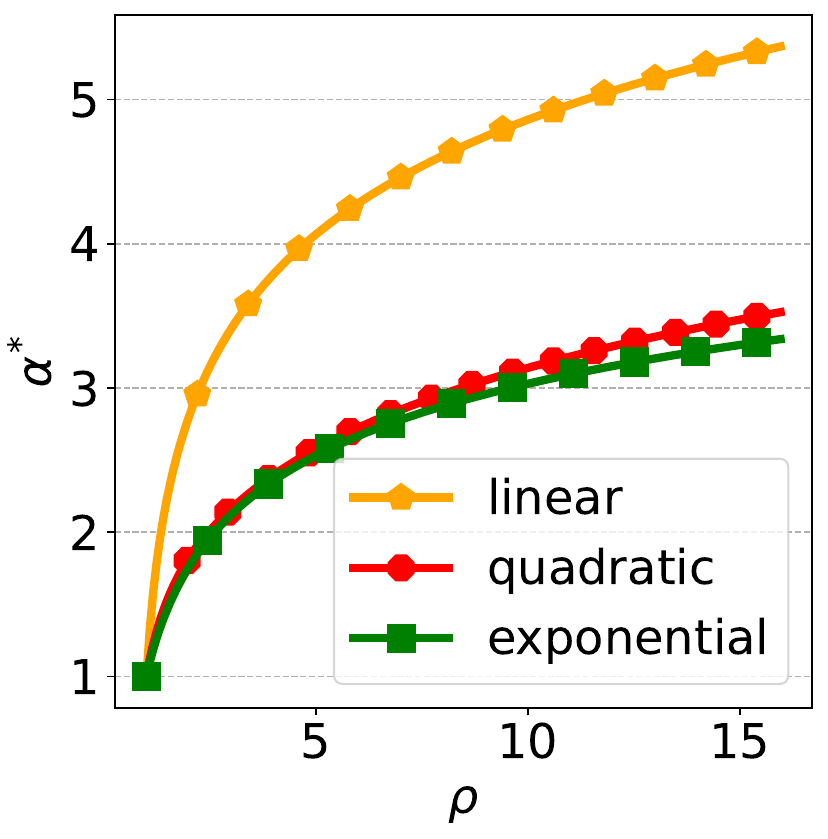}}
	\qquad 
	\subfigure[ AER \textit{vs.} $ k $ (quadratic $ f $; $ \rho = 8 $)]{\includegraphics[height=4cm]{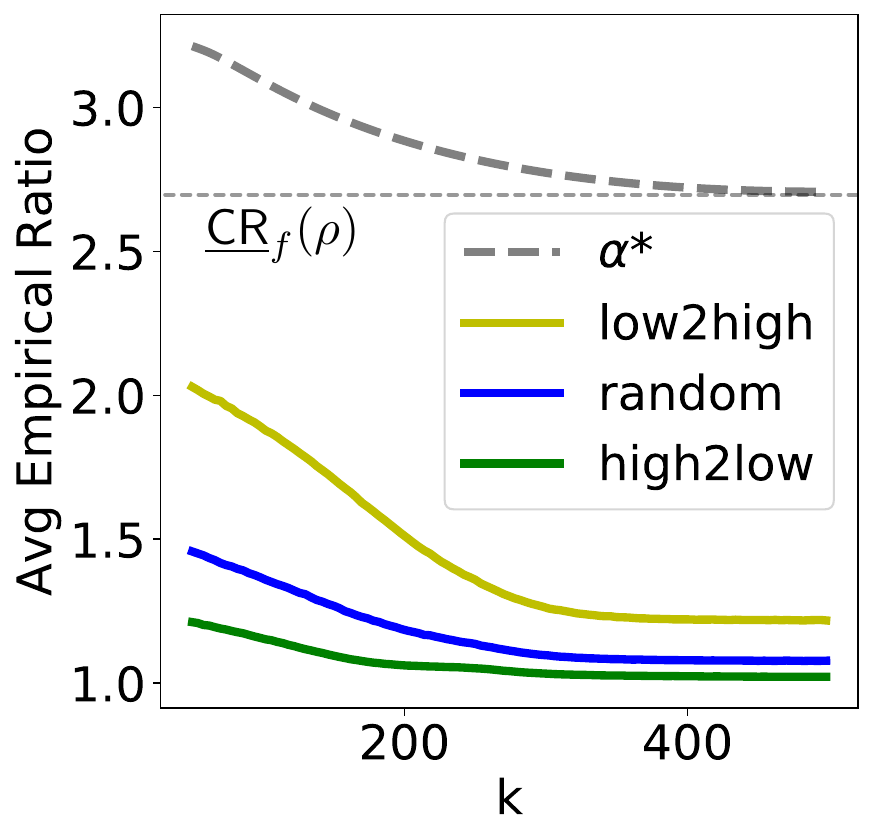}}
	\qquad
	\subfigure[AER \textit{vs.} $ \rho $]{\includegraphics[height=4cm]{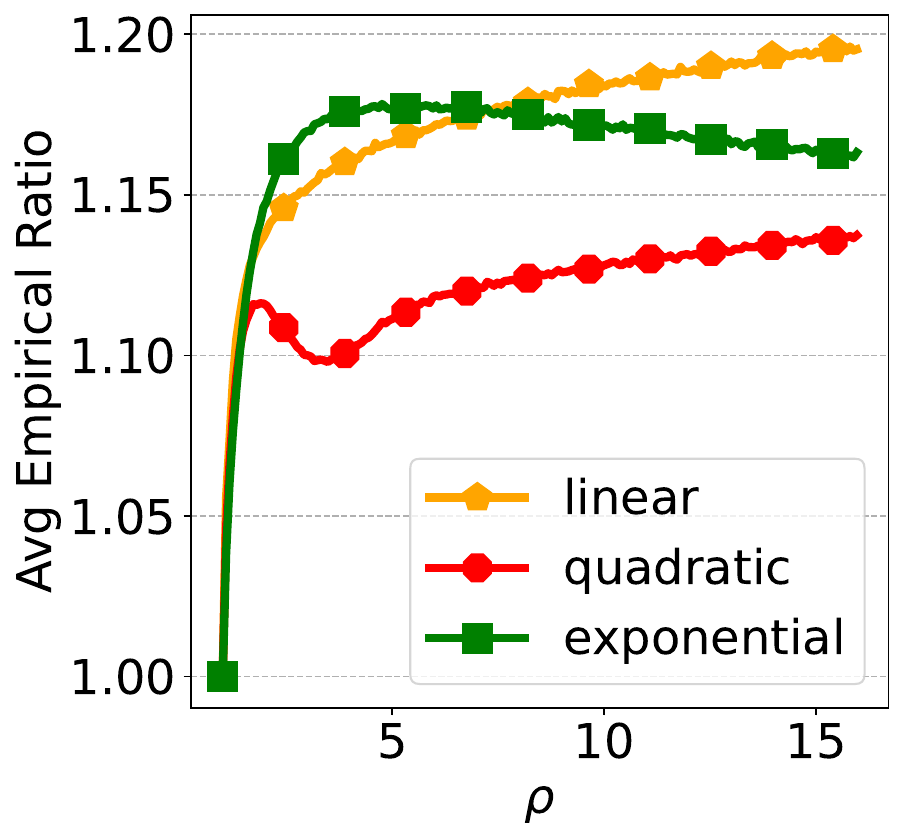}}
	%\includegraphics[width=6cm]{ER_vs_M}  
	%\vspace{-0.3cm}
	\caption{Comparison of the optimal competitive ratio $ \alpha^* =  \textsf{CR}_f^*(\rho,k) $ and the average empirical ratio (AER) of $ \textsf{TOS}_{\boldsymbol{\lambda}^*} $ under different setups.  Subfigure (a): $ \alpha^* =  \textsf{CR}_f^*(\rho,k) $ under different $ f $ (linear, quadratic, and exponential) with $ k = 300 $. Subfigure (b): AER of $ \textsf{TOS}_{\boldsymbol{\lambda}^*} $ under different types of arrival instances (\texttt{low2high}, \texttt{random}, and \texttt{high2low}) with quadratic $ f $ and $ \rho = 8 $. Other than $ \alpha^* $, each curve shows the AER of $ \textsf{TOS}_{\boldsymbol{\lambda}^*} $ over 1000 arrival instances sampled by their corresponding types. Subfigure (c): AER of $ \textsf{TOS}_{\boldsymbol{\lambda}^*} $ under different $ f $ (linear, quadratic, and exponential). Each curve shows the AER of $ \textsf{TOS}_{\boldsymbol{\lambda}^*} $ over 1000 arrival instances sampled by Type 2: \texttt{random}.  }  
	\label{fig_alpha_AER_rho_k}
	%\vspace{-0.5cm}
\end{figure}

\subsection{Numerical Results}

Fig. \ref{fig_alpha_AER_rho_k} (a) illustrates the optimal competitive ratio $ \alpha^* =  \textsf{CR}_f^*(\rho,k) $ as a function of the fluctuation ratio $ \rho \in [1, +\infty) $. For each given setup $ \mathcal{S} = (f, p_{\min}, p_{\max}, k) $, we solve the system of equations in Eq. \eqref{eq_system_of_equations} and plot $ \alpha^* =  \textsf{CR}_f^*(\rho,k) $ \textit{vs.} $ \rho \in [1, 16] $ in Fig. \ref{fig_alpha_AER_rho_k} (a). The strict monotonicity of $ \textsf{CR}_f^*(\rho,k) $ in $ \rho $ validates that higher level of price volatility leads to worse performance (i.e., a larger competitive ratio) under all production cost functions including linear, quadratic, and exponential. Moreover, the faster the cost function $ f $ grows, the better the competitive ratio $ \alpha^* = \textsf{CR}_f^*(\rho,k) $ is. For instance, the competitive ratio with linear $ f $ is strictly worse than that with exponential $ f $. Therefore, \textit{existence of faster-rising marginal costs  leads to a better and more robust performance in the worst-case}.  

Fig. \ref{fig_alpha_AER_rho_k} (b) shows the comparison between $\alpha^* = \textsf{CR}_f^*(\rho,k)$  and the empirical performance of $ \textsf{TOS}_{\boldsymbol{\lambda}^*} $ under different input instances. In Fig. \ref{fig_alpha_AER_rho_k} (b), we focus only on the quadratic $ f $ and vary $ k $ from 50 to 500. We also set the fluctuation ratio $ \rho = 8 $ and the length of input sequence $ T = 500 $. Given such a setup, the optimal competitive ratio $\alpha^* = \textsf{CR}_f^*(\rho,k)$ is roughly within [2.5, 3.2] as $ k $ varies. In comparison, the AERs of $ \textsf{TOS}_{\boldsymbol{\lambda}^*} $ are always close to 1 when the input sequences are of type \texttt{high2low}, but become considerably worse in face of \texttt{low2high}. The AER of $ \textsf{TOS}_{\boldsymbol{\lambda}^*} $  over \texttt{random} is between that of \textsf{low2high} and \texttt{high2low}, as one would expect. For all these three types, the AERs are below $\alpha^*$ since by definition  $\alpha^*$ captures the worst-case ERs. Meanwhile, Fig. \ref{fig_alpha_AER_rho_k} (b) shows that as $k$ grows, all the AERs decrease and converge to some constant (independent of $ k $), so does $ \alpha^* $. This validates Theorem \ref{theorem_asymptotic} as $ \alpha^* $ converges to  $ \underline{\textsf{CR}}_f(\rho)$ when $ k \rightarrow +\infty $.

\begin{figure} 
	\centering
	\subfigure[ER \textit{vs.} $ \rho $ (linear $ f $)]{\includegraphics[height=4cm]{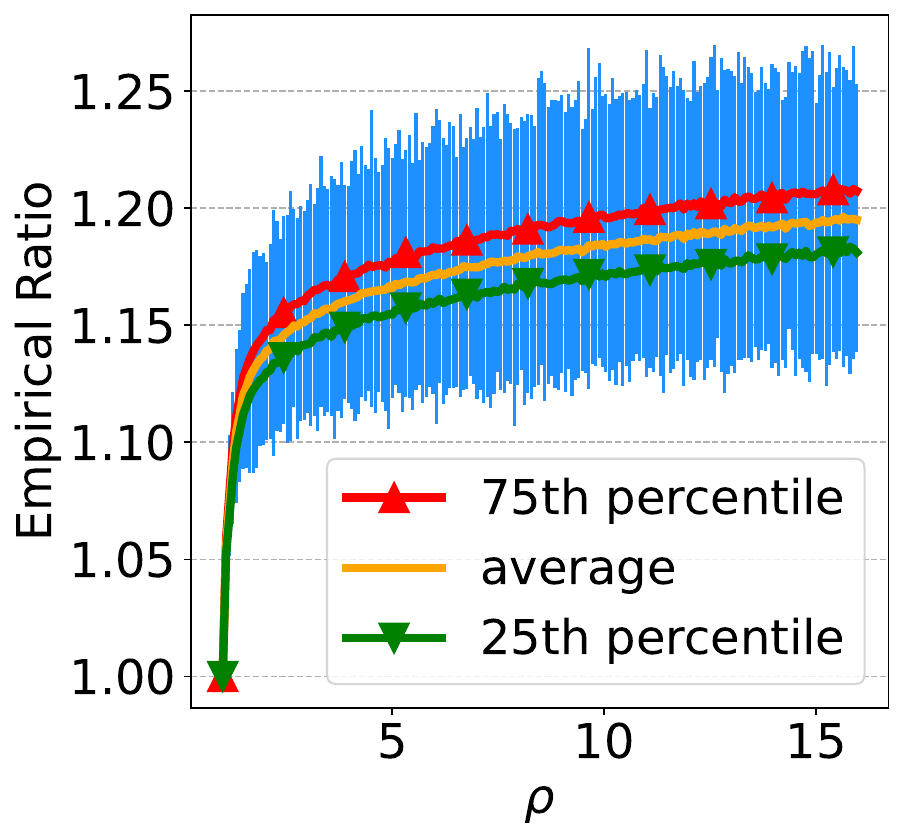}}
	\qquad
	\subfigure[ER \textit{vs.} $ \rho $ (quadratic $ f $)]{\includegraphics[height=4cm]{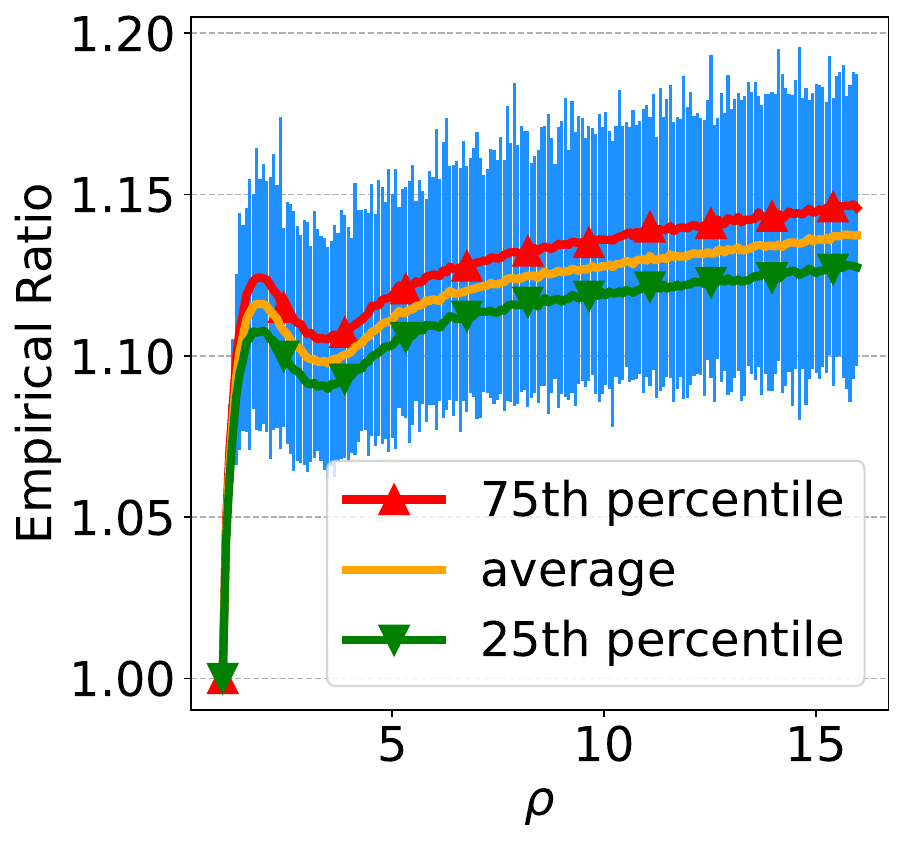}}
	\qquad
	\subfigure[ER \textit{vs.} $ \rho $ (exponential $ f $)]{\includegraphics[height=4cm]{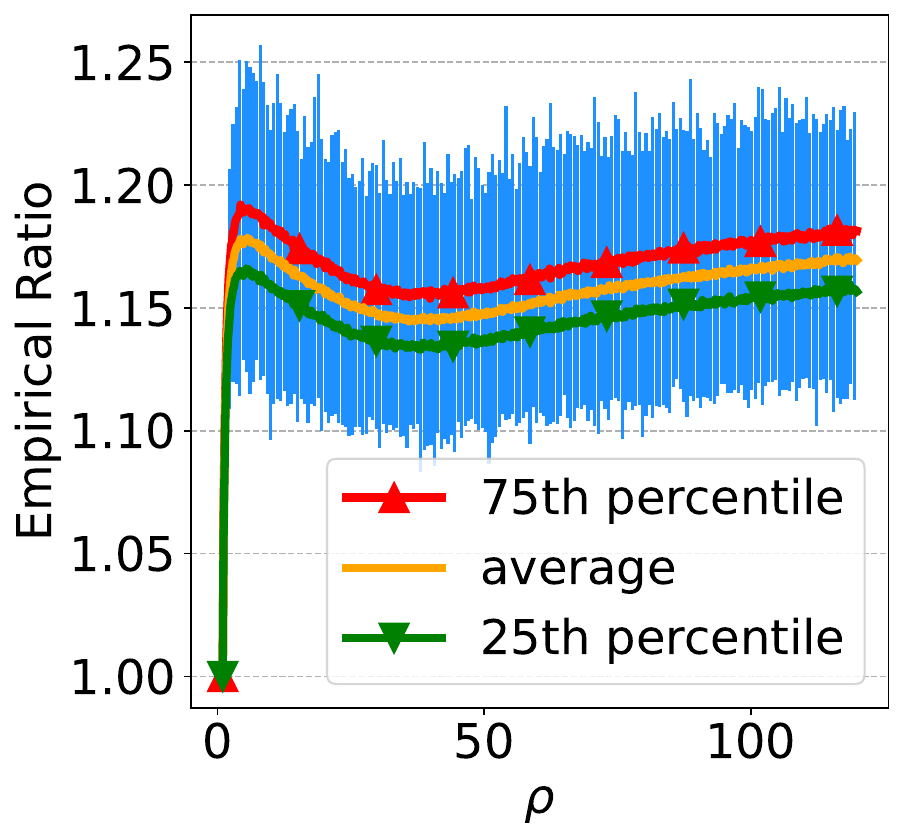}}
	
	\caption{Empirical ratio (ER) of $ \textsf{TOS}_{\boldsymbol{\lambda}^*} $ under different $ f $  with $ k = 300 $. The three real curves in each subfigure show the 25th percentile, 75th percentile, and  average of ERs over 1000 arrival instances sampled by Type 2: \textsf{random}. The scatter plots in blue show the minimum and maximum ER among the 1000 samples. }  
	\label{fig_AER_lin_qua_exp}
\end{figure}

\begin{figure} 
	\centering
	\subfigure[AER \textit{vs.} $ \hat{\rho}/\rho $ (linear $ f $)]{\includegraphics[height=4cm]{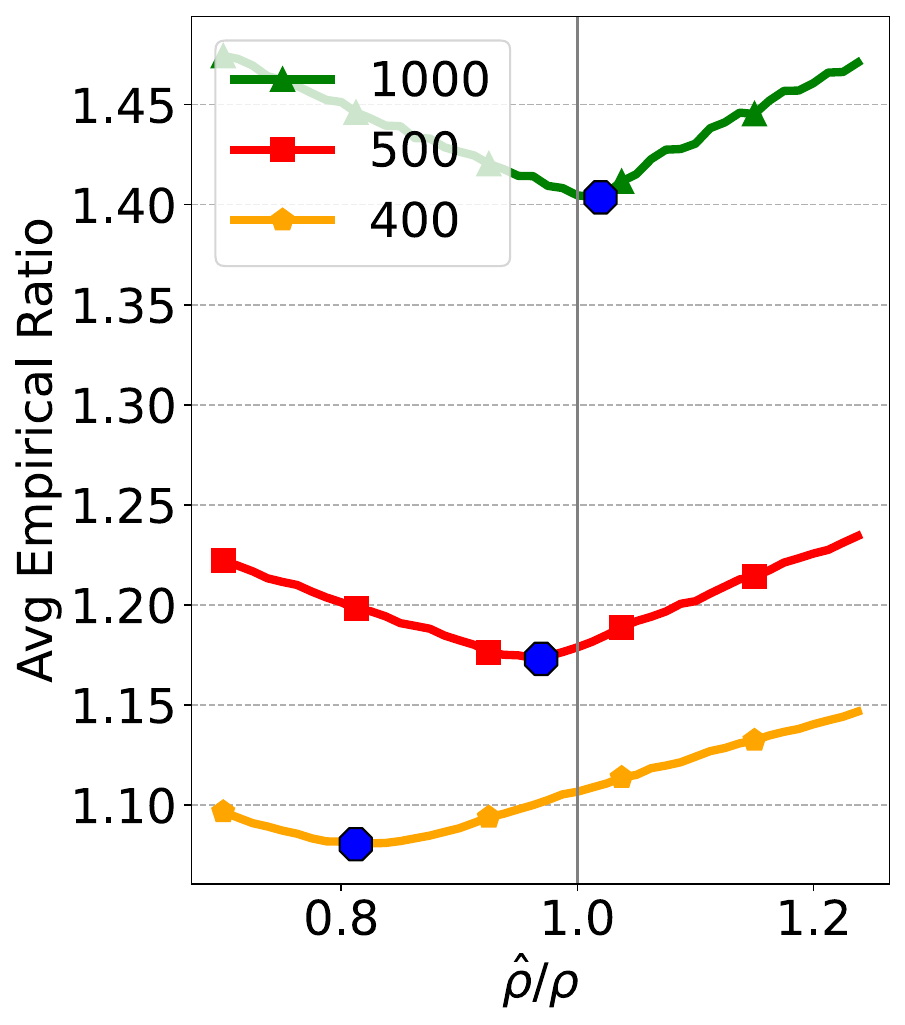}}
	\qquad \quad  
	\subfigure[AER \textit{vs.} $ \hat{\rho}/\rho $ (quadratic $ f $)]{\includegraphics[height=4cm]{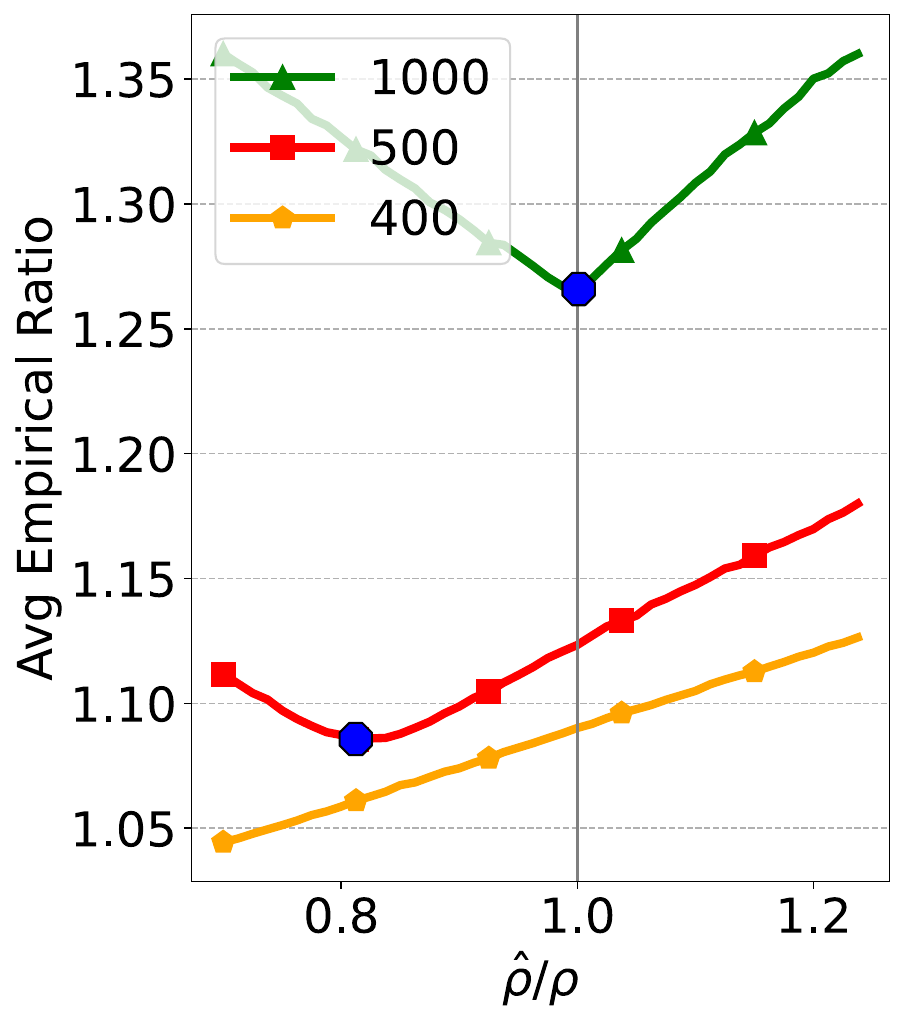}}
	\qquad \quad
	\subfigure[AER \textit{vs.} $ \hat{\rho}/\rho $ (exponential $ f $)]{\includegraphics[height=4cm]{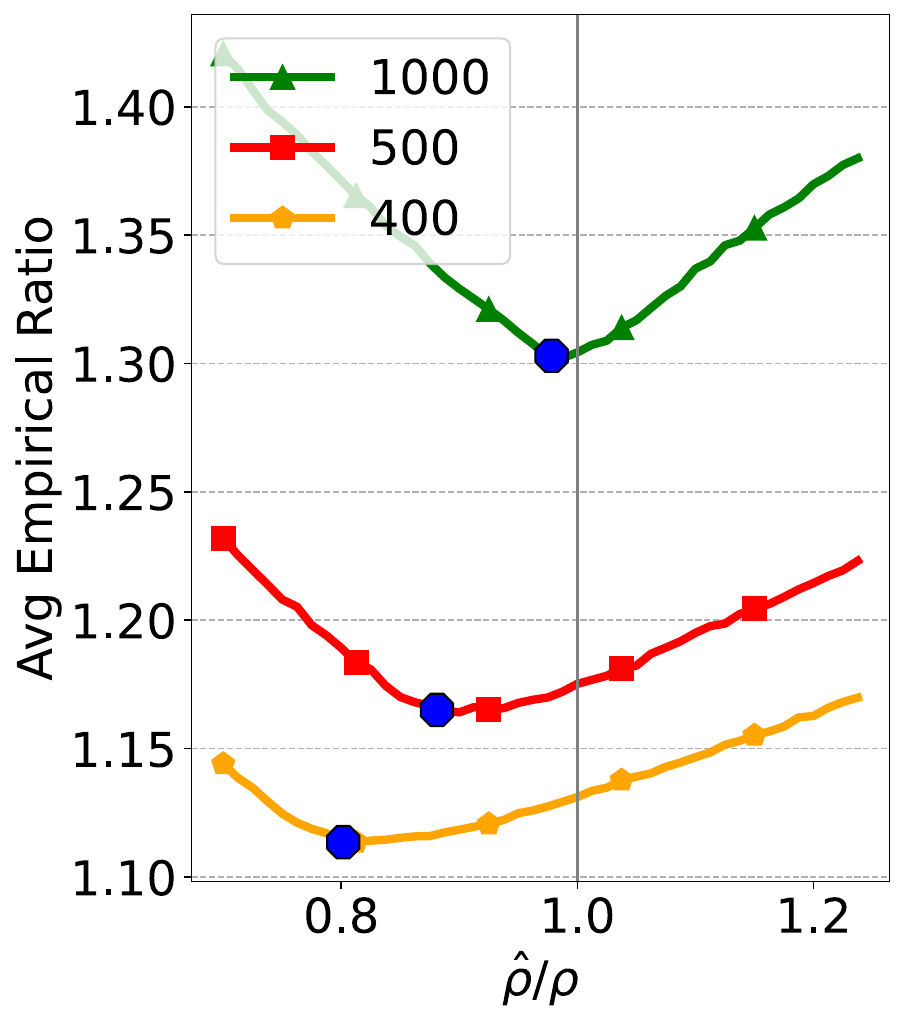}}
	
	%\vspace{-0.3cm}
	\caption{Average empirical ratio (AER) of $ \textsf{TOS}_{\boldsymbol{\lambda}^*} $ with $ k = 300 $ under different $ f $ (linear, quadratic, and exponential) when the \textit{effective fluctuation ratio} $ \hat{\rho} $ (for designing the optimal threshold $ \boldsymbol{\lambda}^* $) differs from the \textit{true fluctuation ratio} $ \rho $ (for generating the samples of arrival instances $ \mathcal{I}_n $). In each subfigure, the three curves illustrate the AER of $ \textsf{TOS}_{\boldsymbol{\lambda}^*} $ over 1000 samples of arrival instance $ \mathcal{I}_n $ sampled by Type 2: \texttt{random} with length $ T = 400, 500 $, and $ 1000 $, respectively.} 
	\label{fig_AER_rho_estimate} 
\end{figure}

In contrast to the optimal competitive ratio (i.e., the worst-case performance of $ \textsf{TOS}_{\boldsymbol{\lambda}^*} $) illustrated in Fig. \ref{fig_alpha_AER_rho_k} (a), Fig. \ref{fig_alpha_AER_rho_k} (c) shows that the empirical performance of $ \textsf{TOS}_{\boldsymbol{\lambda}^*} $ has completely distinct behaviors \textit{w.r.t.} $ \rho\in [1,+\infty) $. Specifically, we have the following observations. 
\begin{itemize}[leftmargin=*]
	\item \textit{A higher level of price volatility does not necessarily lead to a worse empirical performance}. In fact, the results in Fig. \ref{fig_alpha_AER_rho_k} (c) regarding the quadratic and exponential $ f $ show quite the opposite. For example, when $ f $ is quadratic, it is sometimes beneficial to have higher levels of price volatility as the AERs of $ \textsf{TOS}_{\boldsymbol{\lambda}^*} $ first peak at $ \rho = 1.88 $ and then decrease when $ \rho $ varies from 1.88 (local maximum) to 3.48 (local minimum). Similar observations can also be made regarding the exponential cost in Fig. \ref{fig_alpha_AER_rho_k} (c), where the curve of exponential $ f $ first peaks and then decreases w.r.t $ \rho\in [1, 15] $. To be more clear, we also plot the 25th/75th percentile of the ERs in Fig. \ref{fig_AER_lin_qua_exp} and zoom out of the curve regarding the exponential cost in Fig. \ref{fig_AER_lin_qua_exp} (c) with a wider range of fluctuation ratio $ \rho\in [1,120] $. As shown in Fig. \ref{fig_AER_lin_qua_exp} (b) and Fig. \ref{fig_AER_lin_qua_exp} (c),  the empirical performance of $ \textsf{TOS}_{\boldsymbol{\lambda}^*} $ indeed improves when $ \rho $ is roughly within $ [5, 40] $ and then deteriorates when $ \rho $ further increases.
	
	\item \textit{A faster-rising marginal cost does not necessarily lead to a worse empirical performance}. As discussed earlier, Fig. \ref{fig_alpha_AER_rho_k} (a) shows that existence of  faster-rising marginal costs leads to better competitive ratios, and thus a more robust performance in the worst case. However, no similar definitive conclusion can be made regarding the empirical performance. In fact, combining the results illustrated in Fig. \ref{fig_alpha_AER_rho_k} (c) and Fig. \ref{fig_AER_lin_qua_exp} (a)-(c), we can see that the empirical performance of $ \textsf{TOS}_{\boldsymbol{\lambda}^*} $ has much more complex behaviors w.r.t. different cost functions. For example, Fig. \ref{fig_alpha_AER_rho_k} (c) shows that the AERs under exponential $ f $ are, quite surprisingly, smaller than those under linear $ f $ for $ \rho > 7.5$. Similarly, it is interesting to see that the AER under quadratic $ f $ is always the best of the three.  We argue that such counter-intuitive results are due to the fundamental flaw of worst-case analysis and more research is needed to study the impact of different production costs on online selection under various types of input instances.
\end{itemize}

Recall that the fluctuation ratio $ \rho $ plays a pivotal role in the design of the optimal threshold $ \boldsymbol{\lambda}^* $ as well as in shaping the arrival instance $ \mathcal{I}_n $. Thus, an intriguing question is: what if the \textit{effective fluctuation ratio} $ \hat{\rho} $ for designing  $ \boldsymbol{\lambda}^* $ differs from the \textit{true fluctuation ratio} $ \rho $ for generating the input instance $ \mathcal{I}_n $? To answer this question, Fig. \ref{fig_AER_rho_estimate} visualizes our numerical results regarding the empirical performance of $ \textsf{TOS}_{\boldsymbol{\lambda}^*} $ w.r.t. $ \hat{\rho}/\rho $. Specifically,  $ \hat{\rho}/\rho > 1 $ indicates an ``\textit{overestimation}" of the true fluctuation ratio $ \rho $ and conversely, $ \hat{\rho}/\rho < 1 $ implies ``\textit{underestimation}."  Fig. \ref{fig_AER_rho_estimate} shows that in general it is increasingly more beneficial to underestimate $ \rho $ when we have a shorter arrival instance (i.e., fewer buyers). An intuitive explanation is that, when there are fewer buyers to arrive, it is better to be more aggressive or greedy (i.e., using a smaller $ \hat{\rho} $ to design $ \boldsymbol{\lambda}^* $) since this will help produce/sell more units as early as possible; otherwise, the seller may risk of having no buyers in the end. However, depending on different setups, overestimation may also lead to better empirical performance. For example, Fig. \ref{fig_AER_rho_estimate} (a) shows that it is better to slightly overestimate $ \rho $ when the length of the input instance is 1000.

\section{Conclusions and Future Work}

We initiated the study of online selection with convex cost (OSCC) and presented a class of simple threshold policies with tight guarantees for various setups of OSCC. Specifically, we developed a threshold-based online selection algorithm that achieves the optimal competitive ratio of all deterministic algorithms. We also derived a tight, setup-dependent lower bound on competitive ratios for randomized algorithms, and proved that the competitive ratio of our threshold policy converges to this lower bound when the capacity $ k $ is large. Our results generalize various classic online search, pricing, and auction problems, and broadly extend the applicability of these modes to real-world online resource allocation problems involving production or operating costs.

One intriguing question left for future research is whether there exists any randomized algorithm that can achieve a competitive ratio matching our derived lower bound for finite $ k $, especially in the presence of arbitrary increasing marginal costs? It is also interesting to investigate whether our results  can be extended to more complex settings such as online selection with multi-unit demand and/or multiple types of resources. Finally, it is a natural next step to investigate how to leverage the recent progress on learning-augmented online optimization to overcome the overly-robust nature of worst-case analysis. 

\bibliographystyle{plain}
\bibliography{OSCC_bibliography}

\newpage
\appendix

\section{Proof of Proposition \ref{theorem_uniqueness}}
\label{proof_of_solution_SoE}
Before proving Proposition \ref{theorem_uniqueness}, we first give the following lemmas. 

\begin{lemma}\label{theorem_OPD_sufficiency_proof}
	$ \normalfont\textsf{TOS}_{\boldsymbol{\lambda}} $ is $ \alpha $-competitive if the  admission threshold $ \boldsymbol{\lambda} $ satisfies the following conditions:
	\begin{subequations}
		\begin{align} \label{eq_OPD_inequalities_proof_1}
			& p_{\min}(\tau+1) - f(\tau+1) \geq \frac{1}{\alpha} f^{*}(\lambda_{\tau+1}),\\ \label{eq_OPD_inequalities_proof_2}
			& \lambda_i - c_{i+1}
			\geq  \frac{1}{\alpha}\Big(f^{*}(\lambda_{i+1}) - f^{*}(\lambda_i)\Big), \forall i\in \{\tau+1, \cdots, \bar{k}-1\},\\ \label{eq_OPD_inequalities_proof_3}
			& \lambda_{\bar{k}} = p_{\max}.
			%& p_0 = p_1 = \cdots  = p_k =  p_{\min} \leq p_{k+1} \leq p_{k+2}\leq \cdots \leq p_{\bar{k}-1} \leq p_{\bar{k}} = p_{\max}.
		\end{align}
	\end{subequations}
\end{lemma}
\begin{proof}
	Eq. \eqref{eq_OPD_inequalities_proof_1} implies that $ p_{\min}(\tau+1) - f(\tau+1) = g(\tau+1) \geq  \frac{1}{\alpha} f^{*}(\lambda_{\tau+1}) \geq \frac{1}{\alpha} f^{*}(p_{\min})$, where $ \lambda_{\tau+1} \geq p_{\min}$ is because $ \boldsymbol{\lambda} $ is an admission threshold (i.e., Eq. \eqref{eq_admission_threshold}). Thus, the turning point $ \tau $ is guaranteed to be within the discrete set of $ \mathcal{T}^{(\alpha)} $ as long as Eq. \eqref{eq_OPD_inequalities_proof_1} is satisfied. 
	
	For each $ i \in \{\tau+1, \cdots, \bar{k}-1\}$,  combining Eq. \eqref{eq_OPD_inequalities_proof_1} and Eq. \eqref{eq_OPD_inequalities_proof_2} via telescoping sum leads to
	\vspace{-0.2cm}
	\begin{equation*}\vspace{-0.1cm}
		\sum_{j=0}^i \big(\lambda_j - c_{j+1}\big) \geq \frac{1}{\alpha} f^{*}(\lambda_{i+1}), \quad \forall i\in \{\tau, \tau+1,\cdots,\bar{k}-1\},
	\end{equation*}
	where we use the fact that $ \lambda_0 = \lambda_1 = \cdots  = \lambda_\tau =  p_{\min} $ and $ f(\tau+1) = \sum_{j=1}^{\tau+1} c_j $.  Thus, an admission threshold $ \boldsymbol{\tau} $ that satisfies Eq. \eqref{eq_OPD_inequalities_proof_1} -- Eq. \eqref{eq_OPD_inequalities_proof_3} indicates that the sufficient inequalities in Corollary \ref{theorem_OPD_inequality} hold, and thus $\textsf{TOS}_{\boldsymbol{\lambda}}$ is $ \alpha $-competitive. 
\end{proof}

We remark that any admission threshold $ \boldsymbol{\lambda} $ that satisfies Eq. \eqref{eq_OPD_inequalities_proof_1} -- Eq. \eqref{eq_OPD_inequalities_proof_3} implies that it also satisfies Eq. \eqref{eq_OPD_sufficiency}, but not vice versa. Thus, Lemma \ref{theorem_OPD_sufficiency_proof} provides a set of more restricted sufficient conditions than Corollary \ref{theorem_OPD_inequality}. 

As can be seen from Eq. \eqref{eq_OPD_inequalities_proof_2}, the admission threshold satisfies a certain recursive inequality parameterized by $ \alpha $. It is known that such recursive equations can be characterized by first-order difference equations. The following Lemma \ref{theorem_difference_equation} shows the existence of a unique solution for general first-order difference equations. 

\begin{lemma}\label{theorem_difference_equation}
	For a given $ z_0 $, every first-order difference equation $ z_t = H(z_{t-1},t) $ has a unique solution with the initial value given by $ z_0 $ at $ t = 0 $.
\end{lemma}
\begin{proof}
	This is a standard result in first-order linear/nonlinear difference equations.  
\end{proof}	

Based on Lemma \ref{theorem_difference_equation}, we give Lemma \ref{theorem_difference_equation_X} below, which plays a key role in our following proof.

\begin{lemma}\label{theorem_difference_equation_X}
	For each given  $ \tau \in \{0,1,\cdots,\ubar{k}-1\} $, if $ \alpha_{\textsc{is}}^{(\tau)} > 0 $, then the following \textbf{first-order nonlinear difference equation} has a unique solution with terminal value $ \chi_{\bar{k}-\tau}^{(\tau)} = p_{\max} $:
	\begin{align}\label{eq_nonliner_difference_equation}%\color{blue}
		f^{*}\left(\chi_i^{(\tau)}\right) = f^{*}\left(\chi_{i-1}^{(\tau)}\right) + \alpha_{\textsc{is}}^{(\tau)}  \chi_{i-1}^{(\tau)} - \alpha_{\textsc{is}}^{(\tau)} c_{\tau + i}, \quad i\in \{2, 3, \cdots,\bar{k}-\tau\}.
	\end{align}	
	Meanwhile, the solution has the following properties:
	\begin{itemize}
		\item $ \chi_{i}^{(\tau)} > \chi_{i-1}^{(\tau)} $, for all $ i\in \{2, 3, \cdots,\bar{k}-\tau\} $.
		\item $ \chi_i^{(\tau)} > c_{\tau+i+1} $, for all $ i\in \{2, 3, \cdots,\bar{k}-\tau\} $.
	\end{itemize}
\end{lemma}
\begin{proof}
	Here, the subscript ``IS" of $ \alpha_{\textsc{is}}^{(\tau)} $ means ``increasing segment." In the following we also use ``HS" to represent ``horizontal segment." The  terms of ``increasing" and ``horizontal" are because $\alpha_{\textsc{is}}^{(\tau)} $ and $ \alpha_{\textsc{hs}}^{(\tau)} $ (to be defined in Eq. \eqref{eq_def_IS_FS}) correspond to the increasing and horizontal segments of the admission threshold, respectively (to be made clear in Eq. \eqref{eq_def_IS_FS}).
	
	Lemma \ref{theorem_difference_equation_X} largely follows Lemma \ref{theorem_difference_equation}.  Specifically, if Eq. \eqref{eq_nonliner_difference_equation} transits in the reverse order from $ \chi_{\bar{k}-\tau}^{(\tau)} = p_{\max} $ to $ \chi_1^{(\tau)} $ (i.e., the initial value is given by $ \chi_{\bar{k}-\tau}^{(\tau)} = p_{\max} $), then it is easy to see that we can recursively obtain $ \chi_1^{(\tau)} $ as follows:
	\begin{align*}
		p_{\max} = \chi_{\bar{k}-\tau}^{(\tau)} \rightarrow \chi_{\bar{k}-\tau-1}^{(\tau)}\rightarrow \chi_{\bar{k}-\tau-2}^{(\tau)} \rightarrow \cdots \rightarrow \chi_{i-1}^{(\tau)} \rightarrow \cdots \rightarrow \chi_2^{(\tau)} \rightarrow \chi_1^{(\tau)}.
	\end{align*}
	We emphasize the above reverse transition is unique and well-defined as the conjugation function $ f^* $ is strictly increasing (and thus its inverse is well-defined).  The uniqueness of solutions to any first-order difference equation (i.e., Lemma \ref{theorem_difference_equation}) implies that the mapping from $ \alpha_{\textsc{is}}^{(\tau)}  $ to $ \chi_1^{(\tau)}  $ is  unique as well. Since $ \alpha_{\textsc{is}}^{(\tau)} > 0 $ and $ p_{\max} \geq  c_{\bar{k}} $, $ \chi_{i}^{(\tau)}  > \chi_{i-1}^{(\tau)}  $ and $ \chi_i^{(\tau)}  > c_{\tau+i+1} $ trivially follow. 
\end{proof}

(\textbf{Rationality of} $ \textsf{SoE}(\boldsymbol{\chi}^{(\tau)}) $). The system of equations $ \textsf{SoE}(\boldsymbol{\chi}^{(\tau)}) $ follows by enforcing the equality in Eq. \eqref{eq_OPD_inequalities_proof_1} and Eq. \eqref{eq_OPD_inequalities_proof_2}. The only change is the notation of variables in $ \textsf{SoE}(\boldsymbol{\chi}^{(\tau)}) $, which is revisited here for a better reference.
\begin{align}\label{eq_system_of_equations_X_proof}
	\Big(\textsf{SoE}(\bm{\chi}^{(\tau)})\Big):\ 
	%\chi_0 = 
	\frac{f^{*}\big(\chi_1^{(\tau)}\big)}{g(\tau+1)} = \frac{f^{*}\big(\chi_2^{(\tau)}\big)-f^{*}\big(\chi_1^{(\tau)}\big)}{\chi_1^{(\tau)}-c_{\tau+2}} = \cdots =  \frac{f^{*}(p_{\max})-f^{*}\big(\chi_{\bar{k}-\tau-1}^{(\tau)}\big)}{\chi_{\bar{k}-\tau-1}^{(\tau)} - c_{\bar{k}}}.
	%f^*(p_m) = f^*(p_{m-1}) + \alpha p_{m-1} - \alpha c_m, m\in \{k+2,\cdots,M\},
\end{align}

%as Proposition \ref{theorem_uniqueness} is treated as a general mathematical property that is independent of our design and analysis of threshold policies. 

%As noted earlier, the above system of $ \bar{k}-k -1 $ equations has $ \bar{k}-k-1 $ variables. 

(\textbf{Proof of Proposition \ref{theorem_uniqueness}}). Now we provide the formal proof of Proposition \ref{theorem_uniqueness}.  Lemma \ref{theorem_difference_equation_X} implies that there exists a unique one-to-one mapping between $ \alpha_{\textsc{is}}^{(\tau)}  $ and $ \chi_1^{(\tau)}  $, that is, a given $ \alpha_{\textsc{is}}^{(\tau)} \in (0,+\infty) $ implies a unique  $ \chi_1^{(\tau)} \in (c_{\tau+2},p_{\max}) $. If we denote $ \chi_1^{(\tau)} = \theta \in  (c_{\tau+2},p_{\max})  $, and divide Eq. \eqref{eq_system_of_equations_X_proof} into the following two parts:
\begin{align}\label{eq_def_IS_FS}
	\alpha_{\textsc{is}}^{(\tau)}(\theta) =  \frac{f^*(\chi_2^{(\tau)})-f^*(\theta)}{\theta-c_{\tau+2}} = \cdots =  \frac{f^{*}(p_{\max})-f^{*}\big(\chi_{\bar{k}-\tau-1}^{(\tau)}\big)}{\chi_{\bar{k}-\tau-1}^{(\tau)} - c_{\bar{k}}}, \quad \alpha_{\textsc{hs}}^{(\tau)}(\theta) =
	\frac{f^*(\theta)}{g(\tau+1)},
\end{align}
then $ \alpha_{\textsc{is}}^{(\tau)}(\theta) $ is a well-defined function over $ \theta\in (c_{\tau+2},p_{\max}) $ (by Lemma \ref{theorem_difference_equation_X}). We can further prove that $ \alpha_{\textsc{is}}^{(\tau)}(\theta) $ and $ \alpha_{\textsc{hs}}^{(\tau)}(\theta) $ are strictly-decreasing and strictly-increasing over $ \theta\in (c_{\tau+2},p_{\max}) $, respectively. The proof is elementary, so here we only explain the intuitions:
\begin{itemize}
	\item \textit{Monotonicity of $ \alpha_{\textsc{is}}^{(\tau)}(\theta) $}. By Lemma \ref{theorem_difference_equation_X}, for each given  $ \tau\in \{0,1,\cdots,\ubar{k}-1\} $, $ \chi_{\bar{k}-\tau-1}^{(\tau)} $ is strictly-decreasing over $ \alpha_{\textsc{is}}^{(\tau)} \in (0,+\infty) $. We can show recursively that $ \chi_1^{(\tau)} $ is strictly-decreasing over $ \alpha_{\textsc{is}}^{(\tau)} \in (0,+\infty) $ as well. Thus, based on the monotonicity property of  inverse functions, $ \alpha_{\textsc{is}}^{(\tau)}(\theta) $ is strictly-decreasing over $ \theta\in (c_{\tau+2},p_{\max}) $.
	
	\item \textit{Monotonicity of  $ \alpha_{\textsc{hs}}^{(\tau)}(\theta) $}. Based on the strict monotonicity of the conjugate function $ f^* $ (i.e., Lemma \ref{theorem_conjugate}), it is obvious  that $ \alpha_{\textsc{hs}}^{(\tau)}(\theta) $  is strictly-increasing over $ \theta \in (c_{\tau+2},p_{\max}) $. 
\end{itemize}

%$ \alpha_{\textsc{is}}^{(\tau)}(\theta) $ is a well-defined and strictly-decreasing function over $ \theta \in (c_{k+2},p_{\max}) $.

Based on the strict monotonicity of $ \alpha_{\textsc{is}}^{(\tau)}(\theta) $ and $ \alpha_{\textsc{hs}}^{(\tau)}(\theta) $, to prove the existence of a unique solution to the system of equations $ \textsf{SoE}(\bm{\chi}^{(\tau)}) $, it suffices to prove that for each given $ \tau\in [\ubar{k}-1] $, there exists a unique $ \theta\in (c_{\tau+2},p_{\max}) $ so that $ \alpha_{\textsc{is}}^{(\tau)}(\theta) = \alpha_{\textsc{hs}}^{(\tau)}(\theta) $. This can be proved by evaluating the lower and upper bounds of $ \alpha_{\textsc{is}}^{(\tau)}(\theta) $ and $ \alpha_{\textsc{hs}}^{(\tau)}(\theta) $. Specifically, when $ \theta \rightarrow c_{\tau+2} $, we have $ \alpha_{\textsc{is}}^{(\tau)}(\theta) \rightarrow +\infty $; when $ \theta \rightarrow p_{\max} $, we have $ \alpha_{\textsc{is}}^{(\tau)}(\theta) \rightarrow 0 $. On the other hand, $ \alpha_{\textsc{hs}}^{(\tau)}(\theta)  $ is always a finite positive real number when $ \theta \in  (c_{\tau+2},p_{\max}) $. The strict monotonicity of $ \alpha_{\textsc{is}}^{(\tau)}(\theta) $ and $\alpha_{\textsc{hs}}^{(\tau)}(\theta) $ implies that there exists a unique $ \theta\in (c_{\tau+2},p_{\max}) $ such that $ \alpha_{\textsc{is}}^{(\tau)}(\theta) = \alpha_{\textsc{hs}}^{(\tau)}(\theta) $. Therefore, there exists a unique set of $ \bar{k}-\tau -1$ positive real numbers $ \bm{\chi}^{(\tau)} = \{\chi_1^{(\tau)}, \chi_2^{(\tau)}, \cdots, \chi_{\bar{k}-\tau-1}^{(\tau)}\} $ that satisfies the system of equations $\normalfont \textsf{SoE}(\bm{\chi}^{(\tau)}) $. The properties of $ \{\chi_1^{(\tau)}, \chi_2^{(\tau)}, \cdots, \chi_{\bar{k}-\tau-1}^{(\tau)}\} $ being monotonic and lower bounded directly follow Lemma \ref{theorem_difference_equation_X}. We thus complete the proof of Proposition \ref{theorem_uniqueness}.

\section{Proof of Theorem \ref{theorem_optimality}}
	\label{sec_proof_of_theorem_optimality}
	
	In this section, we sketch the proof of Theorem \ref{theorem_optimality}. Section \ref{section_optimal_threshold} has explained the sufficiency of Theorem \ref{theorem_optimality}, namely, $\normalfont \textsf{TOS}_{\boldsymbol{\lambda}^*} $ is $ \textsf{CR}_f^*(\rho, k) $-competitive if $ \boldsymbol{\lambda}^* $ satisfies the \textbf{SoSE} in Eq. \eqref{eq_system_of_equations}. Here in this section,  we focus on illustrating the necessity of Theorem \ref{theorem_optimality}, that is, a unique solution to Eq. \eqref{eq_system_of_equations} is necessary for the existence of any $ \textsf{CR}_f^*(\rho, k) $-competitive deterministic online algorithm. 
	%Based on this necessity, we show that $ \textsf{CR}_f^*(\rho, M) $ is indeed optimal for all deterministic online algorithms.

	\subsection{Definition of Selection Function $ \psi(p) $}
	\label{section_def_packing_function}
	
	An important first step for our following proof is the definition of selection functions below.
	
	\begin{definition}[Selection Functions]\label{def_packing_functions}
		A selection function $ \psi(p)$ is a mapping from $ p\in [p_{\min},p_{\max}] $ to an integer within $ [\bar{k}]$, denoting the total number of selected buyers whose offered price is $ p $.
	\end{definition}
	
	Given an arrival instance, the realization of any deterministic algorithm can be fully characterized by the resulting selection function. For any $ p\in [p_{\min}, p_{\max}] $, we say $ p $ is a \textbf{non-zero point} of $ \psi $ if $ \psi(p) \neq 0 $. We use $ \Omega \triangleq \{\omega_1,\omega_2,\cdots, \omega_L\} $ to denote the set of non-zero points of a selection function $ \psi $,
	%\begin{align*}
	%  \Omega \triangleq \{\omega_1,\omega_2,\cdots, \omega_L\}.
	%\end{align*}
	where $ L $ denotes the total number of non-zero points. We assume without loss of generality that the non-zero points are labelled in the ascending order, namely, $
	p_{\min} \leq   \omega_1 < \omega_2  < \cdots < \omega_L \leq  p_{\max} $. Based on the definition of $ \Omega $, a selection function $ \psi(p) $ can be written as a finite combination of indicator functions as follows:
	%\vspace{-0.1cm}
	\begin{equation}\label{eq_packing_function}%\vspace{-0.1cm}
		\psi(p) = \sum_{\ell =1}^L A_{\ell}\cdot \mathds{1}_{\{p=\omega_{\ell}\}},
	\end{equation}
	where $ A_{\ell} $ denotes the number of selected buyers whose price is $ \omega_{\ell} $. Recall that at most $ \bar{k} $ items will be produced, we thus have $ A_{\ell}\in [\bar{k}] $  and $ L \in [\bar{k}] $. 
	%$ L = \bar{M} $ indicates that $ \bar{M} $ items are accepted and all of them have different values (i.e., $ \lambda_1 = \lambda_2 = \cdots = \lambda_L = 1 $). 
	%For instance, consider the case with $ p_{\min} = 1.5 $, $ p_{\max} = 4.5 $, and $ \bar{M} = 5 $. If $ \Omega = \{1.5,\ 2.4,\ 3.6\} $, then $ L = 3 $. In this case if we assume  $ \lambda_1 = 1,  \lambda_2 = 3$, and $ \lambda_3 = 1 $, then  one item with value 1.5, three items with value 2.4, and one item with value 3.6, are accepted.  If $ \Omega = \{1.5,\ 2.5,\ 3.5,\ 4.1,\ 4.5\} $, then $ L = \bar{M} = 5 $. In this case, the only  possible way is $ \lambda_1 = \lambda_2 = \cdots = \lambda_5 = 1 $, meaning that the five accepted items are all with different values. In summary, a selection function $ \psi $ can be uniquely determined by $ \Omega = \{\omega_1,\omega_2, \cdots, \omega_L\} $ and $ \{\lambda_{\ell}\}_{\ell = 1,2,\cdots, L} $
	
	(\textbf{Notations}) To simplify the notations in the remaining of this section, we define $ \tau^{(\alpha)} $ by
	%\vspace{-0.1cm}
	\begin{equation*}\label{eq_k_alpha}%\vspace{-0.1cm}
		\tau^{(\alpha)} \triangleq  g^{\textsf{inv}}\Big(\frac{1}{\alpha}f^{*}(p_{\min})\Big) - 1,
	\end{equation*}
	which is the smallest integer in $ \mathcal{T}^{(\alpha)} $ (defined in Eq. \eqref{eq_K_alpha}) for a given competitive ratio parameter $ \alpha\in [1,+\infty) $. Note that $ \tau^{(\alpha)} $ is monotonically decreasing in $ \alpha\in [1,+\infty) $. In particular, when $ \alpha\rightarrow +\infty $, $ \tau^{(\alpha)} \rightarrow 0 $; when $ \alpha\rightarrow 1 $, $ \tau^{(\alpha)} \rightarrow \ubar{k}-1 $.
	
	We next prove that \textit{\bfseries the existence of an $ \alpha $-competitive deterministic algorithm, not necessary a threshold policy, is  related to the existence of a set of selection functions subject to some feasibility conditions}.

	%\vspace{-0.2cm}
	\subsection{Necessary Conditions: Existence of A Non-Empty $ \mathcal{P}^{(\alpha)} $} 
	
	%Our proof is heavily based on a series of special arrival instances constructed as follows.
	
	\begin{definition}[$ \varepsilon $\textsf{-Instance} $ \mathcal{I}_{n}^{(\varepsilon)} $]\label{def_I_n}
		We discretize the interval $ [p_{\min},p_{\max}] $ into $ N^{(\varepsilon)} -1 $ segments with step size $ \varepsilon $, where $ \varepsilon $ is a small positive real number so that $ N^{(\varepsilon)} = (p_{\max}-p_{\min})/\varepsilon + 1$ is an integer. For a given step size $ \varepsilon $, we define $ \mathcal{I}_{n}^{(\varepsilon)} = \{\mathcal{A}_{1}^{(\varepsilon)}, \mathcal{A}_{2}^{(\varepsilon)}, \cdots, \mathcal{A}_{n}^{(\varepsilon)}\}$ and $ p_{n}^{(\varepsilon)} = p_{\min}  + (n-1)\varepsilon $,
		where $ n \in [N^{(\varepsilon)}] $. Specifically, $ \mathcal{I}_{n}^{(\varepsilon)} $ consists of $ n $ groups of buyers denoted by $ \mathcal{A}_{\ell}^{(\varepsilon)} $ for $ \ell\in  [n] $. The first group  $ \mathcal{A}_{1}^{(\varepsilon)} $ consists of $ \Gamma(p_{1}^{(\varepsilon)}) $ copies of identical buyers with price $ p_{1}^{(\varepsilon)} $. $ \mathcal{A}_1^{(\varepsilon)} $ is followed by the second group $ \mathcal{A}_{2}^{(\varepsilon)} $, which consists of $ \Gamma(p_{2}^{(\varepsilon)}) $ copies of identical buyers with price  $ p_{2}^{(\varepsilon)} $. In general, after group $ \mathcal{A}_{\ell}^{(\varepsilon)} $, there is a group $ \mathcal{A}_{\ell+1}^{(\varepsilon)} $ of $ \Gamma(p_{\ell+1}^{(\varepsilon)})  $ copies of identical buyers with price  $ p_{\ell+1}^{(\varepsilon)}$, where $ \ell = [n-1] $.	
	\end{definition}
	%\vspace{-0.2cm}
	
	Definition \ref{def_I_n} indicates that $ p_1^{(\varepsilon)} = p_{\min} $ and $ p_{N^{(\varepsilon)}}^{(\varepsilon)} = p_{\max} $ hold for any small step size $ \varepsilon > 0 $. Based on the above definition of $ \varepsilon $\textsf{-instance}, Proposition \ref{theorem_necessary} below argues that to guarantee the existence of an $ \alpha $-competitive deterministic algorithm, there must exist a set of selection functions that satisfy a certain feasibility conditions.  
	
	\begin{proposition}[Feasibility Conditions]
		\label{theorem_necessary}
		If there is an $ \alpha $-competitive deterministic online algorithm, then there exists a selection function $ \psi(p): [p_{\min},p_{\max}] \rightarrow \{0,1,\cdots, k\} $ such that
		\begin{subequations}
			\begin{align} \normalfont  
				\int_{p_{\min}}^p \eta \psi(\eta)d\eta - f\left(\int_{p_{\min}}^p \psi(\eta)d\eta\right) \geq  \frac{1}{\alpha} f^{*}(p), \quad \forall p \in (p_{\min},p_{\max}],\label{eq_necessary_inequality}\\
				\psi(p_{\min}) \in \{\tau^{(\alpha)} + 1,\cdots, \ubar{k}\}, \  \int_{p_{\min}}^{p_{\max}} \psi(\eta)d\eta \leq \bar{k}. \label{eq_necessary_boundary}
			\end{align}
		\end{subequations}
	\end{proposition}
	
	\vspace{-0.3cm}
	\begin{corollary}\label{theorem_def_P_alpha}
		If there is an $ \alpha $-competitive deterministic online algorithm, and 	 $ \mathcal{P}^{(\alpha)} $ is defined as
		%\vspace{-0.1cm}
		\begin{equation*}%\vspace{-0.1cm}
			\mathcal{P}^{(\alpha)} \triangleq  \{\psi | \psi \text{ satisfies Eq. \eqref{eq_necessary_inequality} and Eq. \eqref{eq_necessary_boundary}} \},
		\end{equation*}
		then $  \mathcal{P}^{(\alpha)} $ is a well-defined and non-empty set. 
	\end{corollary}
	\begin{proof}
		Proposition \ref{theorem_necessary} and Corollary \ref{theorem_def_P_alpha} are equivalent to each other.  The proof is based on the fact that if there is an $ \alpha $-competitive online algorithm, denoted by \textsf{ALG}, then \textsf{ALG}  must achieve at least $ \frac{1}{\alpha} $ portion of optimum under all possible arrival instances, including the $ \varepsilon $\textsf{-Instance} $ \mathcal{I}_{n}^{(\varepsilon)} $ for all $ n \in  [N_{\varepsilon}] $.  
		
		Suppose $ n=1 $, and the arrival instance is given by $ \mathcal{I}_1^{(\varepsilon)} $. The optimal profit in hindsight is  $
		\textsf{OPT}(\mathcal{I}_1^{(\varepsilon)}) = p_{1}^{(\varepsilon)}  \Gamma(p_{1}^{(\varepsilon)}) - f(\Gamma(p_{1}^{(\varepsilon)})) = f^{*}(p_{1}^{(\varepsilon)}) = f^{*}(p_{\min}) $. In the online setting, similar to the proof of Proposition \ref{theorem_k_alpha}, an online algorithm is $ \alpha $-competitive only if there exists an integer $ \tau\in \mathcal{T}^{(\alpha)} $ such that  $ \textsf{ALG}(\mathcal{I}_1^{(\varepsilon)}) =  p_{\min}(\tau+1) - f(\tau+1) = g(\tau+1) \geq \frac{1}{\alpha} \textsf{OPT}(\mathcal{I}_1^{(\varepsilon)}) = \frac{1}{\alpha} f^{*}(p_{\min})$. Therefore, there must exist a selection function $ \psi_1 $ with $ \tau \in \mathcal{T}^{(\alpha)} $ such that  $ \psi_1(p_{\min}) = \tau+1 \in \{\tau^{(\alpha)} + 1, \cdots, \ubar{k}\} $.
		
		Suppose $ n\in \{2,3,\cdots,N^{(\varepsilon)}\} $, and the arrival instance is given by $ \mathcal{I}_n^{(\varepsilon)} $. In the offline setting, similar to the case with $ \mathcal{I}_1^{(\varepsilon)} $, the optimal strategy in hindsight is to select all the buyers in the last group $ \mathcal{A}_{n}^{(\varepsilon)} $, i.e., $ \textsf{OPT}(\mathcal{I}_{n}^{(\varepsilon)}) =  v_{n}^{(\varepsilon)} \Gamma(p_{n}^{(\varepsilon)}) -  f(\Gamma(p_{n}^{(\varepsilon)})) = f^{*}(p_{n}^{(\varepsilon)}) $. The existence of an $ \alpha $-competitive online algorithm is equivalent to the existence of at least one selection function $ \psi_n $ so that $ \textsf{ALG}(\mathcal{I}_{n}^{(\varepsilon)}) = \sum_{\ell=1}^{n} p_{\ell}^{(\varepsilon)} \psi_n(p_{\ell}^{(\varepsilon)}) - f(\sum_{\ell=1}^{n} \psi_n(p_{\ell}^{(\varepsilon)})) \geq \frac{1}{\alpha}\textsf{OPT}(\mathcal{I}_n^{(\varepsilon)})  = \frac{1}{\alpha} f^{*}(p_{n}^{(\varepsilon)})$ and $ \sum_{\ell=1}^{n} \psi_n(p_{\ell}^{(\varepsilon)})\leq \bar{k} $ (i.e., at most $ \bar{k} $ items will be produced), where $ p_{\ell}^{(\varepsilon)} = p_{\min} + (\ell-1)\varepsilon $. Thus, the online algorithm being $ \alpha $-competitive indicates that 
		%\vspace{-0.2cm}
		\begin{equation*}%\vspace{-0.2cm}
			\sum_{\ell=1}^{n} p_{\ell}^{(\varepsilon)} \psi_n\big(p_{\ell}^{(\varepsilon)}\big) - f\Big(\sum_{\ell=1}^{n} \psi_n\big(p_{\ell}^{(\varepsilon)}\big)\Big) \geq \frac{1}{\alpha}f^{*}\big(p_{n}^{(\varepsilon)}\big) \text{ and } \sum_{\ell=1}^{n} \psi_n(p_{\ell}^{(\varepsilon)})\leq \bar{k}.
			%\forall n\in \{2,3,\cdots, N^{(\varepsilon)}\}.
		\end{equation*}
		
		Notice that in Definition \ref{def_I_n}, the $ \varepsilon $\textsf{-Instance} is constructed in a nested manner, namely, $ \mathcal{I}_{n+1}^{(\varepsilon)} = \mathcal{I}_{n}^{(\varepsilon)} \cup \{\mathcal{A}_{n+1}^{(\varepsilon)}\}$ for all $ n = [N^{(\varepsilon)}-1] $. Meanwhile, the arrival instance to be presented to an online algorithm can be $ \mathcal{I}_{n}^{(\varepsilon)} $ for any $ n \in  [N^{(\varepsilon)}] $, and it is impossible for an online algorithm to know the information of $ n $ a priori (recall that only the setup $ \mathcal{S} $ is known). Thus,  the series of selection functions $ \psi_n $ for $ n \in  [N^{(\varepsilon)}] $ must overlap and be the same function $ \psi $ that satisfies 
		\begin{align*}
			\begin{cases}
				\sum_{\ell=1}^{n} p_{\ell}^{(\varepsilon)} \psi\big(p_{\ell}^{(\varepsilon)}\big) - f\left(\sum_{\ell=1}^{n} \psi\big(p_{\ell}^{(\varepsilon)}\big)\right) \geq \frac{1}{\alpha}f^{*}\big(p_{n}^{(\varepsilon)}\big), \smallskip \\
				\psi(p_{\min}) \in \{\tau^{(\alpha)} + 1,\cdots, \ubar{k}\},\ \sum_{\ell=1}^{n} \psi(p_{\ell}^{(\varepsilon)})\leq \bar{k},
			\end{cases}
			\ \ \forall n \in  [N^{(\varepsilon)}].
		\end{align*}
		The above inequalities hold for all $ \varepsilon > 0 $, and thus it holds for $ \varepsilon \rightarrow 0 $ as well, leading to the integral inequalities in Eq. \eqref{eq_necessary_inequality} and  Eq. \eqref{eq_necessary_boundary}.
		%the feasibility conditions in Eq. \eqref{eq_necessary_inequality} and Eq. \eqref{eq_necessary_boundary}. 
		We thus complete the proof of Proposition \ref{theorem_necessary}. 
	\end{proof}
	
	(\textbf{$\alpha$-Feasible}) Proposition \ref{theorem_necessary} indicates that  $ \mathcal{P}^{(\alpha)} $ is non-empty as long as there exists an $ \alpha $-competitive deterministic online algorithm. In what follows we say $ \psi $ is $ \alpha $-feasible if $ \psi\in \mathcal{P}^{(\alpha)}  $. Intuitively, for any $ \alpha $-feasible selection function $ \psi $, $ p_{\min} $ is always the first non-zero point, i.e.,  $ \omega_1 = p_{\min} $. Thus, an $ \alpha $-feasible $ \psi $ must have at least one non-zero point. Corollary \ref{theorem_number_of_omega} below shows that, the number of non-zero points of an $ \alpha $-feasible selection function is upper bounded.
	\begin{corollary}\label{theorem_number_of_omega}
		If  $\psi $ is $ \alpha $-feasible, i.e., $ \psi \in \mathcal{P}^{(\alpha)} $, then $ \psi $ has at most $ \bar{k} - \tau^{(\alpha)} $ non-zero points.
	\end{corollary}
	
	Corollary \ref{theorem_number_of_omega} is intuitive based on Definition \ref{def_packing_functions} and Proposition \ref{theorem_necessary}. Under the $ \varepsilon $\textsf{-instance}, Eq. \eqref{eq_necessary_boundary} implies that no more than  $ \bar{k} $ units should be produced/sold, of which at least $ \tau^{(\alpha)} + 1 $ units must be sold to buyers from group $ \mathcal{A}_1^{(\epsilon)} $, i.e., buyers with price $ p_{\min} $. Thus, based on the definition of $ \psi $ in Eq. \eqref{eq_packing_function}, we know that at most $ \bar{k} - \tau^{(\alpha)}   $ non-zero points exist with $ A_1  =  \tau^{(\alpha)} + 1$ and $A_2 = A_3= \cdots = A_{\bar{k} - \tau^{(\alpha)}} = 1 $.

	%\vspace{-0.2cm}
	\subsection{Characteristic Selection Function $ \psi^{(\alpha)} $ and Its Feasibility}
	\label{section_characteristic_packing_function}
	
	A key step in proving Theorem \ref{theorem_optimality} is the construction of a special $ \alpha $-feasible selection function, termed \textbf{characteristic selection function}, given in Theorem \ref{theorem_characteristic} below.
	
	%\footnote{The term `characteristic' means that  $ \psi^{(\alpha)} $ captures the core properties of $ \alpha $-competitive deterministic online algorithms.}
	
	%Theorem \ref{theorem_characteristic} below shows that can be constructed and proved to be feasible in $ \mathcal{P}^{(\alpha)} $.
	
	\begin{theorem}[Characteristic selection function] 
		\label{theorem_characteristic}
		If there is an $ \alpha $-competitive deterministic online algorithm, then $ \mathcal{P}^{(\alpha)} $ consists of a unique characteristic selection function $ \psi^{(\alpha)}  $ given by
		\begin{align} \label{eq_Lambda_alpha}
			\psi^{(\alpha)}(p) = 
			\begin{cases}
				\tau^{(\alpha)} + 1 & \text{ if } p = \omega_1^{(\alpha)} = p_{\min},\\
				1   & \text{ if } p \in \Omega^{(\alpha)}\backslash \{\omega_1^{(\alpha)}\},\\ %\Big\{\omega_2^{(\alpha)}, \omega_3^{(\alpha)}, \cdots, \omega_{L^{(\alpha)}}^{(\alpha)}\Big\},\\
				0   & \text{ if } p \in [p_{\min},p_{\max}]\backslash \Omega^{(\alpha)}, 
			\end{cases} 
		\end{align}
		where the non-zero points $ \Omega^{(\alpha)} = \{\omega_1^{(\alpha)},\omega_2^{(\alpha)},\cdots, \omega_{L^{(\alpha)}}^{(\alpha)}\} $ satisfy:
		%n{subequations}
		\vspace{-0.1cm}
		\begin{equation}\label{eq_omega_alpha} \vspace{-0.1cm}
			\sum_{j=1}^{\ell-1} \omega_{j}^{(\alpha)}\cdot \psi^{(\alpha)}\big(\omega_{j}^{(\alpha)}\big) - f\Big(\tau^{(\alpha)} + \ell - 1\Big) = \frac{1}{\alpha}f^*\big(\omega_{\ell}^{(\alpha)}\big), \quad \forall \ell \in \{2, 3, \cdots, L^{(\alpha)}\}.
			%\\ \label{eq_L_alpha}
			%& \sum_{\ell=1}^{L^{(\alpha)}} \omega_{\ell}^{(\alpha)}\cdot \psi^{(\alpha)}\big(\omega_{\ell}^{(\alpha)}\big) - f\Big(k^{(\alpha)} + L^{(\alpha)}\Big) \geq  \frac{1}{\alpha}f^*\big(p_{\max}\big).
		\end{equation}
		%\end{subequations}
		Note that $ L^{(\alpha)}\in [\bar{k} - \tau^{(\alpha)}]$ and Eq. \eqref{eq_omega_alpha} vanishes if $ L^{(\alpha)} = 1 $.
	\end{theorem}
	\begin{proof}
		Theorem \ref{theorem_characteristic} can be considered a stronger version of Proposition \ref{theorem_necessary}.  Here the conclusion is twofold: i) $ \psi^{(\alpha)} $ is always $ \alpha $-feasible as long as $ \mathcal{P}^{(\alpha)} $ is non-empty, and ii) $ \psi^{(\alpha)} $ is uniquely determined by Eq. \eqref{eq_Lambda_alpha} and Eq. \eqref{eq_omega_alpha} for whatever types of $ \alpha $-competitive deterministic algorithms. The proof is based on three \textit{feasibility-preserving operations}, based on which we can show that  $ \psi^{(\alpha)} $ can be constructed from any $ \psi\in\mathcal{P}^{(\alpha)} $ by repeatedly applying these operations without violating the feasibility conditions in Eq. \eqref{eq_necessary_inequality} and Eq. \eqref{eq_necessary_boundary}. The uniqueness of $\psi^{(\alpha)} $ follows the fact that the construction of the non-zero points $ \Omega^{(\alpha)} = \{\omega_1^{(\alpha)},\omega_2^{(\alpha)},\cdots, \omega_{L^{(\alpha)}}^{(\alpha)}\} $ is unique.
		
		The three feasibility-preserving operations are given as follows:
		\begin{claim}[\textsc{Push-Down-to-Minimum}]\label{claim_k}
			If there exists a selection function $ \psi \in \mathcal{P}^{(\alpha)} $ with $ \psi(p_{\min}) > \tau^{(\alpha)} + 1$, then there exists another selection function $ \tilde{\psi}\in \mathcal{P}^{(\alpha)} $ with $ \tilde{\psi}(p_{\min}) =  \tau^{(\alpha)} + 1$.
		\end{claim}
		
		\begin{claim}[\textsc{Push-Down-to-One}]\label{claim_two_one}
			If there exists a selection function $ \psi \in \mathcal{P}^{(\alpha)} $ with $ \psi(\omega_{\ell}) > 1 $ for some non-zero point $ \omega_{\ell}\in \Omega\backslash\{p_{\min}\} $, then there exists another selection function $ \tilde{\psi}\in \mathcal{P}^{(\alpha)} $  which is the same as $  \psi $ for  $ p\in [p_{\min},\omega_{\ell}) $ and satisfies $ \tilde{\psi}(\omega_{\ell}) = 1 $. 
		\end{claim}
		
		\begin{claim}[\textsc{Push-Right-To-Equality}]\label{claim_equality}
			If there exists a selection function $ \psi \in \mathcal{P}^{(\alpha)} $ with at least two non-zero points (i.e., $ L\geq 2$) and the following inequality holds for some $ \ell \in \{2,3,\cdots, L\} $:
			\begin{align}\label{eq_larger_than}
				\sum_{j=1}^{\ell-1} \omega_{j} \psi(\omega_{j})  - f\Big(\sum_{j=1}^{\ell-1} \psi(\omega_{j})\Big) >  \frac{1}{\alpha} f^{*}(\omega_{\ell}), 
			\end{align}
			then there exists a selection function $ \tilde{\psi}\in \mathcal{P}^{(\alpha)} $ which is the same as $  \psi $ for  $ p\in [p_{\min},\omega_{\ell}) $ and satisfies
			\begin{align}\label{eq_equality}
				\sum_{j=1}^{\ell-1} \omega_{j} \tilde{\psi}(\omega_{j}) - f\Big(\sum_{j=1}^{\ell-1} \tilde{\psi}(\omega_{j})\Big) =  \frac{1}{\alpha} f^{*}(\tilde{\omega}_{\ell}),
			\end{align}
			where $ \tilde{\omega}_{\ell} > \omega_{\ell} $   and $ \tilde{\omega}_{\ell} $ is the $ \ell $-th non-zero point of $ \tilde{\psi} $ if $ \tilde{\omega}_{\ell}\in (\omega_{\ell},p_{\max}] $. 
		\end{claim}
		
		\textit{Intuition of Claim \ref{claim_k}}:  
		The proof of this claim is given in Appendix \ref{proof_of_claim_k}. For any $ \alpha $-competitive deterministic algorithm, Proposition \ref{theorem_necessary} indicates that there exists an $ \alpha $-feasible selection function $ \psi $  so that $  \psi(p_{\min}) \geq \tau^{(\alpha)} + 1 $. Here, Claim \ref{claim_k} gives a stronger result by showing that it is sufficient to consider $ \psi(p_{\min}) = \tau^{(\alpha)} + 1 $ only. We refer to Claim \ref{claim_k} by the \textsc{Push-Down-to-Minimum} operation as it always pushes $ \psi(p_{\min}) $ down to some value as small as possible.

		\textit{Intuition of Claim \ref{claim_two_one}}: 
		The proof of this claim is given in Appendix \ref{proof_of_claim_two_one}. Claim \ref{claim_two_one} argues that if an $ \alpha $-feasible $ \psi $ selects more than one buyer with price $\omega_{\ell} $ for some $ \ell\in \{2,3,\cdots, L\} $ (note that $ \omega_1 = p_{\min} $), then it is always possible to construct a new $ \alpha $-feasible selection function $ \tilde{\psi} $ by keeping the same number of selected buyers with prices lower than $ \omega_{\ell} $; accepting only one buyer with price $ \omega_{\ell} $ (i.e., $\tilde{\psi}(\omega_{\ell}) =1$); and reserving the residual production capacity for buyers with prices  higher  than $ \omega_{\ell} $.  Here our proof is based on the convexity of $ f $, and the details are given in  Appendix \ref{proof_of_claim_two_one}. We refer to Claim \ref{claim_two_one} by the \textsc{Push-Down-to-One} operation as it always pushes the value of $ \psi(p) $ down to 1 for all $ p \in \Omega\backslash\{p_{\min}\} $. 
		
		\textit{Intuition of Claim \ref{claim_equality}}: 
		The proof of this claim is given in Appendix \ref{proof_of_claim_equality}. Claim \ref{claim_equality}  argues that whenever an $ \alpha $-feasible selection function $ \psi $ has a non-zero point $ \omega_{\ell} $ that satisfies the strict inequality in Eq. \eqref{eq_larger_than}, then it is always possible to increase the value of this non-zero point to satisfy the equality in Eq. \eqref{eq_equality}, namely, the existence of a unique $ \tilde{\omega}_{\ell} > \omega_{\ell} $ so that Eq. \eqref{eq_equality} holds.  Here our proof is heavily based on the strict monotonicity of $ f^* $, and the details are given in  Appendix \ref{proof_of_claim_equality}. We refer to Claim \ref{claim_equality} by the \textsc{Push-Right-to-Equality} operation as it always pushes the non-zero points of $ \psi(p) $ to the rightmost (i.e., the increasing direction) to achieve the equality in Eq. \eqref{eq_equality}.
		
		We are now ready to demonstrate that Theorem \ref{theorem_characteristic} follows by repeatedly applying  the above three feasibility-preserving operations. Specifically, for any given $ \alpha $\textit{-feasible} selection function $ \psi $ (note that such $  \psi $ always exists since $ \mathcal{P}^{(\alpha)} $ is non-empty), if $  \psi =  \psi^{(\alpha)} $, then we finish the proof; otherwise,  we apply the above three operations in the following order:
		\begin{itemize}
			\item First, if $ \psi(p_{\min})  > \tau^{(\alpha)} + 1 $,  then we apply the \textsc{Push-Down-to-Minimum} operation by   Claim \ref{claim_k} so that $ \psi(\omega_1) = \psi(p_{\min})  = \tau^{(\alpha)} + 1 $. 
			
			\item Second, for all $ \ell \in \{ 2, 3,\cdots, L\} $, if  $ \psi(\omega_{\ell}) > 1 $, then we apply the \textsc{Push-Down-to-One} operation by Claim \ref{claim_two_one} to push $ \psi(\omega_{\ell}) = 1 $ while keeping the selection function unchanged over the real interval $ [p_{\min},\omega_{\ell}) $. We can repeatedly apply this operation so that $ \psi(p) = 1 $ holds for all the non-zero points except $ \omega_1 $. 
			
			\item Third, if the resulting selection function consists of at least two non-zero points, then from its second non-zero point up to the last one, we repeatedly apply the \textsc{Push-Right-to-Equality} operation by Claim \ref{claim_equality} so that Eq. \eqref{eq_equality} holds for all the non-zero points  except $ \omega_1 $. 
		\end{itemize}
		
		Based on the above operations, we can construct the unique $ \psi^{(\alpha)} $ given by Eq. \eqref{eq_Lambda_alpha} and guarantee that $ \psi^{(\alpha)}  $ is $ \alpha $-feasible, i.e., $ \psi^{(\alpha)} \in \mathcal{P}^{(\alpha)} $.  The non-zero points of  $ \psi^{(\alpha)} $  should satisfy Eq. \eqref{eq_equality}, and thus Eq. \eqref{eq_omega_alpha} naturally follows after substituting the expression of $ \psi^{(\alpha)} $ into Eq. \eqref{eq_equality}. In summary, for any $ \alpha $-competitive deterministic algorithm, there exists a unique $ \alpha $-feasible characteristic selection function $ \psi^{(\alpha)} $ given by Eq. \eqref{eq_Lambda_alpha}. We thus complete the proof of Theorem \ref{theorem_characteristic}.
	\end{proof}
	\vspace{-0.2cm}
	
	The key principle revealed by Theorem \ref{theorem_characteristic} is that, \textit{\bfseries all $ \alpha $-competitive deterministic online algorithms share a common and unique $ \alpha $-feasible characteristic selection function $ \psi^{(\alpha)} $. Thus, if the competitive ratio $ \alpha^* $ is such that for any $ \alpha  < \alpha^*  $, $ \psi^{(\alpha)} $ does not exist, and for any $ \alpha  \geq  \alpha^*  $, $ \psi^{(\alpha)} $ exists and is unique, then $ \alpha^* $ is the best-possible competitive ratio of all deterministic online algorithms}. Based on this principle, in the next subsection we give Theorem \ref{theorem_optimality_conditions}, which summarizes the conditions that once satisfied by $ \psi^{(\alpha)} $, then the corresponding competitive ratio $ \alpha $ is optimal.

	\subsection{Optimality Conditions and Proof of Theorem \ref{theorem_optimality}}
	
	Now we are ready to prove the optimality of $ \boldsymbol{\lambda}^* $ given by Theorem \ref{theorem_optimality}.
	
	\begin{theorem}%[\textsc{Optimality Conditions}] 
		\label{theorem_optimality_conditions}
		If  $ \alpha \geq 1 $ is such that the characteristic selection function $ \psi^{(\alpha)}  $  has $ \bar{k} - \tau^{(\alpha)} $ non-zero points, denoted by $ \Omega^{(\alpha)} = \{\omega_1^{(\alpha)}, \omega_2^{(\alpha)}, \cdots, \omega_{\bar{k}-\tau^{(\alpha)} }^{(\alpha)} \}$, and satisfy
		\vspace{-0.1cm}
		\begin{equation}\label{eq_optimality_conditions}\vspace{-0.1cm}
			\sum_{\ell=1}^{\bar{k} - \tau^{(\alpha)}} \omega_{\ell}^{(\alpha)}\cdot \psi^{(\alpha)}(\omega_{\ell}^{(\alpha)}) - f(\bar{k}) = \frac{1}{\alpha}f^*(p_{\max}),
		\end{equation}
		then, $ \alpha $ is the best-possible competitive ratio of all deterministic   algorithms. 
	\end{theorem}
	\begin{corollary}
		\label{theorem_optimality_system_of_equations}
		If $ \alpha \geq 1 $ is such that the characteristic selection function $ \psi^{(\alpha)}  $  has $ \bar{k} - \tau^{(\alpha)} $ non-zero points which satisfy the following system of equations:
		%\vspace{-0.1cm}
		\begin{equation}%\vspace{-0.1cm}
			\label{eq_system_equations_omega_alpha}
			\normalfont
			\textbf{(SoNE):}\quad 
			\alpha = \frac{f^{*}(\omega_2^{(\alpha)})}{g(\tau^{(\alpha)} + 1)} =  \frac{f^{*}(\omega_3^{(\alpha)})-f^{*}(\omega_2^{(\alpha)})}{\omega_2^{(\alpha)}-c_{\tau^{(\alpha)}+2}} = \cdots 
			%=  \frac{f^{*}\big(\omega_{\bar{M}-k^{(\alpha)}}^{(\alpha)}\big)-f^{*}\big(\omega_{\bar{M}-k^{(\alpha)}-1}^{(\alpha)}\big)}{\omega_{\bar{M}-k^{(\alpha)}-1}^{(\alpha)} - c_{\bar{M} - 1}} 
			= \frac{f^{*}\big(p_{\max} \big)-f^{*}\big(\omega_{\bar{k}-\tau^{(\alpha)}}^{(\alpha)}\big)}{\omega_{\bar{k}-\tau^{(\alpha)}}^{(\alpha)} - c_{\bar{k}}},
		\end{equation}
		then, $ \alpha $ is the best-possible competitive ratio of all deterministic   algorithms. 
	\end{corollary}
	\begin{proof}
		Theorem \ref{theorem_optimality_conditions}  is proved in Appendix \ref{proof_of_optimality_conditions}.  Corollary \ref{theorem_optimality_system_of_equations} directly follows 
		Theorem \ref{theorem_optimality_conditions} since Eq. \eqref{eq_system_equations_omega_alpha} is equivalent to Eq. \eqref{eq_optimality_conditions} and Eq. \eqref{eq_omega_alpha} after some basic manipulations.  The intuition of Theorem \ref{theorem_optimality_conditions} is that, a smaller $ \alpha $ tends to increase the value of $ L^{(\alpha)} $, i.e., the number of non-zero points increases (at least non-decreasing) when $ \alpha $ is smaller. By this property, when $ \alpha_* $ is small enough such that $ L^{(\alpha_*)} $ achieves its maximum, i.e., $ L^{(\alpha_*)} = \bar{k} - \tau^{(\alpha_*)} $, and at the same time, the non-zero points are such that Eq. \eqref{eq_optimality_conditions} or Eq. \eqref{eq_system_equations_omega_alpha} hold, then it is impossible to have another $ \alpha $-feasible characteristic selection function $ \psi^{(\alpha)} $ with $ \alpha < \alpha_* $ since otherwise this would violate the uniqueness of solutions to the system of equations $ \textsf{SoE}(\boldsymbol{\chi}^{(\tau)}) $ in Eq. \eqref{eq_system_of_equations_X}.
		Note that Eq. \eqref{eq_system_equations_omega_alpha} shares the same structure as Eq. \eqref{eq_system_of_equations_X}.
		We refer to Eq. \eqref{eq_system_equations_omega_alpha} as the \textbf{system of necessary equations} (\textbf{SoNE}).
	\end{proof}

	\noindent\textbf{Proof of Theorem \ref{theorem_optimality}}. The \textbf{SoNE} in Eq. \eqref{eq_system_equations_omega_alpha} differs from the \textbf{SoSE} in Eq. \eqref{eq_system_of_equations} in notation  only: the non-zero points $ \{\omega_2^{(\alpha)}, \omega_3^{(\alpha)}, \cdots, \omega_{\bar{k}-\tau^{(\alpha)} }^{(\alpha)} \} $ correspond to the thresholds $ \{\lambda_{\tau^{(\alpha)} + 1}^*, \lambda_{\tau^{(\alpha)} + 2}^*, \cdots, \lambda_{\bar{k}-1}^*\} $ in an element-wise manner. %namely $p_{k^{(\alpha)} + \ell}^* = \omega_{\ell}^{(\alpha)}, \quad \ell \in \{ 1,2,\cdots, \bar{M}-k^{(\alpha)} \}$. 
	Thus, if the admission threshold $ \boldsymbol{\lambda}^* $ is designed based on Theorem \ref{theorem_optimality}, then the threshold policy $ \textsf{TOS}_{\boldsymbol{\lambda}^*} $ is $ \textsf{CR}_f^*(\rho,k) $-competitive, which is optimal since any $ \alpha < \textsf{CR}_f^*(\rho,k) $ will lead to the non-existence of $ \psi^{(\alpha)} $. We thus complete the proof of Theorem \ref{theorem_optimality}.
	
	\vspace{+0.2cm}
	\begin{remark}[Relationship between \textbf{SoNE} and \textbf{SoSE}]
		The \textbf{SoNE} in Eq. \eqref{eq_system_equations_omega_alpha} coincides with the \textbf{SoSE} in Eq. \eqref{eq_system_of_equations}, and they both share the same structure as $\normalfont \textsf{SoE}(\boldsymbol{\chi}^{(\tau)}) $ given by Eq. \eqref{eq_system_of_equations_X}. However, as demonstrated by our proof in this section, both \textbf{SoSE} and \textbf{SoNE} are developed separately without relying on each other, and they respectively contribute to the sufficiency and necessity of the competitiveness, uniqueness, and optimality of $\normalfont \textsf{TOS}_{\boldsymbol{\lambda}^*} $.
	\end{remark}

\section{Proof of Theorem \ref{theorem_hardness_results_general}}
\label{proof_theorem_hardness_results}
%It is easy to observe that $ \textsf{CR}_f^*(\rho, M) $ is strictly decreasing in $ M $. Therefore, an interesting question is what will happen when $ M $ approaches infinity?  In practice, this corresponds to the scenario when the weight of each item is much smaller than the total capacity of the knapsack, e.g., the resource requirement of each computing job is negligible when compared to the total capacity of a large-scale cloud platform. Motivated by this, this subsection aims to understand the asymptotic performance of our threshold-based online algorithm in the case of $ M \rightarrow +\infty $.
In this section, we show how to derive the lower bound $ \textsf{CR}_f^{\textsf{lb}}(\rho,k) $ in Theorem \ref{theorem_hardness_results_general}. We remark that the proof below is heavily related to Appendix \ref{sec_proof_of_theorem_optimality}, especially the concept of selection functions in Definition \ref{def_packing_functions} and the  $ \varepsilon $\textsf{-instance} described in Definition \ref{def_I_n}. Thus, readers are suggested to have an overview of  Appendix \ref{sec_proof_of_theorem_optimality} first before proceeding with the following proof.

\subsection{Average Selection Function: Definition and Existence}

Based on Definition \ref{def_packing_functions}, we define \textit{average selection functions} below.

\begin{definition}[Average selection functions]\label{def_packing_functions_average}
	An average selection function $ \bar{\psi}(p): [p_{\min},p_{\max}] \rightarrow [1, \bar{k}]$ is a mapping from $ p\in [p_{\min},p_{\max}] $ to a non-negative real number within $ [0, \bar{k}]$, denoting the \textbf{average number} of selected buyers whose offered price is $ p $.
\end{definition}

Given an arrival instance, the realization of any randomized algorithm can be fully characterized by the resulting average selection function. Definition \ref{def_packing_functions_average} implies that an average selection function $ \bar{\psi} $ can be considered the expectation of $ \psi $ over possible randomness of the algorithm: $ \bar{\psi}(p) = \mathbb{E}\left[\psi(p) \right]$, where $ \psi $ is the selection function discussed in Section \ref{section_def_packing_function}. By definition, the core difference between $ \bar{\psi} $ and $ \psi $  is that \textit{the average selection function $ \bar{\psi}  $ ranges from 0 to $ \bar{k} $ continuously while $ \psi $ only takes integer values in $ [\bar{k}] $}.
%This is the core difference between a randomized algorithm and its deterministic counterpart.

Based on Definition \ref{def_packing_functions_average}, we give Proposition \ref{theorem_necessary_random} below to show the existence of an average selection function whenever there exists an $ \alpha $-competitive randomized algorithm.
\begin{proposition} \label{theorem_necessary_random}
	If there is an $ \alpha $-competitive randomized online algorithm, then there exists an average selection function $ \bar{\psi} $ such that
	\begin{subequations}\label{eq_necessary_random}
		\begin{align} \normalfont  
			\int_{p_{\min}}^p \eta \bar{\psi}(\eta)d\eta - f\Big(\int_{p_{\min}}^p \bar{\psi}(\eta)d\eta\Big) \geq  \frac{1}{\alpha} f^{*}(p), \quad \forall p \in (p_{\min},p_{\max}], \label{eq_necessary_random_1}\\
			\bar{\psi}(p_{\min}) \in \Big[\hat{g}^{-1}\Big(\frac{1}{\alpha}f^*(p_{\min})\Big), \ubar{k} \Big], \ \int_{p_{\min}}^{p_{\max}}  \bar{\psi}(\eta)d\eta \leq \bar{k}, \label{eq_necessary_random_2}
		\end{align}
	\end{subequations}
	where $ \hat{g}^{-1} $ is the inverse of $ \hat{g}(y) \triangleq  p_{\min}y - f(y) $ for $ y\in [0,\bar{k}] $.
\end{proposition}
\begin{proof} 
	Proposition \ref{theorem_necessary_random} is an extension of Proposition \ref{theorem_necessary} to consider randomized algorithms. Similar to our proof of Proposition \ref{theorem_necessary}, here we also leverage the fact that if there is an $ \alpha $-competitive randomized algorithm \textsf{ALG}, then \textsf{ALG} must achieve at least $ \frac{1}{\alpha} $ portion of optimum under all possible arrival instances, including the $ \varepsilon $\textsf{-Instance} $ \mathcal{I}_{n}^{(\varepsilon)} $ for all $ n \in  [N_{\varepsilon}] $ given by Definition \ref{def_I_n}. 
	
	Let us first consider that the arrival instance is given by $ \mathcal{I}_1^{(\epsilon)} $. Suppose $ \bar{\psi}_1  = \mathbb{E}[\psi_1] $ is the average selection function corresponds to an $ \alpha $-competitive randomized algorithm \textsf{ALG} under $ \mathcal{I}_1^{(\epsilon)} $. Thus, the expected profit achieved by \textsf{ALG} is given by  $ \mathbb{E}\big[\textsf{ALG}(\mathcal{I}_1^{(\epsilon)})\big] =  \mathbb{E}\big[p_{\min} \psi_1(p_{\min}) - f(\psi_1(p_{\min}))\big] \leq p_{\min} \bar{\psi}_1(p_{\min}) - f\big(\bar{\psi}_1(p_{\min})\big) = \hat{g}\big(\bar{\psi}_1(p_{\min})\big) $, where we use the Jensen's inequality. The algorithm \textsf{ALG} being $ \alpha $-competitive indicates that 
	\begin{equation*}%\label{key}
		\hat{g}\Big(\bar{\psi}_1(p_{\min})\Big) \geq \mathbb{E}\big[\textsf{ALG}(\mathcal{I}_1^{(\epsilon)})\big] \geq \frac{1}{\alpha} \textsf{OPT}(\mathcal{I}_1^{(\epsilon)}) = \frac{1}{\alpha} f^{*}(p_{\min}),
	\end{equation*}
	where $ \textsf{OPT}(\mathcal{I}_1^{(\epsilon)}) = f^{*}(p_{\min}) $ denotes the optimal profit achieved in the offline setting under the arrival instance $ \mathcal{I}_1^{(\epsilon)} $.
	Thus, there must exist an average selection function $ \bar{\psi}_1 $ such that $ \bar{\psi}_1(p_{\min}) \in [\hat{g}^{-1}(\frac{1}{\alpha}f^*(p_{\min})), \ubar{k}] $  as long as there exists an $ \alpha $-competitive randomized  algorithm.
	
	Similarly, let us consider that the arrival instance is given by $ \mathcal{I}_n^{(\varepsilon)} $ for $ n\in \{2,3,\cdots,N^{(\varepsilon)}\} $. The existence of an $ \alpha $-competitive randomized algorithm is equivalent to the existence of at least one average selection function $ \hat{\psi}_n $ so that
	\begin{align*}
		\mathbb{E}\left[\textsf{ALG}\big(\mathcal{I}_n^{(\varepsilon)}\big)\right]  
		= \ & \mathbb{E}\left[ \sum_{\ell=1}^{n} p_{\ell}^{(\varepsilon)} \psi_n\big(p_{\ell}^{(\varepsilon)}\big)  -  f\left( \sum_{\ell=1}^{n}  \psi_n\big(p_{\ell}^{(\varepsilon)}\big) \right)\right] \\
		\leq \ &   \sum_{\ell=1}^{n} p_{\ell}^{(\varepsilon)} \bar{\psi}_n\big(p_{\ell}^{(\varepsilon)}\big) - f\left( \sum_{\ell=1}^{n}  \bar{\psi}_n\big(p_{\ell}^{(\varepsilon)}\big) \right),
	\end{align*}
	where we again use the Jensen's inequality and $ p_{\ell}^{(\varepsilon)} = p_{\min} + (\ell-1)\varepsilon $ for $ \ell\in [n] $ (please refer to Definition \ref{def_I_n}).
	In the offline setting, similar to the case with $ \mathcal{I}_1^{(\varepsilon)} $, the optimal strategy is to reject all but the $ \Gamma(p_n^{(\varepsilon)}) $ copies of identical buyers in the last group $ \mathcal{I}_n^{(\varepsilon)} $, i.e.,  $ \textsf{OPT}(\mathcal{I}_n^{(\varepsilon)}) = p_n^{(\varepsilon)} \Gamma(p_n^{(\varepsilon)}) - f\big(\Gamma(p_n^{(\varepsilon)})\big) = f^*\big(p_n^{(\varepsilon)}\big)$. Thus, the randomized algorithm being $ \alpha $-competitive indicates that $ \mathbb{E} [\textsf{ALG}(\mathcal{I}_n^{(\varepsilon)})]  \geq \frac{1}{\alpha} \textsf{OPT}(\mathcal{I}_n^{(\varepsilon)}) $ and $ \sum_{\ell=1}^{n} \bar{\psi}_n(p_{\ell}^{(\varepsilon)})\leq \bar{k} $ (i.e., at most $ \bar{k} $ buyers will be selected), meaning that there must exist an average selection function $ \bar{\psi}_n(p) $ so that the following condition is satisfied:
	\begin{align*}%\label{eq_psi_bar_n}
		\sum_{\ell=1}^{n} p_{\ell}^{(\varepsilon)} \bar{\psi}_n\big( p_{\ell}^{(\varepsilon)} \big) - f\left( \sum_{\ell=1}^{n} \bar{\psi}_n\big( p_{\ell}^{(\varepsilon)} \big) \right) \geq \mathbb{E}\left[\textsf{ALG}\big(\mathcal{I}_n^{(\varepsilon)}\big)\right] \geq \frac{1}{\alpha} \textsf{OPT}(\mathcal{I}_n^{(\varepsilon)}) = \frac{1}{\alpha} f^{*}\big( p_n^{(\varepsilon)}\big).
	\end{align*}
	
	%Note that the $ \varepsilon $\textsf{-Instance} is constructed in a nested manner, namely, $ \mathcal{I}_{n+1}^{(\varepsilon)} = \mathcal{I}_{n}^{(\varepsilon)} \cup \{\mathcal{A}_{n+1}^{(\varepsilon)}\}$ for all $ n = [N^{(\varepsilon)}-1] $. 
	Similar to our proof of Proposition \ref{theorem_necessary}, the arrival instance to be presented to a randomized algorithm can be $ \mathcal{I}_{n}^{(\varepsilon)} $ for any $ n = [N^{(\varepsilon)}] $, and we argue that it is impossible for an online algorithm to know the information of $ n $ a priori. Thus,  the series of average selection functions $ \bar{\psi}_n $ for any $ n \in  [N^{(\varepsilon)}]  $ must overlap and be the same function $ \bar{\psi} $ that satisfies
	\begin{align}\label{eq_psi_bar}
		\begin{cases}
			\sum_{\ell=1}^{n} p_{\ell}^{(\varepsilon)} \bar{\psi}\big( p_{\ell}^{(\varepsilon)} \big) - f\left( \sum_{\ell=1}^{n} \bar{\psi}\big( p_{\ell}^{(\varepsilon)} \big) \right) \geq \frac{1}{\alpha} f^{*}\big( p_n^{(\varepsilon)}\big),\bigskip \\
			\bar{\psi}(p_{\min}) \in \left[\hat{g}^{-1}\Big(\frac{1}{\alpha}f^*(p_{\min})\Big), \ubar{k} \right], \ \sum_{\ell=1}^{n} \bar{\psi}_n\big(p_{\ell}^{(\varepsilon)}\big) \leq \bar{k},
		\end{cases} 
		n \in  [N^{(\varepsilon)}-1].
	\end{align}
	Since Eq. \eqref{eq_psi_bar} holds for all $ \varepsilon > 0 $, it holds when $ \varepsilon \rightarrow 0 $ as well, 
	leading to the integral inequalities in Eq. \eqref{eq_necessary_random_1} and Eq. \eqref{eq_necessary_random_2}. We thus complete the proof of Proposition \ref{theorem_necessary_random}. 
\end{proof}

\subsection{Characteristic Average Selection (CAS) Function}
Similar to our proof of Theorem \ref{theorem_optimality} in Appendix \ref{section_characteristic_packing_function}, an important step in deriving the lower bound  $ \textsf{CR}_f^{\textsf{lb}}(\rho,k) $ is the construction of a special average selection function, termed \textit{characteristic average selection (CAS)} function, given in Theorem \ref{theorem_characteristic_random} below.

\begin{theorem}\label{theorem_characteristic_random}
	If there is an $ \alpha $-competitive randomized online algorithm, then there exists a CAS function $ \bar{\psi}^{(\alpha)} $ such that
	\begin{subequations}
		\begin{align} \normalfont  
			\int_{p_{\min}}^p \eta \bar{\psi}^{(\alpha)}(\eta)d\eta - f\Big(\int_{p_{\min}}^p \bar{\psi}^{(\alpha)}(\eta)d\eta\Big) =  \frac{1}{\alpha} f^{*}(p), \quad \forall p \in (p_{\min},p_{\max}],\label{eq_necessary_inequality_random}\\
			\bar{\psi}^{(\alpha)}(p_{\min}) = \hat{g}^{-1}\Big(\frac{1}{\alpha}f^*(p_{\min})\Big), \ \int_{p_{\min}}^{p_{\max}}  \bar{\psi}^{(\alpha)}(\eta)d\eta \leq \bar{k}.\label{eq_necessary_boundary_random}
		\end{align}
	\end{subequations}
	%Moreover, $ \psi} $
%where $ \beta\in [g^{-1}(\frac{1}{\alpha}f^*(p_{\min})), \ubar{k}]  $. 
\end{theorem}
\begin{proof}
Theorem \ref{theorem_characteristic_random} is similar to Theorem \ref{theorem_characteristic}. The difference is that here the output of the CAS function $ \bar{\psi}^{(\alpha)} $ is continuous  while that of $ \psi^{(\alpha)} $ is discrete. The intuition is that, whenever there exists a feasible  average selection function $ \bar{\psi} $ that satisfies  Eq. \eqref{eq_necessary_random}, then we can always push $ \bar{\psi}(p) $ down to its minimum  for all $ p\in [p_{\min},p_{\max}] $ so that the equality holds. Here we skip the details for brevity. 
\end{proof}

Notice that, the CAS function $ \bar{\psi}^{(\alpha)} $ cannot be written in analytical forms, and thus differs from the characteristic selection function $ \psi^{(\alpha)} $ in Theorem \ref{theorem_characteristic}. To derive the expression of $ \bar{\psi}^{(\alpha)} $, let us first give the following definition.

\begin{definition}%[Average Cumulative selection function]
\label{def_average_cumulative_packing_function}
An average cumulative selection function $   \bar{\Psi}(p): [p_{\min},p_{\max}] \rightarrow [1,k]$ is defined as the integration of $ \bar{\psi} $ over interval $ [p_{\min},p] $ as follows:
\begin{align}\label{eq_stair_case_function}
	\bar{\Psi}(p) \triangleq  \bar{\psi}(p_{\min})\cdot \mathds{1}_{\{p=p_{\min}\}} + 
	\left(\int_{p_{\min}}^p \bar{\psi}(\eta) d\eta\right)\cdot \mathds{1}_{\{p\in (p_{\min},p_{\max}]\}},
	\quad p\in [p_{\min},p_{\max}].
\end{align}
\end{definition}

Definition \ref{def_average_cumulative_packing_function} implies that $ \bar{\psi}(p) = \bar{\Psi}'(p) $ for all $ p\in [p_{\min},p_{\max}] $, where the derivative $ \bar{\Psi}'(p) $ at the lower and upper boundaries are defined as the right and left derivative, respectively. Based on Theorem \ref{theorem_characteristic_random} and Definition \ref{def_average_cumulative_packing_function}, we give Corollary \ref{theorem_optimality_random} below.

\begin{corollary} 
\label{theorem_optimality_random}
If there is an $ \alpha $-competitive randomized online algorithm, then there exists a strictly-increasing function $ \bar{\Psi}^{(\alpha)} $ such that
\begin{align}\label{eq_bvp_Psi}
	\begin{cases}
		%\frac{d \bar{\Psi}^{(\alpha)}}{dv} =  \frac{1}{\alpha}\cdot \frac{\Gamma(p)}{v - f'(\Psi(p))}, 
		\frac{d \bar{\Psi}^{(\alpha)}}{dp} =  \frac{1}{\alpha}\cdot \frac{\sum_{i = \ubar{k}}^{\bar{k}} i \cdot \mathds{1}_{\left\{p\in [c_i, c_{i+1})\right\}} }{p - f'\big(\bar{\Psi}^{(\alpha)}(p)\big)}, \quad  p \in (p_{\min},p_{\max}), \bigskip \\
		\bar{\Psi}^{(\alpha)}(p_{\min}) = \hat{g}^{-1}\big(\frac{1}{\alpha}f^*(p_{\min})\big),  \bar{\Psi}^{(\alpha)}(p_{\max}) \leq  \bar{k}.
	\end{cases}
\end{align}
Moreover, when $ \bar{\Psi}^{(\alpha)}(p_{\max}) = \bar{k} $, the corresponding $ \alpha $ is the competitive ratio that no randomized algorithm can outperform.
\end{corollary}
\begin{proof} 
This corollary largely follows Theorem \ref{theorem_characteristic_random}. By the definition of $ \bar{\Psi} $, we denote by $ \bar{\Psi}^{(\alpha)} $ the average cumulative selection function corresponding to $ \bar{\psi}^{(\alpha)} $. Based on Eq. \eqref{eq_necessary_inequality_random}, we have
\begin{align*}
	\int_{p_{\min}}^p \eta \bar{\psi}^{(\alpha)}(\eta)d\eta - f\Big( \bar{\Psi}^{(\alpha)}(p) \Big) =  \frac{1}{\alpha} f^{*}(p).
\end{align*}
Taking derivative of both sides leads to $ p \frac{d \bar{\Psi}^{(\alpha)}}{dp} - f'(\bar{\Psi}^{(\alpha)}(p)) \frac{d \bar{\Psi}^{(\alpha)}}{dp} = \frac{1}{\alpha} f^{*'}(p) $. 
Since $ f^{*} $ is piecewise linear, the derivative of $ f^{*} $ is piecewise constant (Lemma \ref{theorem_conjugate}). Therefore, we can write the above derivative in different segments, leading to the following combination of indicator functions:
\begin{align*}
	\frac{d \bar{\Psi}^{(\alpha)}}{dp} =  \frac{1}{\alpha}\cdot \frac{1}{p - f'\big(\bar{\Psi}^{(\alpha)}(p)\big)} \cdot \sum_{i = \ubar{k}}^{\bar{k}} i \cdot \mathds{1}_{\left\{p\in [c_i, c_{i+1})\right\}}.
\end{align*}
The two boundary conditions in Eq. \eqref{eq_bvp_Psi} trivially follow the definition of $ \bar{\Psi}^{(\alpha)} $ and Eq. \eqref{eq_necessary_boundary_random}.

%where we leverage the piecewise linearity of $ f^* $ in Lemma \ref{theorem_conjugate}.  We thus complete the proof of Eq. \eqref{eq_bvp}.

Similar to our proof of the two optimality conditions in Theorem \ref{theorem_optimality_conditions}, no online algorithm can achieve a competitive ratio better than $ \alpha_* $  if $ \alpha_* $ is such that $ \bar{\Psi}^{(\alpha_*)}(p_{\max}) = \int_{p_{\min}}^{p_{\max}} \bar{\psi}^{(\alpha_*)}(\eta) d\eta = \bar{k} $ (similar to the optimality condition of $ L^{(\alpha)} = \bar{k} - \tau^{(\alpha)} $ in Theorem \ref{theorem_optimality_conditions}). This is because, if $ \bar{\Psi}^{(\alpha_*)}(p_{\max}) = \bar{k} $, then for any $ \alpha < \alpha_* $, we can prove that there exists no such a function $ \bar{\Psi}^{(\alpha)} $ that satisfies Eq. \eqref{eq_bvp_Psi}. As a result, any $ \alpha < \alpha_* $ will inevitably lead to  the non-existence of $ \bar{\psi}^{(\alpha)} $, which violates the necessary conditions given in Theorem \ref{theorem_characteristic_random}.  Corollary \ref{theorem_optimality_random} thus follows.
\end{proof}

\subsection{Deriving the Lower Bound $ F(\gamma^{(1)}) $ (Proof of Theorem \ref{theorem_hardness_results_general})}
Now we are ready to prove why  no randomized algorithm is $ \big(\textsf{CR}_f^{\textsf{lb}}(\rho,k) -\epsilon\big)$-competitive for any $ \epsilon > 0 $ with $ \textsf{CR}_f^{\textsf{lb}}(\rho,k) $ being given by Eq. \eqref{eq_CR_lb_general}, that is, $ \textsf{CR}_f^{\textsf{lb}}(\rho,k) = F(\gamma^{(1)}) $.

Since $ \bar{\Psi}^{(\alpha)}(p) $ is strictly-increasing over $ p\in [p_{\min}, p_{\max}] $, its inverse function exists and is well-defined. Let us denote by $ \bar{\Phi}^{(\alpha)} = \bar{\Psi}^{(\alpha),-1} $, namely, $ \bar{\Phi}^{(\alpha)} $  and $  \bar{\Psi}^{(\alpha)}  $ are inverse to each other. In the following we drop the superscript `$ (\alpha) $' and directly write $ \bar{\Phi} $ to simplify the notations (we still keep the bar notation to indicate that this is related to the average selection function $ \bar{\psi} $). Based on Eq. \eqref{eq_bvp_Psi}, we know that there exists a unique strictly-increasing function $ \bar{\Phi}(y) $ such that:
\begin{equation}\label{eq_bvp_Phi}
\begin{aligned}
	\normalfont
	\begin{cases}
		\bar{\Phi}'(y) = \normalfont  \frac{\alpha \left(\bar{\Phi}(y) - f'(y)\right) }{\sum_{i = \ubar{k}}^{\bar{k}} i \cdot \mathds{1}_{\left\{\bar{\Phi}(y)\in [c_i, c_{i+1})\right\} } }, \qquad y\in \left( \hat{g}^{-1}\big(\frac{f^*(p_{\min})}{\alpha}\big), \bar{k} \right),\bigskip \\
		\bar{\Phi}\left(\hat{g}^{-1}\big(\frac{f^*(p_{\min})}{\alpha} \big) \right) = p_{\min}, \quad \bar{\Phi}(\bar{k}) =  p_{\max}.
	\end{cases}
\end{aligned}
\end{equation}
Eq. \eqref{eq_bvp_Phi} can be derived from Eq. \eqref{eq_bvp_Psi} based on the inverse relationship between $ \bar{\Phi} $ and $ \bar{\Psi} $. We omit the details as the manipulations are elementary. 

%For each segment of $ \bar{\Phi}(y)\in [c_m, c_{m+1}) $, we can solve the above differential equation:
%\begin{align*}
%& \begin{cases}
%	\bar{\Phi}(y) = \exp\left(\frac{\alpha y}{\ubar{k}}\right)\cdot \left( \beta^{(\ubar{k})}  - \frac{\alpha}{\ubar{k}}\int_0^y \frac{f'(\eta)}{\exp\left(\alpha\eta/\ubar{k}\right)} d\eta \right)\bigskip \\
%\bar{\Phi}(\gamma^{(\ubar{k})}) = p_{\min},  \bar{\Phi}(\gamma^{(\ubar{k}+1)}) = c_{\ubar{k}+1}.
%\end{cases}
%\bigskip \\
%& \begin{cases}
%	\bar{\Phi}(y) = \exp\left(\frac{\alpha y}{m}\right)\cdot \left( \beta^{(m)}  - \frac{\alpha}{m} \int_0^y \frac{f'(\eta)}{\exp\left(\alpha\eta/m\right)} d\eta \right) \bigskip \\
%\bar{\Phi}(\gamma^{(m)}) = c_m,  \bar{\Phi}(\gamma^{(m+1)}) = c_{m+1}.
%\end{cases}
%\forall m\in \{\ubar{k}+1, \cdots, \bar{k}-1\}
%\bigskip \\
%& \begin{cases}
%	\bar{\Phi}(y) = \exp\left(\frac{\alpha y}{\bar{k}}\right)\cdot \left( \beta^{(\bar{k})}  - \frac{\alpha}{\bar{k}}\int_0^y \frac{f'(\eta)}{\exp\left(\alpha\eta/\bar{k}\right)} d\eta \right)\bigskip \\
%\bar{\Phi}(\gamma^{(\bar{k}-1)}) = c_{\min},  \bar{\Phi}(\gamma^{(\bar{k})}) = p_{\max}.
%\end{cases}
%\end{align*}

We emphasize that Eq. \eqref{eq_bvp_Phi} is equivalent to Eq. \eqref{eq_bvp_Psi} in the sense that a function $ \bar{\Phi} $ that satisfies Eq. \eqref{eq_bvp_Phi} means that $ \bar{\Phi}^{-1} $ satisfies Eq. \eqref{eq_bvp_Psi}, and vice versa. The purpose of performing such an equivalent transformation is because Eq. \eqref{eq_bvp_Phi} is much easier to deal with from the perspective of solving differential equations. Specifically, the solution to the differential equation in Eq. \eqref{eq_bvp_Phi} is given by
\begin{align*}
\bar{\Phi}(y) = \sum_{i = \ubar{k}}^{\bar{k}} \left( \exp\left(\frac{\alpha y}{i}\right)\cdot \Big( \beta^{(i-\ubar{k}+1)}  - \frac{\alpha}{i} \int_0^y \frac{f'(\eta)}{\exp\left( \frac{\alpha}{i}\eta \right)} d\eta \Big) \right)\cdot \mathds{1}_{\left\{\bar{\Phi}(y)\in [c_i, c_{i+1}) \right\}},
\end{align*}
where $ \beta^{(i-\ubar{k}+1)} $ is any real constant. Note that here the superscript `$ (i-\ubar{k}+1) $' is purposely chosen so that the constant term starts from $ \beta^{(1)} $ as $ i $ starts from $\ubar{k} $.

As argued by Corollary \ref{theorem_optimality_random}, \textit{\bfseries if we can find an $ \alpha $ such that there exists a unique strictly-increasing function $ \bar{\Phi} $ that satisfies Eq. \eqref{eq_bvp_Phi}, then this $ \alpha $ is the lower bound of competitive ratios that no randomized algorithm can outperform}. Thus, we focus on demonstrating how to compute $ \alpha $  based on Eq. \eqref{eq_bvp_Phi}, or more specifically, obtain the value of $ \alpha $  such that $ \bar{\Phi}(y) $ is uniquely determined as follows: 
\begin{equation*}
\begin{aligned}
	\normalfont
	\begin{cases}
		\bar{\Phi}(y) = \sum_{i = \ubar{k}}^{\bar{k}} \left( \exp\left(\frac{\alpha y}{i}\right)\cdot \left( \beta^{(i-\ubar{k}+1)} - \frac{\alpha}{i} \int_0^y \frac{f'(\eta)}{\exp\left( \frac{\alpha}{i}\eta \right)} d\eta \right) \right)\cdot \mathds{1}_{\left\{\bar{\Phi}(y)\in [c_i, c_{i+1})\right\}}, \bigskip \\
		\bar{\Phi}\left(\hat{g}^{-1}\big(\frac{f^*(p_{\min})}{\alpha} \big) \right) = p_{\min}, \quad \bar{\Phi}(\bar{k}) =  p_{\max},
	\end{cases}
\end{aligned}
\end{equation*}
where $ y\in ( \hat{g}^{-1}(\frac{f^*(p_{\min})}{\alpha}), \bar{k} ) $. Mathematically, $ \bar{\Phi}(y) $ is a piecewise continuous function with $ \bar{k}-\ubar{k} + 1 $ segments. Therefore, there exist a sequence of $ \bar{k}-\ubar{k} + 2 $ positive real numbers that define these $ \bar{k}-\ubar{k} + 1 $ consecutive and non-overlapping segments or intervals. If we define $ \gamma^{(1)} \triangleq \hat{g}^{-1}\big(\frac{1}{\alpha}f^*(p_{\min})\big) $, namely,
\begin{align*}
\hat{g}(\gamma^{(1)}) = p_{\min}\gamma^{(1)} - f(\gamma^{(1)}) = \frac{1}{\alpha}f^*(p_{\min}) \quad \text{  or  } \quad \alpha  =  \frac{f^{*}(p_{\min})}{p_{\min}\gamma^{(1)} - f(\gamma^{(1)})}  =   F(\gamma^{(1)}),
\end{align*}
then $ y = \gamma^{(1)} $ is the first endpoint. Meanwhile, $ y = \bar{k} $ must be the $ (\bar{k}-\ubar{k} + 2) $-th endpoint (i.e., the last one). To define the $ \bar{k}-\ubar{k}  $ endpoints  between $ \gamma^{(1)} $ and $ \bar{k} $, we introduce a series of positive real numbers as follows:
\begin{align*}
0 < \gamma^{(1)} < \gamma^{(2)} < \cdots < \gamma^{(\bar{k}-\ubar{k} + 1)} < \gamma^{(\bar{k}-\ubar{k} + 2)} \triangleq \bar{k},
\end{align*}
where $ \gamma^{(\ell)} $ denotes the $ \ell $-th endpoint for $ \ell \in [\bar{k}-\ubar{k} + 2] $. For technical reasons, here we define $ \gamma^{(\bar{k}-\ubar{k} + 2)} \triangleq \bar{k} $ to simplify the notations in our following analysis.
Since $ \bar{\Phi}(y) $ is continuous at its endpoints, for $ \ell \in [\bar{k}-\ubar{k} + 2] $,  $ \gamma^{(\ell)} $ satisfies:
\begin{itemize}
\item First Endpoint: $ \ell = 1 $. Based on the first boundary condition in Eq. \eqref{eq_bvp_Phi}, we have $ \bar{\Phi}(\gamma^{(1)}) = p_{\min} $ and thus the following equation of $ \gamma^{(1)} $ and $ \beta^{(1)} $ holds:
\begin{align}\label{eq_Phi_1}
	\exp\left(\frac{ \alpha \gamma^{(1)} }{\ubar{k}} \right) \cdot \Big( \beta^{(1)}  - \frac{\alpha}{\ubar{k}}\int_0^{\gamma^{(1)}} \frac{f'(\eta)}{\exp\big(\frac{\alpha \eta}{\ubar{k}}\big)} d\eta \Big) = p_{\min}.
\end{align}

\item Middle Endpoints: $ \ell \in \{2,3,\cdots, \bar{k}-\ubar{k} + 1 \} $. Based on the continuity of $ \bar{\Phi}(y) $ at its endpoints and $ \bar{\Phi}(\gamma^{(\ell)}) = c_{\ell + \ubar{k} - 1} $, the following two equations of $ \gamma^{(\ell)},  \beta^{(\ell-1)} $, and $ \beta^{(\ell)} $ hold  for all $ \ell \in \{2,3,\cdots, \bar{k}-\ubar{k} + 1 \} $:
\begin{align}\label{eq_Phi_2}
	\begin{cases}
		\exp\left(\frac{ \alpha \gamma^{(\ell)} }{\ell + \ubar{k} - 2}\right)\cdot \Big( \beta^{(\ell-1)} - \frac{\alpha}{\ell + \ubar{k} - 2} \int_0^y \frac{f'(\eta)}{\exp\left( \frac{\alpha \eta }{ \ell + \ubar{k} - 2 } \right)} d\eta \Big) = c_{\ell + \ubar{k} - 1}, \\ %= q^{(\ell)}, \\
		\\
		\exp\left(\frac{ \alpha \gamma^{(\ell)} }{ \ell + \ubar{k} - 1 }\right)\cdot \Big( \beta^{(\ell)} - \frac{\alpha}{ \ell + \ubar{k} - 1 } \int_0^y \frac{f'(\eta)}{\exp\left( \frac{ \alpha \eta }{ \ell + \ubar{k} - 1 } \right)} d\eta \Big) =  c_{\ell + \ubar{k} - 1}. %q^{(\ell)},
	\end{cases}
\end{align}

\item Last Endpoint: $ \ell = \bar{k} - \ubar{k} + 2  $. The second boundary condition in Eq. \eqref{eq_bvp_Phi} implies that $ \bar{\Phi}(\gamma^{(\bar{k} - \ubar{k} + 2)}) = \bar{\Phi}(\bar{k}) = p_{\max} $, and thus the following equation of $ \gamma^{(\bar{k} - \ubar{k} + 2)} $ and $ \beta^{(\bar{k} - \ubar{k} + 1)} $ holds
\begin{align}\label{eq_Phi_3}
	\exp\Big(\frac{\alpha \gamma^{(\bar{k} - \ubar{k} + 2)} }{\bar{k}} \Big)\cdot \bigg( \beta^{(\bar{k} - \ubar{k} + 1)}  - \frac{\alpha}{\bar{k}}\int_0^{\gamma^{(\bar{k} - \ubar{k} + 2)}} \frac{f'(\eta)}{\exp\big(\frac{\alpha \eta}{\bar{k}}\big)} d\eta \bigg) = p_{\max}.
\end{align}
\end{itemize}

For ease of notation, let us define $ q^{(1)}  \triangleq p_{\min}$,  $ q^{(\ell)} \triangleq c_{\ubar{k} + \ell -1} $ for $ \ell \in \{2,3,\cdots, \bar{k}-\ubar{k} + 1\}  $,  and $ q^{(\bar{k}-\ubar{k} + 2)} \triangleq p_{\max} $. We can manipulate Eq. \eqref{eq_Phi_1} -- Eq. \eqref{eq_Phi_3} to eliminate $ \beta^{(\ell)} $ for $ \ell = [\bar{k}-\ubar{k} + 1]  $, and obtain the following equations of $ \gamma^{(\ell)} $  for $ \ell \in  [\bar{k}-\ubar{k} + 1]  $:
\begin{align}
%\begin{cases}
& \frac{q^{(\ell + 1)}(\ubar{k} + \ell -1)}{\exp\left(\frac{\alpha}{\ubar{k} + \ell -1} \gamma^{(\ell + 1)} \right)} - \frac{q^{(\ell)}(\ubar{k} + \ell -1)}{ \exp\left( \frac{\alpha}{\ubar{k} + \ell -1} \gamma^{(\ell)} \right) }  
= \int_{\gamma^{(\ell)}}^{\gamma^{(\ell+1)}} \frac{\alpha f'(y)}{\exp\left(\frac{\alpha y}{\ubar{k} + \ell -1}\right)} dy, \quad \forall \ell \in [\bar{k} - \ubar{k} + 1].
\bigskip \label{eq_system_of_equations_abg_1}
%=  \exp\left(\frac{\alpha}{\ubar{k} + \ell -1}\right)\cdot \Big( \Lambda^{(\alpha)}\big(\gamma^{(\ell+1)}, \ubar{k} + \ell -1\big) - \Lambda^{(\alpha)}\big(\gamma^{(\ell)}, \ubar{k} + \ell -1\big)\Big), 
%& \exp\left(\frac{\alpha \gamma^{(\ell+1)}}{\big(\ubar{k} + \ell -1\big)\big(\ubar{k} + \ell\big)} \right) = \frac{ \beta^{(\ell+1)}  - \Lambda^{(\alpha)}\Big(\gamma^{(\ell+1)}, \ubar{k} + \ell\Big) }{ \beta^{(\ell)} - \Lambda^{(\alpha)}\Big(\gamma^{(\ell+1)}, \ubar{k} + \ell - 1\Big) }, \bigskip \label{eq_system_of_equations_abg_2}\\
%\end{cases}
\end{align}
Substituting $ \alpha = \frac{f^{*}(p_{\min})}{p_{\min}\gamma^{(1)} - f(\gamma^{(1)})}  =  F(\gamma^{(1)}) $ into Eq. \eqref{eq_system_of_equations_abg_1} leads to Eq. \eqref{eq_beta_gamma} in Theorem \ref{theorem_hardness_results_general}. We emphasize that since $ \bar{\Phi} $ is strictly-increasing, the sequence of increasing positive real numbers $ 0 < \gamma^{(1)} < \gamma^{(2)} < \cdots < \gamma^{(\bar{k}-\ubar{k} + 1)} < \bar{k} $ always exists.  
%\begin{align} \label{eq_system_of_equations_abg} 
%	\frac{p^{(\ell + 1)}(\ubar{k} + \ell -1)}{\exp\big(\frac{F(\gamma^{(1)})}{\ubar{k} + \ell -1}\gamma^{(\ell + 1)}\big)}  - 
%	\frac{p^{(\ell)}(\ubar{k} + \ell -1)}{ \exp\big(\frac{F(\gamma^{(1)})}{\ubar{k} + \ell -1} \gamma^{(\ell)} \big) } 
%	= \int_{\gamma^{(\ell)}}^{\gamma^{(\ell+1)}} \frac{ F(\gamma^{(1)}) f'(y)}{\exp\big(\frac{F(\gamma^{(1)})}{\ubar{k} + \ell -1}  y \big)}  dy, \quad  \ell = [\bar{k} - \ubar{k} + 1],
%\end{align}
As argued by Corollary \ref{theorem_optimality_random}, no randomized algorithm is $ \big(F(\gamma^{(1)}) -\epsilon\big)$-competitive for any $ \epsilon > 0 $. We thus complete the proof of Theorem \ref{theorem_hardness_results_general}.

\vspace{+0.2cm}
\begin{remark}[Computation of $ \gamma^{(1)} $]\label{remark_computation_gamma_1}
    Finding the exact solution to Eq. \eqref{eq_system_of_equations_abg_1} requires solving a system of $ \bar{k} - \ubar{k} + 1 $  equations with $ \bar{k} - \ubar{k} + 1 $ variables (i.e., $ \big\{ \gamma^{(\ell)} \big\}_{\ell = [\bar{k}-\ubar{k} + 1] } $). Note that for each given $ \gamma^{(1)}\in (0,\ubar{k}]$, one can sequentially solve Eq. \eqref{eq_system_of_equations_abg_1} to obtain $\gamma^{(2)}, \gamma^{(3)}, \cdots, \gamma^{(\bar{k}-\ubar{k} + 1)}  $. Thus,  we can perform a one-dimensional bisection search over $ \gamma^{(1)}\in (0,\ubar{k}]$ in the outer loop, and sequentially compute the values of  $\gamma^{(2)}, \gamma^{(3)}, \cdots, \gamma^{(\bar{k}-\ubar{k} + 1)}  $ in the inner loop. In the \textsc{High-Value} case when  $\bar{k} = \ubar{k} = k $, the inner loop vanishes as  $ \gamma^{(1)}\in (0,\ubar{k}]$ becomes the only variable. In this case the system of equations in Eq. \eqref{eq_system_of_equations_abg_1} reduces to a single equation of $ \gamma^{(1)}\in (0,\ubar{k}]$ as follows:
    \begin{equation*}
		\frac{p_{\max} k}{\exp(\alpha)} - \frac{p_{\min} k }{\exp\left(\gamma^{(1)} \alpha/k\right)} =  \int_{\gamma^{(1)}}^{k} \frac{\alpha f'(y)}{\exp\left(y \alpha /k\right)} dy,
	\end{equation*}
    which is exactly the same as Eq. \eqref{eq_gamma} after substituting $\alpha  = F(\gamma^{(1)})$. It is also worth noting that the monotonicity property of these $ \bar{k} - \ubar{k} + 1 $ variables (i.e., $ 0 < \gamma^{(1)} < \gamma^{(2)} < \cdots < \gamma^{(\bar{k}-\ubar{k} + 1)} <  \bar{k}  $) guarantees that the bisection search over $ \gamma^{(1)}\in (0,\ubar{k}]$ in the outer loop always convergences.
\end{remark}

\newpage
\section{Proof of Theorem \ref{theorem_asymptotic}}
\label{proof_of_theorem_asymptotic}
In this section, we prove the asymptotic properties given by Theorem \ref{theorem_asymptotic}. Before proving the convergence of $ \textsf{CR}_f^*(\rho,k) $ and $ \textsf{CR}_f^{\textsf{lb}}(\rho,k) $ as $ k\rightarrow +\infty $, we first provide some preliminaries regarding the formulation of OSCC as an optimization problem (in the offline setting).

\subsection{Preliminaries}
%Here, we first give an alternative formulation of Problem \eqref{SWM} by normalizing the total capacity of the knapsack to be 1. Based on the normalized formulation, we define some functions that are key to the proof of the convergence of $  \textsf{CR}_f^*(\rho,M) \rightarrow \underline{\textsf{CR}}_f(\rho) $ and $ \normalfont \textsf{CR}_f^{\textsf{lb}}(\rho,M) \rightarrow \underline{\textsf{CR}}_f(\rho) $ when $ M \rightarrow +\infty $.

%(\textbf{Alternative Formulation: Normalization})  
We first give an alternative formulation of OSCC in the offline setting by normalizing the total production capacity to be 1:
\begin{align}\label{SWM_M}
\underset{x_t\in \{0,1\}}{\textsf{maximize}}\quad    \sum_{t=1}^T p_t x_t - \tilde{f}\left(\frac{1}{k}\sum_{t=1}^T x_t\right).
%\quad  \textsf{subject to}\quad \sum_{t=1}^T \frac{1}{k} x_t \leq 1.
\end{align}
In Problem \eqref{SWM_M}, the demand of each buyer is $ 1/k $ (which was 1 in the initial OSCC formulation, hence the \textit{unit-demand setting} hereinafter) and  $ \tilde{f} $ is a scaled version of $ f $ given in Eq. \eqref{equation_f}. More specifically, $ \tilde{f}(y) $ is monotonically increasing and convex over $ y\in [0,1] $, and satisfies
\begin{align*}
\tilde{f}\left(\frac{i}{k}\right) = f(i) = \sum_{j = 1}^i c_{j}, \quad \forall i = [k].
\end{align*}
In other words, the marginal cost of producing the $ i $-th unit is given by
\begin{align*}
c_i = \tilde{f}\left(\frac{i}{k}\right) - \tilde{f}\left(\frac{i-1}{k}\right), \quad  \forall i = [k].
\end{align*}

Problem \eqref{SWM_M} is equivalent to the presented OSCC in the offline setting. The purpose of performing such an equivalent reformulation is that the demand of each buyer approaches zero when $ k \rightarrow +\infty $. This transforms our initial \textit{unit-demand setting} to the \textit{infinitesimal setting} (i.e., the demand of each buyer is infinitesimal) that has proven to be mathematically more convenient \cite{OKP_Zhou_2008}. In what follows we refer to Problem \eqref{SWM_M} as the \textit{normalized unit-demand setting} so as to differ it from our initial \textit{unit-demand setting}.

(\textbf{Fenchel Conjugate}) The Fenchel conjugate of $ \tilde{f} $ is defined as:
\begin{align*}
\tilde{f}_{\textsf{fc}}^{*}(p) \triangleq  \max_{y \in [0, 1]} \ py    - \tilde{f}(y), \quad p \in [0,+\infty),
\end{align*}
which differs from our previous definition of $ f^* $ in Eq. \eqref{eq_f_star} as here the maximization is taken over a continuous interval of $ [0,1] $, while previously it was done over a discrete set of $ \{0,1,\cdots, k\} $. To differentiate these two different definitions, we define the discrete conjugate of $ \tilde{f} $ by
\begin{align*}
\tilde{f}_{\textsf{dc}}^{*}(p) \triangleq  \max_{i \in \{0,1,\cdots, k\} } \ p \frac{i}{k}    - \tilde{f}\left(\frac{i}{k}\right), \quad p \in [0,+\infty).
\end{align*}
Here, the subscript ``\textsf{fc}" and ``\textsf{dc}" denote ``Fenchel conjugate" and ``discrete conjugate", respectively. 
%Obviously, we have
%\begin{align*}
%  f_{\textsf{dc}}^{*}(p) = \tilde{f}_{\textsf{dc}}^{*}(vM)
%\end{align*}

%Fenchel conjugate (a.k.a  convex conjugate) has many nice properties. For instance, the following property given by Lemma \ref{theorem_f_star_convergence} is key to the upcoming asymptotic analysis.

\begin{lemma}\label{theorem_f_star_convergence}
The discrete conjugate $\normalfont \tilde{f}_{\textsf{dc}}^{*} $ is asymptotically equivalent to $\normalfont \tilde{f}_{\textsf{fc}}^{*} $ when $ k \rightarrow +\infty $.
\end{lemma}

The proof of Lemma \ref{theorem_f_star_convergence} is omitted for brevity. Recall that Lemma \ref{theorem_conjugate} implies that $ \tilde{f}_{\textsf{dc}}^{*} $ is piecewise linear with $ k $ segments. When $ k \rightarrow +\infty $, $ \tilde{f}_{\textsf{dc}}^{*} $ asymptotically converges to $ \tilde{f}_{\textsf{fc}}^{*} $ as a piecewise linear approximation with infinitely-many segments. Meanwhile, it is easy to see that $ \tilde{f}_{\textsf{dc}}^{*}(p) \leq \tilde{f}_{\textsf{fc}}^{*}(p) $ holds for all $ p\in [0,+\infty) $. Thus, $ \tilde{f}_{\textsf{dc}}^{*}(p) $ is a piecewise linear approximation of $ \tilde{f}_{\textsf{fc}}^{*}(p) $ from below, and converges to $ \tilde{f}_{\textsf{fc}}^{*}(p) $  when $ k $ is sufficiently large.

(\textbf{Continuous Min-Profit Function}) In Section \ref{section_preliminaries}, we introduced our definition of the min-profit function $ g(i) = p_{\min}i -f(i) $. Here, we extend the domain of $ g $ from a discrete set to a continuous interval as follows:
\begin{align*}
\tilde{g}(y) \triangleq  p_{\min}y - \tilde{f}(y), \quad y\in [0,1].
\end{align*}
%When $ p_{\min}\leq \tilde{f}'(1) $, it is easy to see that $ \tilde{g}(y) $ is strictly increasing over $ y\in [0, \tilde{f}^{'-1}(p_{\min})]$.  
Note that in Theorem \ref{theorem_necessary_random}, we define $ \hat{g}(y) =  p_{\min}y - f(y)$ for $ y\in [0,\bar{k}] $. Thus, $ \tilde{g} $ is a scaled version of $ \hat{g} $, similar to the relationship between $ \tilde{f} $ and $ f $.

%This corresponds to the case when $ p_{\min} \leq c_M $ in our unit-weight setting and $ g(m) $ is strictly increasing over $ m\in \{0,1,\cdots, \ubar{k}\} $.

(\textbf{Asymptotic Lower Bound} $ \underline{\textsf{CR}}_f(\rho) $) The following Lemma \ref{theorem_Tan2020} follows the results developed by Tan et al. \cite{Tan_ORA_2020}. Specifically, Lemma \ref{theorem_Tan2020} relates the asymptotic lower bound $ \underline{\textsf{CR}}_f(\rho) $ to the existence of a unique solution to an ordinary differential equation (ODE) with boundary conditions, based on which we can calculate the value of $ \underline{\textsf{CR}}_f(\rho) $ as long as $ \tilde{f}, p_{\min} $, and $ p_{\max} $ are given.

\begin{lemma}[Asymptotic Lower Bound]
\label{theorem_Tan2020}
For any given convex setup $ \mathcal{S} $ with $ k \rightarrow +\infty $, the best-possible competitive ratio of all online algorithms (possibly randomized) is   $ \underline{\textsf{CR}}_f(\rho) $, where $ \underline{\textsf{CR}}_f(\rho) \in [1,+\infty) $ is such that there exists a unique strictly increasing function $ \phi $ that satisfies the following first-order ODE with two boundary conditions:
\begin{align}\label{eq_bvp_phi}
	\normalfont
	%\textsf{BVP}_1(\phi)
	\begin{cases}
		\phi'(y) = \underline{\textsf{CR}}_f(\rho) \cdot \frac{\phi(y) - \tilde{f}'(y)}{\tilde{f}_{\textsf{fc}}^{*'}\big(\phi(y)\big)}, \quad y\in \Big( \tilde{g}^{-1}\Big(\frac{ \tilde{f}_{\textsf{fc}}^{*}(p_{\min}) }{ \underline{\textsf{CR}}_f(\rho) } \Big), \vartheta \Big),\bigskip \\
		\phi\Big( \tilde{g}^{-1}\Big(\frac{ \tilde{f}_{\textsf{fc}}^{*}(p_{\min}) }{ \underline{\textsf{CR}}_f(\rho) }\Big) \Big) = p_{\min}, \quad \phi\left( \vartheta \right) = p_{\max}.
	\end{cases}
\end{align}
In Eq. \eqref{eq_bvp_phi}, $ \vartheta\in (0,1] $ is defined as follows:
\begin{align*}
	\vartheta \triangleq  \lim_{k \rightarrow +\infty} \frac{\bar{k}}{k} = \lim_{k \rightarrow +\infty} \frac{\Gamma(p_{\max})}{k}=
	\begin{cases}
		\tilde{f}^{'-1}(p_{\max}) & \text{ if }   p_{\max} < \tilde{f}'(1),\\
		\\
		1   & \text{ if } p_{\max} \geq \tilde{f}'(1),
	\end{cases}
\end{align*}
where $ \tilde{f}^{'-1} $ is the inverse of the derivative of $ \tilde{f} $ (i.e., the inverse of the marginal cost).
\end{lemma}

Lemma \ref{theorem_Tan2020} shows that when  $ k\rightarrow +\infty $, the best-possible competitive ratio $ \underline{\textsf{CR}}_f(\rho) $ is a positive real value within $ [1,+\infty) $ such that Eq. \eqref{eq_bvp_phi} has a unique solution\footnote{Eq. \eqref{eq_bvp_phi} has a solution means that there exists a function that satisfies both the ODE and the boundary conditions.}. As noted earlier,  the fact that the demand of each buyer is infinitesimal is equivalent to the normalized unit-weight setting with $ k \rightarrow +\infty $. Thus, $ \underline{\textsf{CR}}_f(\rho) $ is the best-possible competitive ratio of all online algorithms in our unit-demand setting when the production capacity  $ k $ is sufficiently large. 
%For more details about the knapsack problem in the infinitesimal setting, please refer to  \cite{Tan2020, knapsack2008}.

%(e.g.,  the computation of $ \underline{\textsf{CR}}_f(\rho) $)

\subsection{Convergence of $ \normalfont \textsf{CR}_f^*(\rho,k) \rightarrow \protect\underline{\textsf{CR}}_f(\rho) $}
In this subsection we prove the convergence of $ \normalfont \textsf{CR}_f^*(\rho,k) \rightarrow \underline{\textsf{CR}}_f(\rho) $ when $ k\rightarrow +\infty $.  The key is to show that the \textbf{SoSE} in Theorem \ref{theorem_optimality} (i.e., Eq. \eqref{eq_system_of_equations}) is asymptotically equivalent to a first-order ODE with two boundary conditions when $ k \rightarrow +\infty $. Moreover, this converged ODE is the same as Eq. \eqref{eq_bvp_phi}, leading to the convergence of $ \normalfont \textsf{CR}_f^*(\rho,k) \rightarrow \underline{\textsf{CR}}_f(\rho) $ as $ k\rightarrow +\infty $.

Theorem \ref{theorem_asymptotic} implies that our optimal threshold $ \boldsymbol{\lambda}^* $ and the optimal competitive ratio $ \textsf{CR}_f^*(\rho, k) $ satisfy the following system of equations:
\begin{align*}
\normalfont
\textsf{CR}_f^*(\rho, k) =
\frac{f^{*}\big(\lambda_{\tau+1}^*\big)}{g(\tau+1)}
=  \frac{f^{*}\big(\lambda_{\tau+2}^*\big)-f^{*}\big(\lambda_{\tau+1}^*\big)}{\lambda_{\tau+1}^*-c_{\tau+2}} = \cdots =  \frac{f^{*}(\lambda_{\bar{k}}^*)-f^{*}\big(\lambda_{\bar{k}-1}^*\big)}{\lambda_{\bar{k}-1}^* - c_{\bar{k}}},
\end{align*}
which indicates that the sequence of positive real numbers $ \{ \lambda_i^* \}_{\forall i = \{\tau+1, \tau+2, \cdots, \bar{k} \} } $ satisfy the following recursive equation:
\begin{align*}
f^*(\lambda_i^*) = f^*(\lambda_{i-1}^*) + \textsf{CR}_f^*(\rho, k)\cdot (\lambda_{i-1}^* - c_i).
%\quad m\in \{k+2,k+3,\cdots, M\}.
\end{align*}
Note that this is for the unit-demand setting where the demand of each buyer is one unit. After normalization, the demand of each buyer becomes $ 1/k $, and thus we have
\begin{align}\label{eq_1_M}
\tilde{f}_{\textsf{dc}}^*(\lambda_i^*) = \tilde{f}_{\textsf{dc}}^*(\lambda_{i-1}^*) + \textsf{CR}_f^*(\rho, k)\cdot \left( \frac{1}{k} \lambda_{i-1}^* - c_i \right), \quad i\in \{\tau+2,\tau+3,\cdots, \bar{k}\},
\end{align}
%where $ \tilde{p}_m^* = \frac{1}{M}p_m^* $.
%Here, we can view $ p_{m-1}^* $ as the price per unit-weight the item has to pay, as argued by Remark \ref{remark_geometric_economic}. After normalization, the weight of each item is $ 1/M $, and thus $ M p_{m-1}^* $ denotes the ``money" item $ m-1 $ pays to the knapsack.

Since the threshold $ \lambda_i^*  $ depends on the total number of units that have been produced/sold, we can write $ \lambda_i^* $ as a function of the ratio $ i/k $, namely,
\begin{align*}
\lambda_i^* = \phi_{*}\left(\frac{i}{k}\right), \quad i\in \{\tau+2,\tau+3,\cdots, \bar{k}\},
\end{align*}
where $ \phi_* $ is referred to as the optimal threshold function and $\frac{i}{k} $ denotes the utilization ratio after accepting the $ i $-th buyer.  We emphasize that $ \phi_* $ can be uniquely determined as long as the optimal admission threshold $ \boldsymbol{\lambda}^* $ is given, and vice versa.  Substituting $ \lambda_i^* = \phi_{*}\left(\frac{i}{k}\right) $ into Eq. \eqref{eq_1_M}, we have
\begin{align*}
%\left( \phi_{*}\Big(\frac{m}{M}\Big) - \phi_{*}\Big(\frac{m-1}{M} \Big) \right) \cdot \frac{\tilde{f}_{\textsf{dc}}^*\big(\phi_{*}\left(\frac{m}{M}\right)\big) - \tilde{f}_{\textsf{dc}}^*\big(\phi_{*}\left(\frac{m-1}{M}\right)\big) }{ \phi_{*}\left(\frac{m}{M}\right) - \phi_{*}\left(\frac{m-1}{M}\right) }
\tilde{f}_{\textsf{dc}}^*\Big(\phi_{*}\big(\frac{i}{k}\big)\Big) - \tilde{f}_{\textsf{dc}}^*\Big(\phi_{*}\big(\frac{i-1}{k}\big)\Big)  = \textsf{CR}_f^*(\rho,k) \cdot \left( \frac{1}{k} \phi_{*}\big(\frac{i-1}{k}\big) - c_i \right).
\end{align*}
After some simple manipulations, we have the following equivalent form:
\begin{align*}
\frac{ \phi_{*}\left(\frac{i}{k}\right) - \phi_{*}\left(\frac{i-1}{k} \right) }{ \frac{i}{k} - \frac{i-1}{k} }\cdot \frac{\tilde{f}_{\textsf{dc}}^*\big(\phi_{*}\left(\frac{i}{k}\right)\big) - \tilde{f}_{\textsf{dc}}^*\big(\phi_{*}\left(\frac{i-1}{k}\right)\big) }{ \phi_{*}\left(\frac{i}{k}\right) - \phi_{*}\left(\frac{i-1}{k}\right) } 
= \textsf{CR}_f^*(\rho,k) \cdot \bigg( \phi_{*}\big(\frac{i-1}{k} \big) - \frac{ \tilde{f}\left(\frac{i}{k}\right) - \tilde{f}\left(\frac{i-1}{k} \right) }{ \frac{i}{k} - \frac{i-1}{k} } \bigg).
\end{align*}

A key observation that helps derive the asymptotic properties in Theorem \ref{theorem_asymptotic} is that, when $ k \rightarrow 0 $, we asymptotically have the following differential equation:
\begin{align*}
\lim\limits_{k\rightarrow +\infty} \left\{\phi_{*}'\big(\frac{i-1}{k}\big) \cdot \tilde{f}_{\textsf{dc}}^{*'}\Big( \phi_{*}\big(\frac{i-1}{k}\big) \Big) \right\} = \lim\limits_{k\rightarrow +\infty}  \left\{ \textsf{CR}_f^*(\rho, k) \cdot \Big( \phi_{*}\big(\frac{i-1}{k}\big) - \tilde{f}'\big(\frac{i-1}{k}\big)\Big) \right\}.
%\quad y\in \Big[ g^{\textsf{inv}}\Big(\frac{f^{*}(p_{\min})}{\normalfont \textsf{CR}_f^*(\rho, M)}\Big), 1 \Big],
\end{align*}
Note that when $ k\rightarrow +\infty $, Lemma \ref{theorem_f_star_convergence} implies the following asymptotic property:
\begin{align*}
\lim_{k\rightarrow +\infty} \tilde{f}_{\textsf{dc}}^{*'}\Big( \phi_{*}\big(\frac{i-1}{k}\big) \Big) = \tilde{f}_{\textsf{fc}}^{*'}\Big( \phi_{*}\big(\frac{i-1}{k}\big) \Big).
\end{align*}
Thus, we have the following ODE:
\begin{align*}
\lim\limits_{k\rightarrow +\infty} \phi_{*}'\Big(\frac{i-1}{k}\Big) =  \lim\limits_{k\rightarrow +\infty} \left\{ \textsf{CR}_f^*(\rho, k) \cdot \frac{ \phi_{*}\left(\frac{i-1}{k}\right) - \tilde{f}'\left(\frac{i-1}{k}\right)}{ \tilde{f}_{\textsf{fc}}^{*'}\left( \phi_{*}\left(\frac{i-1}{k}\right) \right)} \right\}, \ i\in \{\tau+2,\tau+3, \cdots, \bar{k}+1\}. 
\end{align*}

Based on the design of the optimal turning point $ \tau $ in Theorem \ref{theorem_optimality},
%\begin{align*}
%  k = g^{\textsf{inv}}\left(\frac{1}{\normalfont \textsf{CR}_f^*(\rho, M)}\cdot  f^{*}(p_{\min})\right)  - 1,
%\end{align*}
%which indicates that 
the value of $ \tau $ in our normalized unit-demand setting satisfies
\begin{align*}
\lim_{k\rightarrow +\infty} \tilde{g}\Big( \frac{\tau+1}{k}  \Big) = \lim_{k\rightarrow +\infty} \Big( p_{\min} \cdot \frac{\tau+1}{k} - \tilde{f}\big( \frac{\tau+1}{k}  \big) \Big) = \lim_{k\rightarrow +\infty} \frac{\tilde{f}_{\textsf{dc}}^{*}(p_{\min})}{\textsf{CR}_f^*(\rho,k)} = \lim_{k\rightarrow +\infty} \frac{\tilde{f}_{\textsf{fc}}^{*}(p_{\min})}{\textsf{CR}_f^*(\rho,k)}.
\end{align*}
Thus, the optimal turning point $ \tau $ in the normalized unit-demand setting satisfies
\begin{align*}
\lim_{k\rightarrow +\infty} \frac{\tau+1}{k} = \lim_{k\rightarrow +\infty} \tilde{g}^{-1}\Big( \frac{\tilde{f}_{\textsf{fc}}^{*}(p_{\min})}{ \textsf{CR}_f^*(\rho,k)} \Big).
\end{align*}

Based on the above analysis, in the normalized unit-demand setting with $ k \rightarrow +\infty $, the optimal threshold function $ \phi_* $ satisfies the following ODE with two boundary conditions
\begin{align}\label{eq_bvp_phi_star}
%\textsf{BVP}_2(\phi_{*})
\begin{cases}
	\phi_{*}'(y) = \Big( \lim\limits_{k\rightarrow +\infty} \textsf{CR}_f^*(\rho,k) \Big) \cdot \frac{ \phi_{*}(y) - \tilde{f}'(y)}{\tilde{f}_{\textsf{fc}}^{*'}( \phi_{*}(y) ) }, \quad y\in \Big( \lim\limits_{k\rightarrow +\infty}\tilde{g}^{-1}\big( \frac{\tilde{f}_{\textsf{fc}}^{*}(p_{\min})}{ \textsf{CR}_f^*(\rho, k)} \big), \vartheta \Big),\bigskip \\
	\phi_{*}\Big( \lim\limits_{k\rightarrow +\infty} \tilde{g}^{-1}\big(\frac{\tilde{f}_{\textsf{fc}}^{*}(p_{\min})}{\textsf{CR}_f^*(\rho,k)}\big) \Big) = p_{\min}, \quad \phi_{*}(\vartheta) = p_{\max}.
\end{cases}
\end{align}

Notice that, Eq. \eqref{eq_bvp_phi_star} is the same as Eq. \eqref{eq_bvp_phi} except the limit of $ \textsf{CR}_f^*(\rho, k) $, which indicates that the following asymptotic equivalence holds:
\begin{align*}
\lim_{k \rightarrow +\infty} \textsf{CR}_f^*(\rho,k) = \underline{\textsf{CR}}_f(\rho).
\end{align*}
Thus, we complete the proof of $ \normalfont \textsf{CR}_f^*(\rho,k) \rightarrow \underline{\textsf{CR}}_f(\rho) $ when $ k \rightarrow +\infty $.

\subsection{Convergence of $ \normalfont \textsf{CR}_f^{\textsf{lb}}(\rho,k) \rightarrow \protect\underline{\textsf{CR}}_f(\rho) $}
In this subsection we prove the convergence of $ \normalfont \textsf{CR}_f^{\textsf{lb}}(\rho,k) \rightarrow \underline{\textsf{CR}}_f(\rho) $ when $ k\rightarrow +\infty $. Similar to the proof of  $ \normalfont \textsf{CR}_f^*(\rho,k)$ in the previous subsection, the key is to show that Eq. \eqref{eq_bvp_Phi} asymptotically converges to Eq. \eqref{eq_bvp_phi} when $ k\rightarrow +\infty $.

By Eq. \eqref{eq_bvp_Phi}, the lower bound $ \textsf{CR}_f^{\textsf{lb}}(\rho,k) $ is such that the following ODE has a unique solution:
\begin{equation}\label{eq_bvp_Phi_convergence}
\begin{aligned}
	\normalfont
	\begin{cases}
		\bar{\Phi}'(y) = \normalfont  \frac{ \textsf{CR}_f^{\textsf{lb}}(\rho,k)\cdot( \bar{\Phi}(y) - f'(y)) }{ \sum_{i = \ubar{k}}^{\bar{k}} i \cdot \mathds{1}_{\left\{\bar{\Phi}(y)\in [c_i, c_{i+1})\right\}} },  
		\quad  y\in \Big( \hat{g}^{-1}\big(\frac{f^*(p_{\min})}{\textsf{CR}_f^{\textsf{lb}}(\rho, k)}\big), \bar{k} \Big),\bigskip \\
		\bar{\Phi}\Big(\hat{g}^{-1}\big(\frac{f^*(p_{\min})}{\textsf{CR}_f^{\textsf{lb}}(\rho, k)}\big)\Big) = p_{\min}, \quad \bar{\Phi}(\bar{k}) =  p_{\max}.
	\end{cases}
\end{aligned}
\end{equation}
where $ \hat{g}^{-1} $ is the inverse of $ \hat{g}(y) =  p_{\min}y - f(y) $ for $ y\in [0,\bar{k}] $. Recall that the derivative of the discrete conjugate function $ f^* $ is piecewise constant, and thus we have
\begin{equation*}
f^{*'}(\bar{\Phi}(y)) = \sum_{i = \ubar{k}}^{\bar{k}} i\cdot \mathds{1}_{\left\{ \bar{\Phi}(y)\in [c_i, c_{i+1}) \right\}}.
\end{equation*}
Therefore, the right-hand-side of the ODE in Eq. \eqref{eq_bvp_Phi_convergence} can be written as a function of $ f^{*'}(\bar{\Phi}(y)) $. Let us assume that $ \phi_{\textsf{lb}} $ is a scaled version of $ \bar{\Phi} $ to keep consistent with our normalized unit-demand setting, i.e.,
\begin{equation*}
\phi_{\textsf{lb}}\left(\frac{y}{k}\right) = \bar{\Phi}(y), \quad y\in [0,\bar{k}],
\end{equation*}
then, leveraging the fact that $ \tilde{f}_{\textsf{dc}}^{*} $ is asymptotically equivalent to $\normalfont \tilde{f}_{\textsf{fc}}^{*} $ when $ k \rightarrow +\infty $,  and further using the relationship between $ f $ and $ \tilde{f} $, 
we can transform Eq. \eqref{eq_bvp_Phi_convergence} to the problem as follows:
\begin{align}\label{eq_bvp_phi_lb}
\begin{cases}
	\phi_{\textsf{lb}}'(y) = \left( \lim\limits_{k\rightarrow +\infty} \textsf{CR}_f^\textsf{lb}(\rho, k) \right) \cdot \frac{ \phi_{\textsf{lb}}(y) - \tilde{f}'(y)}{ \tilde{f}_{\textsf{fc}}^{*}( \phi_{\textsf{lb}}(y) )}, \quad y\in \Big( \lim\limits_{k\rightarrow +\infty} \tilde{g}^{-1}\Big(\frac{ \tilde{f}_{\textsf{fc}}^{*}(p_{\min})}{\textsf{CR}_f^\textsf{lb}(\rho, k)}\Big), \vartheta \Big),\bigskip \\
	\phi_{\textsf{lb}}\bigg( \lim\limits_{k\rightarrow +\infty} \tilde{g}^{-1}\Big(\frac{ \tilde{f}_{\textsf{fc}}^{*}(p_{\min})}{\textsf{CR}_f^\textsf{lb}(\rho, k)}\Big)\bigg) = p_{\min}, \quad  \phi_{\textsf{lb}}(\vartheta) = p_{\max}.
\end{cases}
\end{align}

Notice that, Eq. \eqref{eq_bvp_phi_lb} is the same as Eq. \eqref{eq_bvp_phi} except the limit of $ \textsf{CR}_f^{\textsf{lb}}(\rho,k) $, which indicates that the following asymptotic equivalence holds
\begin{align*}
\lim_{k \rightarrow +\infty} \textsf{CR}_f^{\textsf{lb}}(\rho,k) = \underline{\textsf{CR}}_f(\rho).
\end{align*}
Thus, we complete the proof of $ \normalfont \textsf{CR}_f^{\textsf{lb}}(\rho,k) \rightarrow \underline{\textsf{CR}}_f(\rho) $ when $ k\rightarrow +\infty $.

\section{Proof of Proposition \ref{theorem_convex_upper_bound}}
\label{proof_of_theorem_convex_upper_bound}
When $ p_{\min} > c_k $, we have $ \ubar{k} = \bar{k} = k $ and $ f^*(p) = pk -f(k) $. Based on the \textbf{SoSE} in Eq. \eqref{eq_system_of_equations},  we have the following recursive equation:
\begin{align*}
f^{*}(\lambda_i^*) = f^{*}(\lambda_{i-1}^*) + \textsf{CR}_f^*(\rho, k) \cdot (\lambda_{i-1}^* - c_i), \quad i\in \{\tau+2,\cdots,k\}.
\end{align*}
Substituting $ f^*(p) = kp -f(k) $ into the above recursive equation leads to
\begin{align} 
\lambda_i^* = \Big(1 + \frac{\textsf{CR}_f^*(\rho, k)}{k} \Big)\lambda_{i-1}^* -  \frac{\textsf{CR}_f^*(\rho, k)}{k}c_i,\quad  i\in \{\tau+2,\cdots,k\}.
\end{align}
For $ i\in \{\tau+2, \cdots, k\} $, the solution to the above recursive equation is given by
\begin{align*}
\lambda_i^* = \Big(1 + \frac{\textsf{CR}_f^*(\rho,k)}{k} \Big)^{i-\tau-1}\lambda_{\tau+1}^* - \frac{\textsf{CR}_f^*(\rho,k)}{k}\sum_{j=\tau+2}^{i} c_{j}\cdot \Big(1 + \frac{\textsf{CR}_f^*(\rho,k)}{k} \Big)^{j - \tau - 2}.
\end{align*}
Substituting $ i = k $ into the above general solution leads to the value of $ \lambda_k^* $:
\begin{align*}
\lambda_k^* =\ & \Big(1 + \frac{\textsf{CR}_f^*(\rho,k)}{k} \Big)^{k-\tau-1} \cdot \lambda_{\tau+1}^* - \frac{\textsf{CR}_f^*(\rho,k)}{k}\sum_{j=\tau+2}^{k} c_{j}\cdot \Big(1 + \frac{\textsf{CR}_f^*(\rho,k)}{k} \Big)^{j - \tau - 2}\\
\geq\ & \Big(1 + \frac{\textsf{CR}_f^*(\rho,k)}{k} \Big)^{k-\tau-1}\cdot p_{\min}  - \frac{\textsf{CR}_f^*(\rho,k)}{k}\cdot c_k \cdot \frac{\Big(1 + \frac{\textsf{CR}_f^*(\rho,k)}{k} \Big)^{k-\tau-1}- 1}{\frac{\textsf{CR}_f^*(\rho,k)}{k}} \\
=\ & \Big(1 + \frac{\textsf{CR}_f^*(\rho,k)}{k} \Big)^{k-\tau-1}\cdot p_{\min}   - c_k\cdot\bigg(\Big(1 + \frac{\textsf{CR}_f^*(\rho,k)}{k} \Big)^{k-\tau-1}- 1\bigg)\\
=\ & \Big(1 + \frac{\textsf{CR}_f^*(\rho,k)}{k} \Big)^{k-\tau-1}\cdot\left(p_{\min} - c_k\right) + c_k,
\end{align*}
where we use the fact that $ \lambda_{\tau+1}^* \geq p_{\min} = \lambda_\tau^* $ and $ c_i \leq c_k $ for all $ i\in [k] $. Since $ \lambda_k^* = p_{\max} $, we have
\begin{align*}
\Big(1 + \frac{\textsf{CR}_f^*(\rho,k)}{k} \Big)^{k-\tau-1} \leq  \frac{p_{\max} - c_k}{p_{\min} - c_k} = \rho(c_k).
\end{align*}
Recall that $
k = g^{\textsf{inv}}(\frac{f^{*}(p_{\min})}{\textsf{CR}_f^*(\rho, M)}) - 1 $, and thus $ \textsf{CR}_f^*(\rho, k) $ satisfies the following inequality
\begin{align}\label{eq_CR_f_upper_bound_proof}
\Big(1 + \frac{\textsf{CR}_f^*(\rho,k)}{k} \Big)^{k-g^{\textsf{inv}} \big(\frac{f^{*}(p_{\min})}{\textsf{CR}_f^*(\rho,k)} \big) } \leq  \rho(c_k).
\end{align}
%The above equality is reached if $ f(m) $ is linear in $ m\in [M] $ and $ M $ approaches infinity.

For any $ \alpha \in [1,+\infty) $, the monotonicity of $ g $ and the convexity of $ f $ imply
\begin{align*}
g\Big(\big\lceil \frac{k}{\alpha} \big\rceil \Big) = p_{\min}  \big\lceil \frac{k}{\alpha} \big\rceil - f\Big(\big\lceil \frac{k}{\alpha}\big\rceil\Big)
\geq  \frac{1}{\alpha} \left(k p_{\min} - \alpha f\Big(\frac{k}{\alpha}\Big)\right)
\geq \frac{1}{\alpha} \Big(k p_{\min} - f(k)\Big) = \frac{1}{\alpha} f^{*}(p_{\min}),
\end{align*}
where the second inequality  is because  $  \alpha f(\frac{k}{\alpha}) \leq  f(k)$. Thus, by the definition of $ g^{\textsf{inv}} $, we have
\begin{align*}
g^{\textsf{inv}}\Big(\frac{f^{*}(p_{\min})}{\textsf{CR}_f^*(\rho,k)}\Big) \leq \big\lceil \frac{k}{\textsf{CR}_f^*(\rho,k)} \big\rceil.
\end{align*}
Therefore, the optimal competitive ratio $ \textsf{CR}_f^*(\rho,k) $ satisfies
\begin{align*}
\Big(1 + \frac{\textsf{CR}_f^*(\rho,k)}{k} \Big)^{k- \big\lceil \frac{k}{\textsf{CR}_f^*(\rho,k)} \big\rceil} \leq \Big(1 + \frac{\textsf{CR}_f^*(\rho,k)}{k} \Big)^{k-g^{\textsf{inv}}\big(\frac{f^{*}(p_{\min})}{\textsf{CR}_f^*(\rho,k)}\big) } \leq   \rho(c_k).
\end{align*}
We thus complete the proof of the upper bound of $ \textsf{CR}_f^*(\rho,k)$, that is, the first term in Eq. \eqref{eq_upper_bound_convex}. 

As to the asymptotic lower bound, when $ k \rightarrow +\infty $, we have
\begin{align*}
\lim_{k \rightarrow +\infty} \Big(1 + \frac{\textsf{CR}_f^*(\rho,k)}{k} \Big)^{k- \lceil \frac{k}{\textsf{CR}_f^*(\rho,k)} \rceil} = \exp\Big( \underline{\textsf{CR}}_f(\rho)-1 \Big) \leq \lim_{k\rightarrow +\infty}\rho(c_k) \leq \rho(c_{\max}),
\end{align*}
which thus leads to the following inequality
\begin{align*}
\underline{\textsf{CR}}_f(\rho) \leq 1 + \ln \left(\frac{p_{\max} - c_{\max}}{p_{\min} - c_{\max}}\right) = 1 +\ln\big(\rho(c_{\max})\big),
\end{align*}
where the equality holds when the production cost function $ f(y) = a y  $ for $ a \geq 0 $. We thus complete the proof of Proposition \ref{theorem_convex_upper_bound}.

\section{Proof of Theorem \ref{theorem_strongly_convex}}
\label{proof_of_strongly_convex}

Before proving the upper bounds in Theorem \ref{theorem_strongly_convex}, we first give the following two lemmas.

\begin{lemma}
If $ p_{\max}\geq p_{\min} > c_k $ and $ f $ is $ \mu $-strongly convex, then $ \xi $ and $ \zeta $ satisfy
\begin{align*}
	\xi = \frac{p_{\min} - f'(0)}{\mu k} \in \Big [1 - \frac{1}{2k}, +\infty \Big), \quad \zeta =  \lim_{k\rightarrow +\infty} \frac{p_{\min} - f'(0)}{\mu k} \in \big [1,+\infty\big),
\end{align*}
\end{lemma}
\begin{proof}
The proof is trivial based on the strong convexity of $ f $. Specifically, we have
\begin{align*}
	\xi =  \frac{p_{\min} - f'(0)}{\mu k} \geq \frac{c_{k} - f'(0)}{\mu k} \geq \frac{f'(k-1) + \frac{\mu}{2} - f'(0)}{\mu k} \geq \frac{\mu (k-1) + \frac{\mu}{2}}{\mu k} = 1 - \frac{1}{2k}.
\end{align*}
Since $ \zeta $ is the limit of $ \xi $ when $ k \rightarrow +\infty $, $ \zeta \in [1,+\infty) $ follows.
\end{proof}

\begin{lemma} 
\label{theorem_g_inverse_upper_bound}
Given a setup $ \mathcal{S} $, if $ f(y) $ is $ \mu $-strongly convex over $ y\in [0,k] $, then for any $ \alpha \geq 1 $, the following inequality holds:
\begin{align}\label{eq_g_inv_upper_bound}
	g^{\textsf{inv}}\left(\frac{1}{\alpha} f^{*}(p_{\min})\right) \leq \Big\lceil \frac{\ubar{k}}{\alpha} \bigg(\xi \alpha - \sqrt{ \left(\xi \alpha - 1 \right)^2 + \alpha-1}\bigg) \Big\rceil.
\end{align}
\end{lemma}
\begin{proof}
The strong convexity of $ f $ with constant $ \mu $ implies that
\begin{align*}
	\alpha f\Big(\frac{ \ubar{k} }{\alpha}\Big) \leq  f( \ubar{k} ) - \frac{(\alpha-1)\mu}{2\alpha} \ubar{k}^2.
\end{align*}
%	which thus indicates that
%	\begin{align*}
	%	g\Big(\frac{M}{\alpha}\Big) = 	\frac{M}{\alpha}  p_{\min} - f\Big(\frac{M}{\alpha}\Big)
	%	= \frac{1}{\alpha} \left(Mp_{\min} - \alpha f\Big(\frac{M}{\alpha}\Big)\right)
	%	\geq \frac{1}{\alpha} \left(Mp_{\min} - f(M) + \frac{(\alpha-1)\mu}{2\alpha}M^2\right) \geq \frac{1}{\alpha} f^{*}(p_{\min}),
	%	\end{align*}

To simplify the notations in the remaining analysis, let us define $ \varphi_0 $ by
\begin{align*}
	\varphi_0 \triangleq \frac{ \ubar{k} }{\alpha} \bigg(\xi \alpha - \sqrt{ \left(\xi \alpha - 1 \right)^2 + \alpha-1}\bigg).
\end{align*}
By the monotonicity of the min-profit function $ g $, we have $ g\big(\lceil \varphi_0  \rceil\big) = p_{\min} \lceil \varphi_0 \rceil  - f( \lceil \varphi_0 \rceil)  \geq p_{\min} \varphi_0  - f(\varphi_0) $.
%\footnote{Here, we temporarily extend the domain of $ g $ from a discrete set $ \{0,1,\cdots, M\} $ to a continuous interval $ [0,M] $, i.e., $ g(y) = p_{\min}y - f(y) $ for all $ y\in [0,M] $. Since $ g'(y) = p_{\min} - f'(y) \geq 0 $ always holds when $ p_{\min}> c_{\max} \geq  f'(M) $, $ g(\lceil\varphi_0\rceil) \geq g(\varphi_0) $ follows}
Therefore, if we can prove that $ p_{\min} \varphi_0 - f\left(\varphi_0\right) \geq  \frac{1}{\alpha} f^{*}(p_{\min}) $ holds for all $ \mu $-strongly convex $ f $, then based on the definition of $ g^{\textsf{inv}} $ in Eq. \eqref{eq_g}, Eq. \eqref{eq_g_inv_upper_bound}  follows the following inequalities
\begin{align*}
	g\big(\lceil \varphi_0  \rceil\big) = p_{\min} \lceil \varphi_0 \rceil  - f( \lceil \varphi_0 \rceil)  \geq p_{\min}\varphi_0 - f\left(\varphi_0\right) \geq  \frac{1}{\alpha} f^{*}(p_{\min}).
\end{align*}

Our proof is performed in the reversed order. Specifically, we first solve the following inequalty
\begin{align}\label{eq_varphi}
	p_{\min} \varphi - f(\varphi) \geq p_{\min}\frac{ \ubar{k} }{\alpha} - f\big(\frac{ \ubar{k} }{\alpha}\big) -  \frac{(\alpha-1)\mu}{2\alpha^2} \ubar{k}^2,
\end{align}
and then show that $ \varphi = \varphi_0 $ is indeed included in the solution. Note that if Eq. \eqref{eq_varphi} holds, then 
\begin{align*}
	p_{\min} \varphi - f(\varphi) \geq p_{\min}\frac{ \ubar{k} }{\alpha} - f\big(\frac{ \ubar{k} }{\alpha}\big) -  \frac{(\alpha-1)\mu}{2\alpha^2} \ubar{k}^2 \geq  \frac{1}{\alpha} \Big(\ubar{k} p_{\min} - f( \ubar{k} )\Big) = \frac{1}{\alpha} f^{*}(p_{\min}),
\end{align*}
where the second inequality follows the strong convexity of $ f $, and the last equality is based on the definition of $ f^* $. Thus, if we show that $ \varphi = \varphi_0 $ satisfies Eq. \eqref{eq_varphi}, then we finish the proof of  Eq. \eqref{eq_g_inv_upper_bound} and Lemma \ref{theorem_g_inverse_upper_bound} follows. 

We now show that   $ \varphi = \varphi_0 $ indeed satisfies Eq. \eqref{eq_varphi}. Note that Eq. \eqref{eq_varphi} is equivalent to the following inequality:
\begin{align*}
	f\big(\frac{ \ubar{k} }{\alpha}\big) - f(\varphi) \geq p_{\min}\big( \frac{ \ubar{k} }{\alpha} - \varphi \big) - \frac{(\alpha-1)\mu}{2\alpha^2} \ubar{k}^2.
\end{align*}
The strong-convexity of $ f $ with constant $ \mu $ indicates that
\begin{align*}
	f\big( \frac{ \ubar{k} }{\alpha} \big) - f(\varphi) \geq f'(\varphi)\Big(\frac{ \ubar{k} }{\alpha} - \varphi\Big) + \frac{\mu}{2}\Big(\frac{ \ubar{k} }{\alpha} - \varphi\Big)^2 \geq \big(\mu \varphi + f'(0)\big) \Big(\frac{ \ubar{k} }{\alpha} - \varphi\Big) + \frac{\mu}{2}\Big(\frac{ \ubar{k} }{\alpha} - \varphi\Big)^2.
\end{align*}
Therefore, to show whether $ \varphi = \varphi_0 $ satisfies Eq. \eqref{eq_varphi}, it suffices to show whether $ \varphi = \varphi_0 $  satisfies
\begin{align*}
	\Big(\mu \varphi + f'(0)\Big)\Big(\frac{ \ubar{k} }{\alpha} - \varphi\Big) + \frac{\mu}{2}\Big(\frac{ \ubar{k} }{\alpha} - \varphi\Big)^2 \geq p_{\min}\Big(\frac{ \ubar{k} }{\alpha} - \varphi\Big) - \mu\cdot\frac{\alpha-1}{2}\cdot\Big(\frac{ \ubar{k} }{\alpha}\Big)^2.
\end{align*}
In fact, solving the above quadratic inequality of $ \varphi $ leads to the following solution:
\begin{align*}
	\varphi \geq 
	%\frac{\tilde{p}_{\min}}{\mu} - \sqrt{\frac{\tilde{p}_{\min}^2}{\mu^2} - 2\cdot \frac{\tilde{p}_{\min}}{\mu} \cdot \frac{M^2}{\alpha^2} + \alpha\cdot\frac{M^2}{\alpha^2} } =
	%\frac{\tilde{p}_{\min}}{\mu} - \sqrt{ \left(\frac{\tilde{p}_{\min}}{\mu} - \frac{M}{\alpha} \right)^2 + (\alpha-1)\cdot\frac{M^2}{\alpha^2} }  = 
	\frac{ \ubar{k} }{\alpha} \bigg(\xi \alpha - \sqrt{ \left(\xi \alpha - 1 \right)^2 + \alpha-1}\bigg) = \varphi_0.
\end{align*}	
%	\begin{align*}
	%  	\frac{M}{\alpha} - \frac{p_{\min} - c_{\min} +\sqrt{(p_{\min} - c_{\min})^2 - \frac{(\alpha-1)\mu^2}{2\alpha^2}M^2}}{\mu}\leq y_0 \leq \frac{M}{\alpha} - \frac{p_{\min} - c_{\min} - \sqrt{(p_{\min} - c_{\min})^2 - \frac{(\alpha-1)\mu^2}{2\alpha^2}M^2}}{\mu}
	%    \end{align*}
%This implies that $ p_{\min} \varphi_0 - f\left(\varphi_0\right) \geq  \frac{1}{\alpha} f^{*}(p_{\min}) $ holds for all $ \mu $-strongly convex $ f $.  
We thus complete the proof of Lemma \ref{theorem_g_inverse_upper_bound}.
\end{proof}

\textbf{Proof of Theorem \ref{theorem_strongly_convex}}. Now we prove Eq. \eqref{eq_upper_bound_strong_convex} in Theorem \ref{theorem_strongly_convex}. Lemma \ref{theorem_g_inverse_upper_bound} is our general result regarding the upper bound of the turning point $ k $. When $ p_{\min} \geq c_{\max}\geq c_k $, we have $ \ubar{k} = k $. Thus, based on Lemma \ref{theorem_g_inverse_upper_bound} and Eq. \eqref{eq_CR_f_upper_bound_proof}, we have
\begin{align}\label{eq_proof_of_strongly_convex}
\left(1 + \frac{ \textsf{CR}_f^*(\rho,k) }{k} \right)^{k - \left\lceil \frac{k}{ \textsf{CR}_f^*(\rho,k) }\big(\xi \textsf{CR}_f^*(\rho,k) - \sqrt{ \left(\xi \textsf{CR}_f^*(\rho,k) - 1 \right)^2 + \textsf{CR}_f^*(\rho,k) -1}\big)\right\rceil } \leq \rho\big(c_k\big).
\end{align}
Since Eq. \eqref{eq_proof_of_strongly_convex} holds for all $ k $, it holds when $ k \rightarrow +\infty $ as well, leading to the following inequality
\begin{align*}
%\lim_{M\rightarrow \infty}\left(1 + \frac{\alpha}{M} \right)^{M- \frac{M}{\alpha}\big(\xi \alpha - \sqrt{ \left(\xi \alpha - 1 \right)^2 + \alpha-1} \big)} =
\frac{\exp\left(\underline{\textsf{CR}}_f(\rho)\right)}{\exp\left(\zeta\cdot \underline{\textsf{CR}}_f(\rho) - \sqrt{ \left(\zeta\cdot \underline{\textsf{CR}}_f(\rho) - 1 \right)^2 + \underline{\textsf{CR}}_f(\rho) - 1}\right)} \leq \lim_{k\rightarrow +\infty}\rho(c_k) \leq \rho\big(c_{\max}\big).
\end{align*}
Thus, we have
\begin{align*}
\underline{\textsf{CR}}_f(\rho) - \left( \zeta\cdot  \underline{\textsf{CR}}_f(\rho) - \sqrt{ \left(\zeta\cdot \underline{\textsf{CR}}_f(\rho) - 1 \right)^2 + \underline{\textsf{CR}}_f(\rho) -1} \right) \leq \ln \Big(\rho\big(c_{\max}\big)\Big).
\end{align*}
Solving the above inequality leads to the upper bound of $ \underline{\textsf{CR}}_f(\rho) $ in Eq. \eqref{eq_upper_bound_strong_convex_2}. We thus complete the proof of Theorem \ref{theorem_strongly_convex}.

\section{Proof of Claim \ref{claim_k}}
\label{proof_of_claim_k}

%For any $ \alpha $-competitive deterministic online algorithm, Proposition \ref{theorem_k_alpha} argues that there exists a selection function $ \psi $ with an integer $ k \in \mathcal{K}^{(\alpha)} $ so that $  \psi(p_{\min}) = k +1 \geq \tau^{(\alpha)} + 1 $.  
In this section, we prove Claim \ref{claim_k}, that is, it is possible to construct another $ \alpha $-feasible selection function $ \tilde{\psi} $ with  $ \tilde{\psi}(p_{\min}) = \tau^{(\alpha)} + 1 $  whenever there exists an $ \alpha $-feasible selection function $  \psi $ with $  \psi(p_{\min}) > \tau^{(\alpha)}  + 1 $. 
%We remark that the following proof of Claim \ref{claim_k} is the foundation of our proofs of Claim \ref{claim_two_one} and Claim \ref{claim_equality} in the next two sections.

For any $ \alpha $-feasible selection function $  \psi\in \mathcal{P}^{(\alpha)} $ with $  \psi(p_{\min}) > \tau^{(\alpha)} + 1$, let us assume without loss of generality\footnote{For any case with $  \psi(p_{\min}) > \tau^{(\alpha)} + 2  $, the proof follows the same way as $  \psi(p_{\min}) = \tau^{(\alpha)} + 2  $.} that $  \psi(p_{\min}) = \tau^{(\alpha)} + 2 $.  Recall 
that the non-zero points of $ \psi $ are denoted by
$ \Omega = \{\omega_1,\omega_2,\cdots, \omega_L\} $, where the first non-zero point $ \omega_1 = p_{\min} $ (by Corollary \ref{theorem_number_of_omega}). The definitions of $\tau^{(\alpha)}$ and $ \mathcal{K}^{(\alpha)}$ imply that $ g(\tau^{(\alpha)} + 1) = p_{\min} (\tau^{(\alpha)}+1)   - f\left(\tau^{(\alpha)} + 1\right) \geq \frac{1}{\alpha}f^*(p_{\min})$. Thus, based on the strict monotonicity of $ f^* $, there must exist an $ \tilde{\omega}_{2}\geq \omega_1 = p_{\min} $ so that:
\vspace{-0.2cm}
\begin{equation}\label{eq_existence_of_omega_2}\vspace{-0.2cm}
g(\tau^{(\alpha)} + 1) = p_{\min} (\tau^{(\alpha)} + 1)   - f\left(\tau^{(\alpha)} + 1\right) \geq  \frac{1}{\alpha}f^*(\tilde{\omega}_{2}) \geq \frac{1}{\alpha}f^*(p_{\min}).
\end{equation}
Now we show that a new $ \alpha $-feasible selection function $ \tilde{\psi} $ can be constructed based on $ \tilde{\omega}_{2} $. Our proof is organized into two cases depending on whether $ \tilde{\omega}_{2}\in [p_{\min}, \omega_2) $ or $ \tilde{\omega}_{2}\in [\omega_2, +\infty) $.

(\textbf{Case-1}: $ \tilde{\omega}_{2}\in [p_{\min}, \omega_2) $) We first discuss the case when $ \tilde{\omega}_{2}\in [p_{\min}, \omega_2) $. In this case, we prove that the following selection function $ \tilde{\psi} $ is $ \alpha $-feasible:
\begin{align*}
\tilde{\psi}(p) =
\begin{cases}
	\tau^{(\alpha)} + 1 & \text{if } p  = \omega_1 = p_{\min},\\
	1 & \text{if } p = \tilde{\omega}_{2},\\
	\psi(p) & \text{if } p\in  \{ \omega_2, \omega_3, \cdots, \omega_L\},  \\
	0 & \text{otherwise}.
\end{cases}
\end{align*}
To show $ \tilde{\psi} $ is $ \alpha $-feasible, we need to prove that $ \tilde{\psi}(p) $ satisfies the feasibility conditions in Eq. \eqref{eq_necessary_inequality} and Eq. \eqref{eq_necessary_boundary}.
Note that the boundary conditions in Eq. \eqref{eq_necessary_boundary} are always satisfied as $ \tilde{\psi}(p_{\min}) = \tau^{(\alpha)} + 1 $ and $ \int_{p_{\min}}^{p_{\max}} \tilde{\psi}(\eta)d\eta  = \int_{p_{\min}}^{p_{\max}} \psi(\eta)d\eta \leq \bar{k} $ (i.e., the feasibility of $ \psi $). Thus, we only need to prove that $ \tilde{\psi}(p) $ satisfies Eq. \eqref{eq_necessary_inequality} for all $ p\in (p_{\min},p_{\max}] $. We next prove this by evaluating the feasibility of $ \tilde{\psi} $ in three sub-cases as follows:
\begin{itemize}
\item When $ p\in (p_{\min}, \tilde{\omega}^{2}) $, it is easy to check that $ \tilde{\psi}(p) $ satisfies  Eq. \eqref{eq_necessary_inequality}:
\begin{align*}
	\int_{p_{\min}}^{p} \eta \tilde{\psi}(\eta)d\eta - f\left(\int_{p_{\min}}^{p}  \tilde{\psi}(\eta)d\eta\right) = p_{\min}(\tau^{(\alpha)}+1) - f\big(\tau^{(\alpha)} + 1 \big) \geq   \frac{1}{\alpha} f^*(\tilde{\omega}_{2}) \geq \frac{1}{\alpha} f^*(p). 
\end{align*}

\item When $ p\in [\tilde{\omega}_{2}, \omega_2 $),  $ \tilde{\psi} $ satisfies  Eq. \eqref{eq_necessary_inequality} because:
\begin{align*}
	\int_{p_{\min}}^{p} \eta \tilde{\psi}(\eta)d\eta - f\left(\int_{p_{\min}}^{p}  \tilde{\psi}(\eta)d\eta\right) 
	=\ & p_{\min}(\tau^{(\alpha)} + 1) + \tilde{\omega}_{2} - f\big( \tau^{(\alpha)} + 2 \big) \\
	>\ & p_{\min} \big( \tau^{(\alpha)} + 2 \big) - f\big( \tau^{(\alpha)} + 2 \big)\\
	=\ & \int_{p_{\min}}^{p} \eta \psi(\eta)d\eta - f\left(\int_{p_{\min}}^{p}  \psi(\eta)d\eta\right)\\
	\geq\ & \frac{1}{\alpha} f^{*}(p).
\end{align*}
where the last inequality is because $ \psi $ is $ \alpha $-feasible.

\item When $ p\in [\omega_2, p_{\max}] $, $ \tilde{\psi} $ satisfies  Eq. \eqref{eq_necessary_inequality} because:  \begin{align*}
	& \int_{p_{\min}}^{p} \eta \tilde{\psi}(\eta)d\eta - f\left(\int_{p_{\min}}^{p}  \tilde{\psi}(\eta)d\eta\right) \\
	=\ & p_{\min}(\tau^{(\alpha)} + 1) + \tilde{\omega}_{2} +  \int_{\tilde{\omega}_{2}}^{p}  \eta \tilde{\psi}(\eta)d\eta - f\left(\int_{p_{\min}}^{p}  \tilde{\psi}(\eta)d\eta\right) \\
	=\ & p_{\min}(\tau^{(\alpha)} + 1) + \tilde{\omega}_{2} + \int_{\tilde{\omega}_{2}}^{p} \eta \psi(\eta)d\eta  - f\left(\int_{p_{\min}}^{p}\psi(\eta)d\eta\right)\\
	>\ & p_{\min}(\tau^{(\alpha)} + 2) + \int_{\tilde{\omega}_{2}}^{p} \eta \psi(\eta)d\eta  - f\left(\int_{p_{\min}}^{p} \psi(\eta)d\eta\right)\\
	=\ & \int_{p_{\min}}^{p} \eta \psi(\eta)d\eta - f\left(\int_{p_{\min}}^{p}  \psi(\eta)d\eta\right)\\
	\geq\ & \frac{1}{\alpha} f^{*}(p).
\end{align*}
\end{itemize}

Combining the above discussions, we conclude that a new $ \alpha $-feasible selection function $ \tilde{\psi} $ can be constructed with $ \tilde{\psi}(p_{\min}) = \tau^{(\alpha)} + 1$ when $ \tilde{\omega}_{2}\in [p_{\min}, \omega_{2}) $.  We next proceed with \textbf{Case-2} when $ \tilde{\omega}_{2}\in  [\omega_2,+\infty) $. 

(\textbf{Case-2}: $ \tilde{\omega}_{2}\in  [\omega_2,+\infty) $) When $ \tilde{\omega}_{2}\in  [\omega_2,+\infty) $, we need to discuss the following three sub-cases. 
\begin{itemize}
\item First, when $ \tilde{\omega}_{2} = \omega_2 $, we can construct a new $ \alpha $-feasible selection function by slightly modifying the one from \textbf{Case-1}:
\begin{align*}
	\tilde{\psi}(p) =
	\begin{cases}
		\tau^{(\alpha)} + 1 & \text{if } p = \omega_1 = p_{\min},\\
		2 & \text{if } p = \tilde{\omega}_{2} = \omega_2,\\
		\psi(p) & \text{if } p\in  \{\omega_3,\omega_4, \cdots, \omega_L\},  \\
		0 & \text{otherwise}.
	\end{cases}
\end{align*}
%It is easy to prove that the above selection function is feasible. 

\item Second, when $ \tilde{\omega}_{2} \in (\omega_2, \omega_L] $, let us assume $ \tilde{\omega}\in (\omega_{\ell}, \omega_{\ell+1}) $, where $ \ell $ can be any integer within $ \{2,3,\cdots,L-1\} $. Note that here we do not include the case when $ \tilde{\omega}_{2} = \omega_{\ell+1} $ since whenever we can construct a new $ \alpha $-feasible selection function for $ \tilde{\omega}\in (\omega_{\ell}, \omega_{\ell+1}) $,  we can always construct an  $\alpha$-feasible selection function for $ \tilde{\omega}_{2} = \omega_{\ell+1} $, as demonstrated above regarding $ \tilde{\omega}_{2} = \omega_2 $. We argue that when  $ \tilde{\omega}\in (\omega_{\ell}, \omega_{\ell+1}) $ for some $ \ell \in  \{2,3,\cdots, L-1\} $, the following selection function is $\alpha$-feasible:
\begin{align*}
	\tilde{\psi}(p) =
	\begin{cases}
		\tau^{(\alpha)} + 1 & \text{if } p = \omega_1  = p_{\min},\\
		1+\sum_{p\in \{ \omega_2, \omega_3, \cdots, \omega_{\ell}\} } \psi(p) & \text{if } p = \tilde{\omega}_{2},\\
		\psi(p) &\text{if } p\in \{\omega_{\ell+1}, \omega_{\ell+2}, \cdots, \omega_L\},\\
		0 & \text{otherwise}.
	\end{cases}
\end{align*}

\item Third, when $  \tilde{\omega}_{2} \in (\omega_L, p_{\max}] $, we can construct a new $ \alpha $-feasible selection function as follows:
\begin{align*}
	\tilde{\psi}(p) =
	\begin{cases}
		\tau^{(\alpha)} + 1 & \text{if } p = p_{\min} = \omega_1,\\
		1+ \sum_{p\in \{ \omega_2, \omega_3, \cdots, \omega_L\} } \psi(p) & \text{if } p = \tilde{\omega}_{2},\\
		0 & \text{otherwise}.
	\end{cases}
\end{align*}

\item Finally, when $  \tilde{\omega}_2 \in (p_{\max},+\infty) $, we can construct a new $ \alpha $-feasible selection function as follows:
\begin{align*}
	\tilde{\psi}(p) =
	\begin{cases}
		\tau^{(\alpha)} + 1 & \text{if } p = \omega_1 =  p_{\min},\\
		%\sum_{v\in \{ \omega_2, \omega_3, \cdots, \omega_L\} } \psi(p) & \text{if } v = \tilde{\omega}_{2} = p_{\max},\\
		0 & \text{otherwise}.
	\end{cases}
\end{align*}
%It can be shown that the above selection function is a feasible one with $ \tilde{\psi}(p_{\min}) = \tau^{(\alpha)} + 1 $. We skip the details as the proof is similar to \textsf{Feasibility-Case-1}.
\end{itemize}

It can be shown that the new selection functions in the above four sub-cases are all $\alpha$-feasible with $ \tilde{\psi}(p_{\min}) = \tau^{(\alpha)} + 1 $. We omit the details as the proof is similar to \textbf{Case-1}. 

Summarizing the discussions of \textbf{Case-1} and \textbf{Case-2} above, we complete the proof of Claim \ref{claim_k}.

%%%%%%%%%%%%%%%%%%%%%%%%%%%%%%%%%%%%%%%%%%%%%%%%%%%%%%%%%
%\begin{figure}
%  \centering
%  %\includegraphics[height = 6.5cm]{packing_function_example}
%  \subfigure[Examples of $ \psi(p) $ and $ \textsf{P}(p) $.]{\includegraphics[height = 5.5cm]{packing_function_example}}
%  \qquad
%  \qquad  
%  \subfigure[Comparison of $ \textsf{P}_1(p) $ and $ \textsf{P}_2(p) $.]{\includegraphics[height = 5.5cm]{packing_function_compare}}
%  \caption{An example of $ \psi $ and $ \textsf{P} $ for $ p_{\min} = 1.5 $, $ p_{\max} = 4.5 $, and $ M=5 $. In the figure, 2 items with value 1.5, 1 item with value 2.5, and 2 items with value 4, are accepted. Thus, the non-zero points of $ \psi $ are given by $ \Omega = \{1.5,2.5, 4\}$, and we have $ \psi_1 = \psi(1.5) = 2 $, $ \psi_2 = \psi(2.5) = 1 $, and $ \psi_3 = \psi(4) = 2 $.}
%  \label{fig_packing_function}
%\end{figure}

%\begin{figure}
%  \centering
%  \subfigure[Feasibility-Case-1]{\includegraphics[height = 5.5cm]{packing_function_k_alpha_1}}
%  \qquad 
%  \qquad 
%  \subfigure[Feasibility-Case-2]{\includegraphics[height = 5.5cm]{packing_function_k_alpha_2}}
%  \caption{Illustration of the characteristic selection function $ \psi^{(\alpha)}(p) $, as well as its corresponding cumulative selection function $ \textsf{P}_{\alpha}(p) $.}
%  \label{fig_feasibility}
%\end{figure}
%%%%%%%%%%%%%%%%%%%%%%%%%%%%%%%%%%%%%%%%%%%%%%%%%%%%%%%%%

\section{Proof of Claim \ref{claim_two_one}}
\label{proof_of_claim_two_one}
To prove Claim \ref{claim_two_one}, we assume that $  \psi(\omega_{\ell}) = 2 $ for some $ \ell\in \{2,3,\cdots, L\} $, and then show that it is possible to construct another $ \alpha $-feasible selection function $ \tilde{\psi} $ with $ \tilde{\psi}(\omega_{\ell}) = 1 $. 
Note that similar to our proof of Claim \ref{claim_k}, for any case with $  \psi(\omega_{\ell}) >  2  $, the proof follows the same way as $  \psi(\omega_{\ell}) = 2  $. Thus, it is without loss of generality to consider $  \psi(\omega_{\ell}) = 2 $ only.

The $ \alpha $-feasibility of $ \psi $ indicates that before and after selecting two buyers with price $ \omega_{\ell} $, the feasibility condition in Eq. \eqref{eq_necessary_inequality} is satisfied, i.e., 
\begin{subequations}\label{eq_feasibility_before_after}
\begin{align}
	\sum_{j=1}^{\ell-1}\omega_{j}\psi(\omega_{j}) - f\left(\sum_{j=1}^{\ell-1}\psi(\omega_{j})\right) \geq \frac{1}{\alpha}f^*\left( \omega_{\ell}\right),\\
	\sum_{j=1}^{\ell-1}\omega_{j}\psi(\omega_{j}) + 2\omega_{\ell} - f\left(\sum_{j=1}^{\ell-1}\psi(\omega_{j}) + 2\right) \geq \frac{1}{\alpha}f^*\left( \omega_{\ell}\right).
\end{align}
\end{subequations}

Now we show that if $ \tilde{\psi}(p) = \psi(p) $ for $ p\in [p_{\min}, \omega_{\ell}) $, and $ \tilde{\psi}(p) = 1 $ for $ p = \omega_{\ell} $, then we can construct a feasible set of non-zero points for $ \tilde{\psi}(p) $ over the remaining interval $ [\omega_{\ell}, p_{\max}] $. We first show that selecting one buyer with price $ \omega_{\ell} $ is indeed feasible:
\begin{align*}
\ & \sum_{j=1}^{\ell-1}\omega_{j}\tilde{\psi}(\omega_{j}) + \omega_{\ell} - f\left(\sum_{j=1}^{\ell-1}\tilde{\psi}(\omega_{j}) + 1\right) \\
= \ & \sum_{j=1}^{\ell-1}\omega_{j}\tilde{\psi}(\omega_{j}) + \frac{1}{2}\left( 2\omega_{\ell} - 2f\Big(\sum_{j=1}^{\ell-1}\tilde{\psi}(\omega_{j}) + 1\Big) \right)\\
\geq\ & \sum_{j=1}^{\ell-1}\omega_{j}\psi(\omega_{j}) + \frac{1}{2}\left( 2\omega_{\ell} - f\Big(\sum_{j=1}^{\ell-1} \psi(\omega_{j}) + 2 \Big) - f\Big(\sum_{j=1}^{\ell-1} \psi(\omega_{j})\Big) \right)\\
\geq  \ & \sum_{j=1}^{\ell-1}\omega_{j}\psi(\omega_{j}) + \frac{1}{\alpha}f^*\left( \omega_{\ell}  \right) -  \sum_{j=1}^{\ell-1}\omega_{j} \psi(\omega_{j})  \\
=\ & \frac{1}{\alpha}f^*\left( \omega_{\ell}  \right).
\end{align*}
where the first inequality is based on the convexity of $ f $, and the second inequality is a simple manipulation based on Eq. \eqref{eq_feasibility_before_after}.

The above discussion implies that $ \tilde{\psi}(p) $ satisfies Eq. \eqref{eq_necessary_inequality} when $ v $ is within interval $[p_{\min}, \omega_{\ell}] $. The construction of feasible non-zero points for $ \tilde{\psi}(p) $ over the remaining interval $ (\omega_{\ell}, p_{\max}] $ is similar to the proof of Claim \ref{claim_k}. The only difference is that, in Claim \ref{claim_k}, the interval starts from $ p_{\min} $ and we try to push $ \psi(p_{\min}) = \tau^{(\alpha)} +2 $ to $ \tilde{\psi}(p_{\min}) = \tau^{(\alpha)} + 1 $, while here the interval starts from $ \omega_{\ell} $ and we try to push $ \psi(\omega_{\ell}) = 2 $ to $ \tilde{\psi}(\omega_{\ell}) = 1 $. We omit the details for brevity.

In summary, if there exists an $ \alpha $-feasible selection function $ \psi  $ with $ \psi(\omega_{\ell}) > 1 $ for some non-zero point $ \omega_{\ell}\in \Omega\backslash\{p_{\min}\} $, then there exists another $ \alpha $-feasible selection function $ \tilde{\psi} $  which is the same as $  \psi $ for  $ p\in [p_{\min},\omega_{\ell}) $ and satisfies $ \tilde{\psi}(\omega_{\ell}) = 1 $. We thus complete the proof of Claim \ref{claim_two_one}.

\section{Proof of Claim \ref{claim_equality}}
\label{proof_of_claim_equality}
Claim \ref{claim_equality} largely follows our proof of Claim \ref{claim_k} in Appendix \ref{proof_of_claim_k}. Here we briefly explain the intuitions. If $ \psi $ is $ \alpha $-feasible and Eq. \eqref{eq_larger_than} holds for some $ \omega_{\ell} $, then, similar to Eq. \eqref{eq_existence_of_omega_2}, there must exist a unique $ \tilde{\omega}_{\ell} > \omega_\ell $ so that:
\begin{align*}
\sum_{j=1}^{\ell-1} \omega_{j} \psi(\omega_{j})  - f\Big(\sum_{j=1}^{\ell-1} \psi(\omega_{j})\Big) =  \frac{1}{\alpha}f^*(\tilde{\omega}_{\ell}) > \frac{1}{\alpha}f^*(\omega_{\ell}).
\end{align*}
Here, the unique existence of $ \tilde{\omega}_{\ell} $ is because of the strict monotonicity of the conjugate function $ f^* $. Based on $ \tilde{\omega}_{\ell} $ and our proof of Claim \ref{claim_k} in Appendix \ref{proof_of_claim_k}, we can prove in the similar way that a new $ \alpha $-feasible selection function $ \tilde{\psi} $ can be constructed.  Meanwhile, it is easy to see that $ \tilde{\omega}_{\ell}  $ is always larger than $ \omega_{\ell} $, and it could even be larger than $ p_{\max} $ (e.g., the last sub-case of \textbf{Case-2} in Appendix \ref{proof_of_claim_k}). Intuitively, if $ \omega_{\ell} > p_{\max}  $, then $ \tilde{\omega}_{\ell}  $ is not a feasible non-zero point of $ \tilde{\psi} $. Otherwise,  $ \tilde{\omega}_{\ell}  $ is the $ \ell $-th non-zero point of $ \tilde{\psi} $. We thus complete the proof of Claim \ref{claim_equality}.

\section{Proof of Theorem \ref{theorem_optimality_conditions}}
\label{proof_of_optimality_conditions}

To prove Theorem \ref{theorem_optimality_conditions}, it suffices to prove the following two optimality conditions:
\begin{itemize}
\item \textbf{OPT-Condition-1}:  if $ \alpha_* $ is the optimal competitive ratio, then the non-zero points satisfy
\begin{align}\label{eq_opt_condition_one}
	\sum_{\ell=1}^{L^{(\alpha_*)}} \omega_{\ell}^{(\alpha_*)}\cdot \psi^{(\alpha_*)}\Big(\omega_{\ell}^{(\alpha_*)}\Big) - f\Big(\tau^{(\alpha_*)} + L^{(\alpha_*)}\Big) = \frac{1}{\alpha_*}f^*(p_{\max}).
\end{align}

\item \textbf{OPT-Condition-2}: if $ \alpha_* $ is the optimal competitive ratio, then the characteristic selection function $ \psi^{(\alpha_*)} $ has $ \bar{k} - \tau^{(\alpha_*)} $ non-zero points, i.e., $ L^{(\alpha_*)} = \bar{k} - \tau^{(\alpha_*)} $.
\end{itemize}
We emphasize that the above two optimality conditions should be satisfied simultaneously in order to claim that $ \alpha_* $ is the optimal competitive ratio. In what follows we prove the first condition in detail, and give a brief proof for the second condition due to its similarity.
%In fact, our following proof shows that if any of these two optimality conditions is satisfied, then the other one is automatically satisfied.

\textbf{Proof of OPT-Condition-1}. This condition can be proved by contradiction. The key idea is that, if the competitive ratio parameter $ \alpha_* $ is optimal and its characteristic selection function $ \psi^{(\alpha_*)} $ does not satisfy Eq. \eqref{eq_opt_condition_one}, then we can show that a smaller competitive ratio parameter $ \alpha < \alpha_* $ has an $ \alpha $-feasible characteristic selection function $ \psi^{(\alpha)} $, leading to contradiction that $ \alpha_* $ is the optimal competitive ratio of all deterministic online algorithms. 

To be more specific, if  $ \alpha_* $ is such that Eq. \eqref{eq_opt_condition_one} does not hold, then
\begin{align*}
\sum_{\ell=1}^{L^{(\alpha_*)}} \omega_{\tau}^{(\alpha_*)}\cdot \psi^{(\alpha_*)}\big(\omega_{\tau}^{(\alpha_*)}\big) - f\big(\tau^{(\alpha_*)} + L^{(\alpha_*)}\big) > \frac{1}{\alpha_*}f^*(p_{\max}).
\end{align*}
Let us consider $ \alpha  = \alpha_* - \epsilon $ and $ \epsilon $ is a very small positive real number. We assume that $ \epsilon $ is very small so that $ \tau^{(\alpha)} = \tau^{(\alpha_*)} $ and $ L^{(\alpha)} = L^{(\alpha_*)} $. Note that such $ \epsilon $ always exists based on the definitions of $ \tau^{(\alpha)} $ and $ L^{(\alpha)} $. We next show that it is possible to construct another characteristic selection function $ \psi^{(\alpha)} $ that shifts the non-zero points  of $ \psi^{(\alpha_*)} $ slightly towards the left (except the first non-zero point, as $ \omega_1^{(\alpha_*)} = \omega_1^{(\alpha)} = p_{\min} $ always holds). For instance, the new second non-zero point $ \omega_2^{(\alpha)} $ will be smaller than $ \omega_2^{(\alpha_*)} $:
\begin{align*}
f^*(\omega_2^{(\alpha)}) = (\alpha_* - \epsilon)g(\tau^{(\alpha_*)} + 1)  < \alpha_* g(\tau^{(\alpha_*)} + 1) 
=\ & f^*(\omega_2^{(\alpha_*)}).
\end{align*}
Similarly, we can show that $ \omega_{\ell}^{\alpha} < \omega_{\ell}^{\alpha_*} $ holds for all $ \ell = \{2,3,\cdots, L^{(\alpha_*)}\} $, meaning that a new characteristic selection function $ \psi^{(\alpha_*-\epsilon)} $ can be constructed so that
\begin{align*}
& \sum_{\ell=1}^{L^{(\alpha_*)}} \omega_{\ell}^{(\alpha_*)}\cdot \psi^{(\alpha_*)}\big(\omega_{\ell}^{(\alpha_*)}\big) - f\big(\tau^{(\alpha_*)} + L^{(\alpha_*)}\big)\\
>\ & \sum_{\ell=1}^{L^{(\alpha_*)}} \omega_{\ell}^{(\alpha_*-\epsilon)}\cdot \psi^{(\alpha_*-\epsilon)}\big(\omega_{\ell}^{(\alpha_*-\epsilon)}\big) - f\big(\tau^{(\alpha_*)} + L^{(\alpha_*)}\big) \\
\geq\ &  \frac{1}{\alpha_*-\epsilon} f^*(p_{\max}) \\
>\    &  \frac{1}{\alpha_*}f^*(p_{\max}).
\end{align*}
Thus,  $ \psi^{(\alpha)} $ is an $ \alpha $-feasible characteristic selection function with $ \alpha = \alpha_* - \epsilon $, meaning that our threshold policy $ \textsf{TOS}_{\boldsymbol{\lambda}} $ can be $ (\alpha_* - \epsilon) $-competitive as long as the threshold $ \boldsymbol{\lambda} $ is designed based on Theorem \ref{theorem_optimality} with $ \textsf{CR}_f^*(\rho,k) = \alpha_* - \epsilon  $.  This contradicts with the assumption that $ \alpha_* $ is the optimal competitive ratio of all deterministic online algorithms. We thus complete the proof of \textbf{OPT-Condition-1}.

\textbf{Proof of OPT-Condition-2}. The proof of the second optimality condition is similar to the first one. In particular, we can show that whenever an optimal competitive ratio $ \alpha_* $ is such that $ L^{(\alpha_*)} < \bar{k} - \tau^{(\alpha_*)} $, there must exist a smaller competitive ratio parameter $ \alpha < \alpha_* $ so that a new $ \alpha $-feasible characteristic selection function $ \psi^{(\alpha)} $ can be constructed and the number of non-zero points of $ L^{(\alpha)}  $ is within $   [L^{(\alpha_*)}, \bar{k} - \tau^{(\alpha)}] $. This contradicts with the assumption that $ \alpha_* $ is the optimal competitive ratio. The details are omitted for brevity.

%that there must exist an $ \epsilon $ so that
%\begin{align*}
%& \sum_{\tau=1}^{\ell-1} \omega_{\tau}^{(\alpha_*-\epsilon)}\cdot \psi^{(\alpha_*-\epsilon)}\big(\omega_{\tau}^{(\alpha_*-\epsilon)}\big) - f\Big(\tau^{(\alpha_*)} + \ell - 1\Big) = \frac{1}{\alpha_*-\epsilon}f^*\big(\omega_{\ell}^{(\alpha_*-\epsilon)}\big), \quad \forall \ell \in \{2, 3, \cdots, L^{(\alpha)}\},\\
%& \sum_{\ell=1}^{L^{(\alpha_*)}} \omega_{\tau}^{(\alpha_*-\epsilon)}\cdot \psi^{(\alpha_*-\epsilon)}\big(\omega_{\tau}^{(\alpha_*-\epsilon)}\big) - f\Big(\tau^{(\alpha_*)} + L^{(\alpha_*)}\Big) = \frac{1}{\alpha_*-\epsilon} f^*(p_{\max}) >  \frac{1}{\alpha_*}f^*(p_{\max}).
%\end{align*}
%In this case, the non-zero points $ \omega_{\ell}^{(\alpha_*-\epsilon)} < \omega_{\ell}^{(\alpha_*)} $, meaning that the non-zero points all shift towards left except the first one.

Based on the above two optimality conditions, Theorem \ref{theorem_optimality_conditions} follows.

\section{Online Primal-Dual Analysis}
\label{appendix_OPD}
In this section, we briefly explain how our threshold policy is related to the online primal-dual approach. 
%To help derive the primal and dual of Problem \eqref{SWM}, we define $ \bar{f} $ by
%\begin{align*}
%\bar{f}(y) = 
%\begin{cases}
%	f(y)    & \text{ if } y\in [0,k],\\
%	+\infty & \text{ if } y >  k,
%\end{cases}
%\end{align*}
%which is equivalent to $ f $ but with an extended domain. 
Based on statement of OSCC in Section \ref{section_OSCC_statement}, the relaxed version of OSCC in the offline setting can be written as follows:
\begin{subequations}
\begin{alignat*}{3}
	(\textbf{Primal}): \qquad 
	& \underset{\{x_t \geq 0\}_{\forall t}, y}{\text{maximize}}\qquad  & & \sum_{t=1}^T p_t x_t - f(y) \\
	& \text{subject to} & &  y \geq  \sum_{t=1}^T x_t,  \quad &  & (\lambda)\\ 
	& & &  x_t \leq 1,   \forall t\in[T],   \quad & & (u_t)
\end{alignat*}
\end{subequations}
where $ \lambda $ and $ u_t $ denote the dual variables of the corresponding constraints.  The above primal formulation is equivalent to the initial definition of OSCC except the relaxation of $ \{x_t\}_{\forall t} $. The dual of the above primal problem can be written as 
\begin{subequations}
\begin{alignat*}{3}
	(\textbf{Dual}): \qquad  
	& \underset{\lambda, \{u_t\}_{\forall t}}{\text{minimize}} \qquad  & & \sum_{t=1}^T u_t  + f^*(\lambda) \\
	& \text{subject to} & &  u_t \geq p_t - \lambda,  \\ 
	& & &  \lambda, u_t \geq 0,   \forall t\in[T],   
\end{alignat*}
\end{subequations}
where $ f^*(\lambda) $ is given by
\begin{align*}
f^*(\lambda) = \max_{ i \in \mathbb{Z}_{\geq 0} }\ \lambda i    - f (i) = \max_{ i \in \{0,1,\cdots, k\} }\ \lambda i   - f(i), \quad   \lambda \in [0,+\infty).
\end{align*}
It is worth pointing out that this is exactly the conjugate $ f^* $ we defined in Eq. \eqref{eq_f_star}.

\begin{lemma}[Weak Duality]\label{theorem_weak_duality}
The optimal objective value of the \textbf{Dual} program, denoted by $ D^* $, satisfies $ D^* \geq \textsf{OPT} $, where $ \textsf{OPT} $ denotes the optimal profit achieved in the offline setting.
\end{lemma}

%Here, the subscript ``{\textsf{fc}}" means ``Fenchel conjugate". Note that the Fenchel conjugate $ f_{\textsf{fc}}^* $ is a continuous version of our discrete conjugate function $ f^* $ defined in Eq. \eqref{eq_f_star}. Moreover, it is obvious that $ f_{\textsf{fc}}^* $ is lower bounded by $ f^* $:
%\begin{align*}
%  f_{\textsf{fc}}^*(p) \geq f^*(p), \quad p\in [p_{\min}, p_{\max}]. 
%\end{align*}

\textbf{Principle of Online Primal-Dual Approach}. The core idea of the online primal-dual approach is to construct a dual feasible objective
that is close to the primal feasible objective at each stage when there is a new arrival of buyer, namely, construction of $ \{x_t, y^{(t)} \} $, and $ \{\lambda^{(t)}, u_t\} $ after processing each buyer $ t\in [T] $, and then
use weak duality (i.e., Lemma \ref{theorem_weak_duality}) to bound the performance of the online algorithm.  Note that here we denote by $ y^{(t)} $ and $ \lambda^{(t)} $ the primal and dual variables $ y $ and $ \lambda $ at each stage, respectively. This indicates that we need to design a \textit{trajectory} of feasible $ y $'s and $ \lambda $'s.  Specifically, let us denote the change of the primal objective by $ \Delta_P $ after processing a buyer (we call it one stage of change), and denote by $ \Delta_D $ the change of the dual objective. An online algorithm $ \textsf{ALG} $ is $ \alpha $-competitive if 
\vspace{-0.3cm}
\begin{equation*}
\Delta_P  \geq \frac{1}{\alpha}\Delta_D
\end{equation*}
holds at each stage. This is because, after processing all the buyers available, the profit achieved by the online algorithm $ \textsf{ALG} $ satisfies
\begin{align*}
\textsf{ALG} = \sum \Delta_P \geq \frac{1}{\alpha} \sum \Delta_D \geq \frac{1}{\alpha} D^* \overset{(\text{Lemma \ref{theorem_weak_duality}})}{\geq} \frac{1}{\alpha} \textsf{OPT},
\end{align*}
which thus implies that at least $ 1/\alpha $ fraction of optimum will be achieved. Note that $ \textsf{ALG} = \sum \Delta_P $ is because of the telescoping sum of $ \Delta_P $ over all stages.
%where $ D^* $ (resp. $ P^*$) denotes the optimal dual (resp. primal) objective.
%and $ \textsf{OPT} $ denotes the optimal welfare achieved in the offline setting.

\textbf{Relationship to Our Threshold Policies}. The design of $ \textsf{TOS}_{\boldsymbol{\lambda}} $, as well as its competitive analysis, is primarily inspired by the principle described above, except that the design and analysis of the admission threshold at its horizontal segment (i.e, $ \lambda_0 = \lambda_1 =\cdots = \lambda_\tau = p_{\min} $) and the boundary part (i.e, $ \lambda_{\bar{k}} = p_{\max} $) are derived differently based on the worst-case analysis. We  next briefly explain the online primal-dual flavor of our threshold policy.

If buyer $ t\in [T] $ is rejected, then $ \Delta_P = \Delta_D = 0 $, and thus $ \Delta_P  \geq \frac{1}{\alpha}\Delta_D $ holds; If buyer $ t $ is accepted, and suppose that this buyer is the $ i $-th buyer that has been selected, then to guarantee $ \Delta_P  \geq \frac{1}{\alpha}\Delta_D $, we have to design the threshold so that
\begin{align*}
\Delta_P = p_t - \Big( f(i) - f(i-1) \Big) \geq \frac{1}{\alpha} \Delta_D = \frac{1}{\alpha} \Big( u_t + f^*(\lambda_{i}) - f^*(\lambda_{i-1}) \Big)
\end{align*}
holds for all possible arrival instances, namely, all possible $ p_t\in [p_{\min},p_{\max}] $. Here, the threshold $ \{\lambda_i\}_{\forall i} $ is exactly the trajectory of $ \{\lambda^{(t)}\}_{\forall t} $ as mentioned above. If we designed the dual variables in a way such that $ u_t = p_t - \lambda_{i-1} \geq 0 $ (note that this is a feasible design), then the following inequality must hold in order to guarantee $ \Delta_P  \geq \frac{1}{\alpha}\Delta_D $:
\begin{align*}
p_t - c_i = \lambda_{i-1} + u_t - c_i \geq \frac{1}{\alpha}\Big( u_t + f^*(\lambda_{i}) - f^*(\lambda_{i-1}) \Big).
\end{align*}
Since $ u_t\geq 0 $, it suffices to guarantee that the following inequality holds:
\begin{align*}
\lambda_{i-1} -  c_i  \geq \frac{1}{\alpha}\Big( f^*(\lambda_{i}) - f^*(\lambda_{i-1}) \Big),
\end{align*}
which is exactly one of the key sufficient inequalities given by Lemma \ref{theorem_OPD_sufficiency_proof} (i.e., Eq. \eqref{eq_OPD_inequalities_proof_2}). Recall that an admission threshold $ \boldsymbol{\lambda} $ satisfying the $ \alpha $-parameterized inequalities in Lemma \ref{theorem_OPD_sufficiency_proof} implies that the sufficient inequalities in Corollary \ref{theorem_OPD_inequality} hold. In this regard, Lemma \ref{theorem_OPD_sufficiency_proof} follows the online primal-dual framework and Corollary \ref{theorem_OPD_inequality} follows Lemma \ref{theorem_OPD_sufficiency_proof}.

%\begin{align*}
%\underset{x_t\in \{0,1\}, \forall t\in [T]}{\textsf{maximize}}\quad    \sum_{t=1}^T v_t x_t - f\left(\sum_{t=1}^T x_t\right)\quad  \textsf{subject to}\quad \sum_{t=1}^T x_t \leq M.
%\end{align*}

\end{document}